%% file: main.tex
\documentclass[11pt]{article}
\usepackage[letterpaper,margin=1in]{geometry}

\usepackage[utf8]{inputenc} 
\usepackage[T1]{fontenc}    
\usepackage{url}            
\usepackage{booktabs}       
\usepackage{amsfonts}       
\usepackage{nicefrac}       
\usepackage{microtype}      
\usepackage{xcolor}         

\usepackage{nicefrac,bm,bbm}
\usepackage{amsmath,amsthm,color,colortbl,amssymb}

\usepackage{hyperref}

\usepackage{enumitem}
\usepackage{caption}
\usepackage{natbib}

\usepackage{wrapfig}
\usepackage{comment}

\usepackage{graphicx}
\usepackage{multirow}
\usepackage{xspace}
\usepackage{mathtools}
\usepackage{euscript}
\usepackage{thm-restate}
\usepackage{enumitem}
\usepackage{tikz}
\usepackage{subcaption}
\usepackage{tablefootnote,tabularx}
\usepackage[T1]{fontenc}    

\usepackage{pifont}

\usepackage{euscript}
\usepackage{mleftright}

\usepackage[linesnumbered,ruled,lined,noend]{algorithm2e}
\makeatletter
\patchcmd\algocf@Vline{\vrule}{\vrule \kern-0.4pt}{}{}
\patchcmd\algocf@Vsline{\vrule}{\vrule \kern-0.4pt}{}{}
\makeatother
\SetKwComment{Hline}{}{\vspace{-3mm}\textcolor{gray}{\hrule}\vspace{1mm}}
\definecolor{darkgrey}{gray}{0.3}
\definecolor{commentcolor}{gray}{0.5}
\SetKwComment{Comment}{\color{commentcolor} $\triangleright$\ }{}
\SetCommentSty{}
\SetNlSty{}{\color{darkgrey}}{}
\SetKwProg{Fn}{function}{}{}
\SetKwProg{Subr}{subroutine}{}{}

\usepackage[capitalize]{cleveref}

\crefalias{AlgoLine}{line}%
\crefname{algocf}{Algorithm}{Algorithms}

\usepackage{cancel}
\usepackage{newpxtext}
\usepackage{newpxmath}
\newcommand{\ie}{{\em i.e.,~\xspace}}
\newcommand{\eg}{{\em e.g.,~\xspace}}

\newcommand{\cA} {\ensuremath{\EuScript{A}}}

\newcommand{\cAG} {\ensuremath{\EuScript{A}}_{\textsc{G}}}
\newcommand{\cAP} {\ensuremath{\EuScript{A}}_{\textsc{P}}}

\newcommand{\cB} {\ensuremath{\mathcal{B}}}

\newcommand{\cC} {\ensuremath{\mathcal{C}}}
\newcommand{\cD} {\ensuremath{\mathcal{D}}}
\newcommand{\cE} {\ensuremath{\mathcal{E}}}
\newcommand{\cF} {\ensuremath{\mathcal{F}}}
\newcommand{\cG} {\ensuremath{\mathcal{G}}}
\newcommand{\cH} {\ensuremath{\mathcal{H}}}

\newcommand{\cM} {\ensuremath{\mathcal{M}}}
\newcommand{\cN} {\ensuremath{\mathcal{N}}}
\newcommand{\cO} {\ensuremath{\mathcal{O}}}

\newcommand{\cS} {\ensuremath{\mathcal{S}}}
\newcommand{\cX} {\ensuremath{\mathcal{X}}}

\newcommand{\cW} {\ensuremath{\mathcal{W}}}

\DeclareMathOperator*{\argmax}{arg\,max}

\newcommand{\E}[1]{\mathbb{E}\mleft[#1\mright]}
\newcommand{\subE}[2]{\mathop{\mathbb{E}}\limits_{#1}\mleft[#2\mright]}

\newcommand{\subP}[2]{\mathbb{P}_{#1}\left(#2\right)}

\newcommand{\vx}{\bm{x}}
\newcommand{\eps}{\varepsilon}
\newcommand{\ind}[1]{\mathbb{I}{\left\{#1\right\}}}

\newcommand{\gft}{\textsc{GFT}}

\newcommand{\prof}{\textsc{Profit}}

\newcommand{\term}[1]{\ensuremath{\texttt{#1}}\xspace}
\def \OPT {\term{OPT}}

\newcommand{\erevmax}{\textsc{Profit-Max}\xspace}

\newcommand{\egftmax}{\textsc{GFT-Max}\xspace}

\newcommand{\gftest}{\textsc{GFT-Est}\xspace}
\newcommand{\hedge}{\textsc{Hedge}\xspace}
\newcommand{\blockdecomposition}{\textsc{Block-Decomposition}\xspace}

\definecolor{mygreen}{rgb}{0.0, 0.5, 0.0}
\definecolor{myorange}{rgb}{0.55, 0.62, 1}

\newcommand{\initOneLiners}{%
 	\setlength{\itemsep}{0pt}
	\setlength{\parsep }{0pt}
  	\setlength{\topsep }{0pt}     	
}
\newenvironment{OneLiners}[1][\ensuremath{\bullet}]
    {\begin{list}
        {#1}
        {\initOneLiners}}
    {\end{list}}

\newcommand{\xhdr}[1]{\vspace{2mm} \noindent{\bf #1}}

\newcommand{\defeq}{=}

\newcommand{\supp}{\textnormal{supp}}

\newcommand{\bull}{\hspace*{2mm}$\bullet$\hspace*{1mm}}

\newcommand{\bmid}{\hspace*{2mm}\Bigg\vert\hspace*{2mm}}

\newcommand{\egft}{\widehat\gft}
\newcommand{\bgft}{\hat{\bm{r}}}
\newcommand{\cT}{\mathcal{T}}
\theoremstyle{plain}
\newtheorem{theorem}{Theorem}[section]
\newtheorem{proposition}[theorem]{Proposition}
\newtheorem{claim}[theorem]{Claim}
\newtheorem{lemma}[theorem]{Lemma}
\newtheorem{corollary}[theorem]{Corollary}
\theoremstyle{definition}
\newtheorem{definition}[theorem]{Definition}

\theoremstyle{remark}

\definecolor{niceRed}{RGB}{190,38,38}
\definecolor{Red2}{RGB}{219, 50, 54}
\definecolor{mgreen}{HTML}{9ECA8D}
\definecolor{blueGrotto}{HTML}{059DC0}
\definecolor{limeGreen}{HTML}{81B622}
\definecolor{myellow}{rgb}{0.88,0.61,0.14}
\definecolor{navyBlueP}{HTML}{315794}
\definecolor{Sepia}{HTML}{7F462C}
\definecolor{red2}{HTML}{1F462C}
\definecolor{orange2}{HTML}{FF8000}
\definecolor{mgray}{HTML}{ABB3B8}
\definecolor{myPurple}{RGB}{175,0,124}
\definecolor{royalBlue}{HTML}{057DCD}

\hypersetup{
	colorlinks = true,
	linkcolor = niceRed,
	citecolor = royalBlue,
	linktocpage = true,
	urlcolor = mpink
}

\title{No-Regret Learning in Bilateral Trade via Global Budget Balance\thanks{This is the full version of \citet{longVersion}.}}

\author{
    Martino Bernasconi$^\dagger$ \quad
    Matteo Castiglioni$^\ddagger$ \quad
    Andrea Celli$^\dagger$ \quad
    Federico Fusco$^\#$  \vspace{6mm}\\
    $^\dagger$\ Bocconi university\\
    $^\ddagger$\ Politecnico di Milano\\
    $^\#$\ Sapienza University of Rome\vspace{2mm}\\
    {\textcolor{black}{\small\texttt{\{martino.bernasconi,andrea.celli2\}@unibocconi.it}, \quad \texttt{matteo.castiglioni@polimi.it,}}}\\
    {\textcolor{black}{\small\texttt{federico.fusco@uniroma1.it}}}
}
\date{}
\begin{document}

\maketitle
\thispagestyle{empty}

\begin{abstract}
    Bilateral trade models the problem of intermediating between two rational agents --- a seller and a buyer --- both characterized by a private valuation for an item they want to trade. 
    We study the online learning version of the problem, in which at each time step a new seller and buyer arrive and the learner has to set prices for them without any knowledge about their (adversarially generated) valuations.

    In this setting, known impossibility results rule out the existence of no-regret algorithms when budget balanced has to be enforced at each time step. In this paper, we introduce the notion of \emph{global budget balance}, which only requires the learner to fulfill budget balance over the entire time horizon. Under this natural relaxation, we provide the first no-regret algorithms for adversarial bilateral trade under various feedback models. First, we show that in the full-feedback model, the learner can guarantee $\tilde O(\sqrt{T})$ regret against the best fixed prices in hindsight, and that this bound is optimal up to poly-logarithmic terms. Second, we provide a learning algorithm guaranteeing a $\tilde O(T^{\nicefrac 34})$ regret upper bound with one-bit feedback, which we complement with a $\Omega(T^{\nicefrac 57})$ lower bound that holds even in the two-bit feedback model. Finally, we introduce and analyze an alternative benchmark that is provably stronger than the best fixed prices in hindsight and is inspired by the literature on bandits with knapsacks.
\end{abstract}

\clearpage

\tableofcontents

\clearpage

\pagenumbering{arabic}

\section{Introduction}

    Bilateral trade is a classic economic problem where two agents --- a seller and a buyer --- are interested in trading a good. Both agents are characterized by a private valuation for the item, and their goal is to maximize their own utility. Solving this problem requires the design of a mechanism that intermediates between the two parties, facilitating the trade.
    Ideally, the mechanism should maximize efficiency (\ie trade whenever the buyer's valuation exceeds the seller's one) while ensuring that agents behave according to their true preferences (\emph{incentive compatibility}), and that the utility for participating in the mechanism of each agent is non-negative (\emph{individual rationality}). These properties ensure favorable outcomes for the agents, yet they do not guarantee the economic viability of the mechanism. To see this, consider the following mechanism $\cM$.  $\cM$ asks the agents for their valuations, $s$ for the seller and $b$ for the buyer, and makes the trade happen if it is convenient (\ie if $s \le b$). In case of a trade, $\cM$ then charges $s$ to the buyer and pays $b$ to the buyer. It is not hard to see that $\cM$ enforces incentive compatibility and individual rationality, and is efficient by design. However, it exhibits the major drawback of allowing the intermediary to incur a net loss when $b>s$. To avoid such situations, a crucial constraint in bilateral trade is \emph{budget balance}, which restricts the mechanism from subsidizing the agents.

    As highlighted by the above example, an incentive compatible mechanism maximizing efficiency for bilateral trade may not be budget balanced. This phenomenon was first observed by \citet{Vickrey61}; subsequently \citet{MyersonS83}, provided a more general impossibility result by showing the existence of instances where a fully efficient mechanism that satisfies incentive compatibility, individual rationality, and budget balance does not exist. This result holds even when probabilistic information on the agents' valuations is available. To circumvent these impossibility results, the extensive subsequent research primarily focuses on finding approximately efficient mechanisms in the Bayesian setting. There, various incentive compatible mechanisms exist that give a constant-factor approximation to the social welfare (see, \eg \citet{BlumrosenD14,kang22fixed}, while more recent works also consider the harder problem of approximating the gain from trade \citep{Mcafee08,BlumrosenM16,brustle2017approximating,DengMSW21,fei2022improved}. While the Bayesian assumption of having perfect knowledge about the underlying distributions of valuations is, in some sense, necessary for extracting meaningful approximations to the social welfare \citep{Duetting20}, it is important to observe that this assumption is oftentimes unrealistic.  
    
    Following the recent line of work initiated by \citet{CesaBianchiCCF21}, we study this fundamental mechanism design problem through the lens of regret minimization in a repeated setting where at each time $t$, a new seller/buyer pair arrives. The seller arriving at time $t$ has a private valuation $s_t$ representing the lowest price they are willing to accept for the item. Analogously, the buyer has a private valuation $b_t$ representing the highest price they are willing to pay for the item. The learner, without any knowledge about the private valuations at the current time $t$, posts two (possibly randomized) prices: $p_t$ to the seller and $q_t$ to the buyer. A trade happens when both agents agree to trade, i.e., when $s_t \le p_t$ and $ q_t \le b_t$. After posting $(p_t,q_t)$, the learner observes some feedback about the transaction, and is awarded the \emph{gain from trade}:
    \[
        \gft_t(p,q) \defeq \ind{s_t \le p} \ind{q\le b_t} (b_t - s_t).
    \]
    
    The goal of the learner is to maximize the overall gain from trade or, equivalently, minimize the regret with respect to the best price in hindsight. Prior research has investigated the impact of different budget balance notions on the problem's learnability. When the mechanism is constrained to enforce per-round \emph{strong budget balance} (\ie $p_t = q_t$ at each time step $t$), it is possible to attain sublinear regret only when the sequence of valuations is drawn i.i.d. from some fixed unknown distribution, and the learner has either full feedback, or some stringent assumptions regarding the sequence of valuations are enforced. Specifically, in partial feedback regime, valuations have to be drawn i.i.d. from a smooth distribution, independently for the seller and the buyer \citep{CesaBianchiCCF21,cesa23mor}. If the learner is only required to enforce (step-wise) \emph{weak budget balance} (\ie $p_t \le q_t$ for each $t$), then \citet{AzarFF22} provide a learning algorithm 
    achieving sublinear $2$-regret when the sequence of valuation is generated by an oblivious adversary.\footnote{The $\alpha$-regret measures the difference between the gain from trade of the best fixed price in hindsight and $\alpha$ times that of the algorithm (see e.g., \citet{KakadeKL09}).} They also show that this result is tight: no algorithm can achieve sublinear $(2-\eps)$-regret in the adversarial case, for any constant $\eps>0$. In an attempt to overcome this barrier, \citet{CesaBianchiCCFL23} show that sublinear regret can be achieved beyond the i.i.d. stochastic setting, under the assumption that the adversary is constrained to choose randomized (possibly non-stationary) sequences of valuations that are not ``too concentrated'' (\ie under a $\sigma$-smooth adversary model).
    Inspired by the positive results obtained in the literature by transitioning from strong to weak budget balance, we investigate the following natural open question:
    \begin{center}
    \em Is it possible to achieve sublinear regret against an oblivious adversary in the repeated bilateral trade problem under a realistic notion of budget balance?
    \end{center}
    We answer this question positively by introducing \emph{global budget balance}, where the learner is required to maintain budget balance only ``overall''. The idea behind global budget balance is to allow the learner to reinvest the profit gained in previous rounds (obtained by posting a lower price for the seller compared to the buyer), with the constraint that the learner cannot subsidize the market \emph{over the whole time horizon}. Formally, a learning algorithm that posts prices $(p_1, q_1), (p_2, q_2), \dots$ is global budget balanced if the following inequality holds almost surely: $\sum_{t=1}^T \prof_t(p_t,q_t) \ge 0$. The profit $\prof_t(p_t,q_t) \defeq \ind{s_t \le p_t} \ind{q_t\le b_t} (q_t - p_t)$ is non-negative when $p_t \le q_t$, and may drop below zero only by posting prices that are not step-wise budget balanced, \ie $p_t > q_t.$   
    We argue that this constraint is more realistic than the restrictive notions of per-round budget balance. For instance, in contexts like ride-hailing platforms (such as Uber and Lyft), the platform might opt to forego some short-term profit to enhance other metrics, like the overall welfare of the system.
    \begin{table*}[t!]
    \centering
        \scalebox{.876}{
        \begin{tabular}{m{4.1cm}m{2cm}m{1.8cm}m{4.3cm}m{4.2cm}}
         & \bf Type of\newline Adversary & {\bf Budget Balance} & {\bf Regret \newline Upper Bounds}& \bf Regret  \newline Lower Bounds\\
        \toprule
        \cite{CesaBianchiCCF21} &  stochastic\newline setting~${}$ & strong & \bull Full: $\tilde O(T^{\nicefrac12})$\newline \bull Partial: $\tilde O(T^{\nicefrac23})^*$  & \bull Full: $\Omega(T^{\nicefrac12})$\newline \bull Partial: $\Omega(T^{\nicefrac23})$ \\
        \midrule
        \cite{AzarFF22} &  adversarial\newline setting & weak & \bull Full: $\tilde O(T^{\nicefrac12})^\dagger$\newline \bull Partial: $\tilde O(T^{\nicefrac34})^\dagger$ & \bull Full: $\Omega(T^{\nicefrac12})^\dagger$ \newline \bull Partial: $\Omega(T^{\nicefrac23})^\dagger$\\
        \midrule
        \cite{CesaBianchiCCFL23} & $\sigma$-smooth\newline adversary & weak & \bull Full: $\tilde O(T^{\nicefrac12})$\newline \bull Partial: $\tilde O(T^{\nicefrac34})$ & \bull Full: $\Omega(T^{\nicefrac12})$\newline \bull Partial: $\Omega(T^{\nicefrac34})$\\
        \midrule
        \rowcolor{navyBlueP!20}\textbf{This paper}& adversarial\newline setting & \bf global & \bull Full: $\tilde O(T^{\nicefrac12})$\newline \bull Partial: $\tilde O(T^{\nicefrac34})$ & \bull Full: $\Omega(T^{\nicefrac12})$\newline \bull Partial: $\Omega(T^{\nicefrac57\approx 0.714})$\\ \bottomrule
        \end{tabular}
        }
        \caption{Comparison of prior results on bilateral trade. The positive result for a stochastic adversary in the partial feedback, marked with an asterisk ($*$), holds under the assumption that the seller and buyer valuations are drawn independently from smooth distributions. All the bounds in the second row (\citet{AzarFF22}), marked with a dagger ($\dagger$), apply to $2$-regret.}\label{tab:recap results}
    \end{table*}

\subsection{Overview of Our Results}

    We report here an overview of our results, we also refer to \Cref{tab:recap results} for a comparison with the state of the art. In this paper we introduce the notion of global budget balance for the repeated bilateral trade problem, and provide the following results in terms of regret with respect to the best fixed price in hindsight in the adversarial case:
    \begin{itemize}
        \item In the full feedback model, when the learner observes seller and buyer valuations after posting prices, we design a learning algorithm characterized by a $\tilde O(T^{\nicefrac12})$ regret upper bound (\Cref{thm:full_upper}). We also prove that no learning algorithm can improve this bound by more than a poly-$\log T$ factor (\Cref{thm:full_lower}). 
        \item In the \emph{one-bit feedback} model, where the learner can observe only whether the trade happened or not, we show that it is possible to guarantee a $\tilde O(T^{\nicefrac34})$ regret upper bound (\Cref{thm:partial_upper}). Then, we provide an $\Omega(T^{\frac57\approx 0.714})$ lower bound, which holds even in the \emph{two-bit feedback} model, where the learner can observe which agent accepted and who declined the offered prices (\Cref{thm:lower2bits}).
    \end{itemize}
    
    These results demonstrate how the notion of global budget balance enables online learnability, allowing us to provide the first no-regret algorithms for repeated bilateral trade within an oblivious adversary framework,
    in contrast to the per-round approaches considered in previous works.
    Furthermore, the regret rates separate full feedback and the two partial feedback models (one or two bits). In partial feedback, the surprising lower bound of $\Omega(T^{\nicefrac 57})$, together with the $O(T^{\nicefrac 34})$ upper bound, mark a clear separation between this problem and other partial feedback models (\eg partial monitoring \citep{BartokFPRS14} and online learning with feedback graph  \citep{AlonCGMMS17}, where the minimax regret have been characterized to fall in one of three admissible rates: $\sqrt T$, $T^{\nicefrac 23}$ and $T$). This separation had already been hinted at in the special case of $\sigma$-smooth adversary by \citet{CesaBianchiCCFL23}.
    
    Finally, inspired by work on bandits with knapsacks (see \Cref{sec:rel} for detailed references), we introduce a stronger learning benchmark: the best fixed feasible distribution over prices. Such benchmark is allowed to post prices that are not per-round budget balanced, but is global budget balanced in ``expectation''.  
    \begin{itemize}
        \item We show that there exists a constant $\varepsilon_0>0$ such that it is impossible to achieve sublinear $\alpha$-regret against this benchmark for any $\alpha\in [1,1+\varepsilon_0)$ (\Cref{thm:alpha_lower}).
        \item We prove that the best feasible distribution over prices collects at most twice the gain from trade extracted by the best fixed price in hindsight (\Cref{thm:2_upper}). This implies the existence of algorithms with sublinear $2$-regret against this new benchmark.
        \item We show that the multiplicative gap of $2$ between the gain from trade attainable by the two different benchmarks is tight (\Cref{thm:2_lower}).
        \end{itemize} 

        First, we observe that the task of learning the best feasible distribution over prices is reminiscent of the problem of bandits with knapsacks in the presence of replenishment \citep{kumar2022non,SlivkinsSF23,BernasconiCCF24ICLR}. In contrast to previous work, we consider the more challenging adversarial setting and provide learning algorithms with a competitive ratio that is an absolute constant.  
        In the adversarial bandits with knapsacks literature, the only setting where sublinear $\Theta(1)$-regret can be achieved is when the available budget is $\Omega(T)$ \citep{castiglioni22online}, while in general the competitive ratio is $O(\log T)$ \citep{immorlica2022adversarial}. Second, the tight multiplicative gap of $2$ between the two benchmarks suggests that  
        to design a better learning algorithm with sublinear $\alpha$-regret with respect to the best feasible distribution (for $\alpha \in (1+\eps_0,2)$), a more direct approach is needed.

    \subsection{Challenges and Techniques}

    The key aspects that distinguish bilateral trade from standard online learning models with full or bandit feedback can be identified in two main features: the action space and the challenging partial feedback structure. The applicability of previous results to our model is significantly limited due to adversarial input sequences and the need to handle the global budget balance constraint effectively. 
    
    \xhdr{Action space.} The action space is continuous and bidimensional (prices belong to $[0,1]^2$), and neither the gain from trade nor the profit functions are continuous in the prices posted. This makes it challenging to discretize the space with a finite grid $G$ such that the best prices in $G$ perform similarly to the best prices in $[0,1]^2$, and such that grid $G$ is small enough that it is possible to learn in an online way its best pair of prices.
    In the absence of any probabilistic or smoothness assumption on the adversary, we cannot rely on a ``smoothing trick'' to induce regularity on the expected gain from trade, as in previous works \citep{CesaBianchiCCFL23}. 
    
    \xhdr{Partial Feedback.} Partial feedback models for bilateral trade are inherently challenging. The one-bit feedback model only informs the learner on whether the trade happened or not, which is significantly less informative than the traditional bandit feedback model, since the learner cannot even reconstruct the gain from trade received for the specific prices it posted. For example, if the learner posts price $\nicefrac 12$ to both agents, and they accept the trade, there is no way of distinguishing between the case in which the gain from trade is constant (\eg valuations are $(0,1)$) from the case in which the gain from trade is arbitrarily small (\eg valuations are $(\nicefrac 12-\varepsilon,\nicefrac 12+\varepsilon)$ for some small $\varepsilon$). On the other hand, if one of the two agents rejects the trade, then the learner can only infer loose bounds on the valuations.

    \xhdr{Gain from Trade vs. Profit trade-off.} Global budget balance requires that the cumulative sum of profits at the end of the time horizon must be greater than or equal to $0$. Therefore, the learner has to maximize its cumulative gain from trade, while accumulating enough profit to enforce global budget balance. Balancing this trade-off is a complex task due to the different nature of the two objectives: gain from trade is maximized by setting identical prices for both agents, whereas profit is maximized by selecting prices that are ``far from each other''. To see this, consider an instance where valuations are either $(s_t,b_t)=(0,1)$ or $(s_t,b_t)=(\nicefrac 12-\varepsilon,\nicefrac 12+\varepsilon)$ with equal probability, for some small $\varepsilon>0$. To achieve maximum expected profit, the learner would always set the price at $0$ for the seller and $1$ for the buyer. On the other hand, to maximize the expected gain from trade, the learner would always offer $\nicefrac 12$ to both agents.

    \xhdr{Our Two-Phase Approach.} Our learning algorithms follow a two-phase approach, initially focusing on maximizing profit through a carefully designed multiplicative grid $F_K$ of candidate prices and then switching to maximizing gain from trade on a different (additive) grid  $H_K$ of non-budget-balanced prices. At a high level, the first phase is used to collect budget, which can be subsequently reinvested in the second phase. This poses several challenges due to the non-stationary nature of the adversary. The pairs of prices in $H_K$, which are not per-round budget balanced, enable the algorithm to circumvent the negative results that hinder discretization in scenarios with per-round budget balance (see, \eg, the ``needle in a haystack'' phenomenon in Theorem 7 of \citet{cesa23mor}). The multiplicative nature of the grid $F_K$ is crucial in ensuring that the gain from trade accrued by the algorithm during the first phase does not yield too much regret. This last result is surprising since, in the first phase, the learning algorithm is maximizing profit, an objective that is inherently orthogonal to the gain from trade.
    Finally, the scarcity of feedback in the one-bit feedback model is addressed via a carefully designed estimation technique that allows the learner to estimate the gain from trade in one point of the grid $H_K$ posting two different prices. In contrast to the technique by \citet{AzarFF22}, our procedure is ``asymmetric'' in how it deals with the seller and buyer, and it provides biased estimates.
    
    \xhdr{Lower bounds.} Besides the typical challenges in proving lower bounds for repeated bilateral trade with respect to the best fixed price in hindsight, in our model the agent is allowed to post prices that are not per-round budget balanced (\ie it may be the case that $p_t>q_t$). This considerably complicates the construction of the hard instances, as any algorithm could sacrifice temporarily some profit by posting prices with $p_t>q_t$ to extract a large gain from trade (that the fixed price benchmark may not be able to obtain). To deter this kind of behavior, we incorporate into the hard instances certain unfavorable trade opportunities that dissuade the learner from setting prices that are not budget balanced. This additional complication comes at some cost: in the partial (two-bit) feedback model we recover a lower bound of $\Omega(T^{\nicefrac57})$, whereas the corresponding lower bound by \citet{CesaBianchiCCFL23} is $\Omega(T^{\nicefrac34})$.
        
    \subsection{Further Related Works}\label{sec:rel}

        \xhdr{Online Learning and  Economics.} 
        Regret minimization techniques have found applications across different domains motivated by economics, with the goal of overcoming unrealistic assumptions.
        For example, they have been applied to one-sided pricing \citep{kleinberg2003value,feldman2016online}, auctions \citep{morgenstern2015pseudo,Cesa-BianchiGM15,LykourisST16,weed2016online,balseiro2019learning,nedelec2022learning,DaskalakisS22,Cesa-BianchiCRFL24}, contract design \citep{HoSV16,ZhuBYWYJ22,DuettinGSW23}, brokerage \citep{bolic2023online}, and Bayesian persuasion \citep{castiglioni2020online,zu2021learning,castiglioni2023regret,bernasconi2023optimal}.

        \xhdr{Partial feedback.} Repeated bilateral trade naturally involves challenges due to partial feedback. Therefore, our work aligns with the research that explores online learning with feedback models beyond the conventional full feedback and bandit models. Our one- and two-bit feedback models share similarities with \emph{graph-structured feedback} \citep{AlonCGMMS17} and with the \emph{partial monitoring} framework \citep{cesa2006regretpartial,BartokFPRS14}.   

        \xhdr{Bandits with knapsacks.}
        Another related line of work is that of online learning under long-term constraints. Some works study the case of static constraints and develop projection-free algorithms with sublinear regret and constraint violations \citep{mahdavi2012trading,jenatton2016adaptive}, while others study the case of time-varying constraints \citep{mannor2009online,yu2017online,sun2017safety}.
        \citet{badanidiyuru2018bandits} introduced and solved the (stochastic) bandits with knapsacks (BwK) framework, in which they consider bandit feedback and stochastic objective and cost functions. In this model, the learner's objective is to maximize utility while guaranteeing that, for each of the $m$ available resources, cumulative costs are below a certain budget $B$. Other optimal algorithms for stochastic BwK were proposed by \citet{agrawal2019bandits,immorlica2022adversarial}. 
        The setting with adversarial inputs was first studied in \citet{immorlica2022adversarial}, where the baseline considered is the best fixed distribution over arms. Achieving no-regret is not possible under this baseline and, therefore, they provide no-$\alpha$-regret guarantees for their algorithm. If we denote by $\rho$ the per-iteration budget of the learner, the best-known guarantees on the competitive ratio $\alpha$ are $\nicefrac 1\rho$ in the case in which $B=\Omega(T)$ \citep{castiglioni22online}, and $O(\log m\log T)$ in the general case \citep{kesselheim2020online}.
        When considering a benchmark similar to the adversarial BwK scenario, we show that our algorithm ensures a $\alpha=2$ guarantee.
        \citet{kumar2022non} recently proposed a generalization of the stochastic BwK model in which resource consumption can be non-monotonic; that is, resources can be replenished or renewed over time. Our model also admits replenishment.
        It should be noted that, in our setting, directly utilizing techniques from BwK is not feasible due to the complex continuous action space and the limited availability of feedback, which is less informative compared to traditional bandit feedback.

\section{Repeated Bilateral Trade} \label{sec:model}

    We study repeated bilateral trade problem in an online learning setting, where the learner has to enforce global budget balance and the sequence of valuations is generated by an oblivious adversary.

    \begin{algorithm}[t]
    \NoCaptionOfAlgo
    \caption{\textbf{Learning Protocol of Repeated Bilateral Trade}}\label{a:learning-model}
        \DontPrintSemicolon
        Initial budget $B_0 = 0$\;
        \For{$t=1,2,\ldots$}{
             The adversary privately chooses $(s_t,b_t)$ in $[0,1]^2$\;
             The learner posts prices $(p_t,q_t) \in [0,1]^2$ such that $p_t - q_t \le B_t$\;
             The learner receives a (hidden) reward $ \gft_t(p_t, q_t)\in [-1,1]$\;
             The budget of the learner is updated $B_t \gets B_{t-1} + \prof_t(p_t,q_t)$\;
             Feedback $z_t$ is revealed to the learner\;
        }
    \end{algorithm} 
    
    \xhdr{The learning protocol.} The learner repeatedly interacts with the environment according to the following protocol (see also pseudocode). At each time step $t$, a new pair of buyer and seller arrives, characterized by valuations $b_t \in [0,1]$ and $s_t\in [0,1]$, respectively. Without knowing $s_t$ and $b_t$, the learner posts two prices: $p_t \in [0,1]$ to the seller, and $q_t\in [0,1]$ to the buyer. If both the seller and the buyer accept (\ie $s_t \le p_t$ and $q_t \le b_t$), then the learner is awarded the gain from trade
    \[
        \gft_t(p_t,q_t) \defeq \ind{s_t \le p_t} \ind{q_t\le b_t} (b_t - s_t),
    \]
    that corresponds to the increase in social welfare generated by the trade.
    To simplify the notation, we omit the second argument of  $\gft_t$ (and of $\prof_t$) when the same price is posted to both agents. After posting the prices, the learner does \emph{not} observe directly the gain from trade or the valuations, but receives some feedback $z_t$. 
    
    \xhdr{Global budget balance.} For each time step $t$, the notion of \emph{profit} of the learner is naturally defined: if the agents accept prices $p_t$ and $q_t$, then the learner receives a net profit of $q_t - p_t \in [-1,1]$. Unlike the case of the gain from trade, the learner naturally knows its profit at the end of each time step, as it sets the prices and always observes whether the trade occurred. The learner maintains a budget $B_t$, which is initially $0$ ($B_0 = 0$) and is updated at each time step according to the profit generated or consumed: ${B_t \gets B_{t-1} + \prof_t(p_t,q_t).}$ We restrict the learner to enforce a {\em global budget balance} property which states that the final budget $B_T$ has to be non-negative with probability $1$.
    In practice, we require the learner to always post prices $p_t,q_t$ such that $(p_t-q_t) \le B_{t-1}$.\footnote{In fact, this condition is not just sufficient, but also necessary. Indeed, if $p_t-q_t>B_t$, the adversary might select valuations $(s_t, b_t)$ such that $\prof_t(p_t,q_t)<-B_{t-1}$ and thus $B_t<0$. After that, the adversary might select valuations $(s_\tau,b_\tau)=(0,0)$ for all $\tau\ge t+1$, thereby forcing $B_T=B_{t}<0$.}

    \xhdr{Feedback models.} In this paper, we study three feedback models, that we list here in increasing order of intricacy:
    \begin{itemize}
        \item\emph{Full feedback}: at the end of each round, the agents reveal their valuations (\ie $z_t = (s_t, b_t)$).
        \item\emph{Two-bit feedback}: the agents only reveal their willingness to accept the prices offered by the learner (\ie $z_t$ is composed by the two bits $(\ind{s_t \le p_t}, \ind{q_t \le b_t})$)
        \item\emph{One-bit feedback}: the learner only observes whether the trade happened or not (\ie $z_t = \ind{s_t \le p_t} \cdot\ind{q_t \le b_t}$).
    \end{itemize}
    These feedback models are not only interesting from the theoretical learning perspective, but they are also well motivated in terms of practical applications. The full-feedback model can be used to describe sealed-bid-type auctions, while the two partial feedback settings (one- and two-bit) enforce the desirable property (for the agents) of revealing a minimal amount of information to the learner. 
    
    \xhdr{Regret with respect to the best fixed price.} The goal is to maximize the total gain from trade on a fixed and known time horizon $T$ while enforcing the global budget balance condition. Following the literature on repeated bilateral trade \citep{CesaBianchiCCF21}, we measure the performance of a learning algorithm in terms of its regret with respect to the best fixed price(s) in hindsight. For any learning algorithm $\cA$ and sequence of valuations $\cS=\{(s_t,b_t)\}_{t=1}^T$ we define:
    \begin{equation} \label{def:benchmark_1}
        R_T(\cA,\cS) \defeq \max_{\substack{(p,q) \in [0,1]^2\\ p\le q}} \sum_{t=1}^T \gft_t(p,q) - \E{\sum_{t=1}^T \gft_t(p_t,q_t)},    
    \end{equation}
    where the sequence $\cS$ induces the $\gft_t$ functions and the expectation is with respect to (possibly) randomized prices $p_t$ and $q_t$ generated by the learning algorithm $\cA$. One simple property that follows immediately by definition is that, for any sequence of valuations, there exists a fixed pair of identical prices that maximizes the gain from trade. This means that the notion of ``best price in hindsight'' is well defined, and confirms the intuition that posting two different prices only helps during learning, but does not impact the maximization of gain from trade in hindsight. Finally, we define the regret of an algorithm $\cA$ (without the dependence on a specific sequence of valuations) as its worst-case performance: $R_T(\cA) = \sup_{\cS} R_T(\cA,\cS)$, where the $\sup$ is over the set of all the possible sequences of $T$ pairs of valuations. 

    \xhdr{A stronger benchmark: the best feasible distribution over prices.} In this paper we also introduce a new (stronger) benchmark for the study of repeated bilateral trade: the best fixed budget-feasible distribution over prices. This benchmark captures the flexibility of the global budget balance condition, and it arises naturally from the literature on \emph{bandits with knapsacks}. Before proceeding with the definition, let $\Delta([0,1]^2)$ be the family of all the probability measures over the measurable space $([0,1]^2, \mathcal B([0,1]^2))$, where $\mathcal B$ denotes the Borel $\sigma$-algebra.

    \begin{definition}[Best feasible distribution]
    \label{def:benchmark_2}
        For any sequence $\cS$ of seller's and buyer's valuations, we define the best fixed budget-feasible distribution over prices as the solution of:
    \begin{align}\label{eq:benchmark_2}
        \sup\limits_{\gamma\in\Delta([0,1]^2)}	& \subE{(p,q)\sim\gamma}{\sum_{t=1}^T\gft_t(p,q)} \\
\textnormal{s.t.}\quad & \subE{(p,q)\sim\gamma}{\sum_{t=1}^T\prof_t(p,q)}\ge 0,\nonumber
    \end{align}
    where $\mathbb{E}_{(p,q)\sim\gamma}$ denotes that the expectation is with respect to prices $(p,q)$ sampled according to $\gamma$. 
    \end{definition}
    This definition is well posed and there exist optimal distributions whose support contains either one or two pairs of prices. For a formal proof of this fact we refer to \Cref{pr:benchmark_2} in \Cref{app:model}.

\section{Price Discretizations and Two-Phase Algorithm}
\label{sec: preliminaries}

    In this section we present our two-phase meta algorithm, preceeded by two key results on how to discretize the price space in a way that ensures certain essential properties about profit and gain from trade.
    First, in \Cref{subsec:additive_discretization} we prove that the gain from trade of the best fixed price in hindsight is close to that of the best pair of (non-budget-balanced) prices on a suitable ``additive'' grid. Second, in \Cref{subsec:multiplicative_discretization} we construct an hybrid ``multiplicative-additive'' grid in which each interval of a one-dimensional additive grid is further divided into sub-intervals with geometrically decreasing length. This grid has the surprising property that the profit of the best fixed pair of prices on it is close to the gain from trade generated by the best fixed price in the $[0,1]$ interval, up to a poly-logarithmic multiplicative factor. Finally, we introduce our two-phase learning via the meta-algorithm $\egftmax$.

    \subsection{Additive Grid for Gain from Trade}\label{subsec:additive_discretization}

        For any integer $K$, we denote by $G_K \defeq \left\{0,\nicefrac{1}{K}, \nicefrac{2}{K}, \dots, 1-\nicefrac{1}{K},1\right\}$ the \emph{$K$-uniform grid} over $[0,1]$. Similarly, we denote with  $H_K\defeq \mleft\{\mleft( \nicefrac{i+1}{K},\nicefrac{i}{K}\mright): i\in\{0,1,\ldots,K-1\}\mright\}$ the set of pairs formed by contiguous points in the $K$-uniform grid such that the first element of the pair is greater than the second. This latter grid can be proved to enjoy the desirable property of well-approximating the gain from trade of the best fixed price, while violating the global budget balance condition by a small amount. The argument behind the approximation guarantee is simple: if $p^*$ is the best fixed price in hindsight, then the pair of prices $(\nicefrac {(i+1)}K, \nicefrac iK)$ such that $p^*$ belongs to the interval $[ \nicefrac iK, \nicefrac {(i+1)}K]$ are nearly as good as $p^*$. We have the following result.
        \begin{proposition}
        \label{prop:additive_discretization}
            For any $K$ and sequence of valuations, we have:
            \[
                \max_{p \in [0,1]} \sum_{t=1}^T \gft_t(p) \le \max\limits_{(p,q)\in H_K} \sum_{t=1}^T \gft_t(p,q) + \frac TK.
            \]
            For any $(p,q)\in H_K$, total profit  $\sum_{t=1}^T \prof_t(p,q)$ is at least $\nicefrac {-T}K$.
        \end{proposition}
        \begin{proof}
            The optimal price in hindsight $p^*$ is contained in some interval $[\nicefrac {i^*}K, \nicefrac{(i^*+1)}K]$. For any time $t$ we have the following cases:
            \begin{itemize}
                \item[$(i)$] If $\gft_t(p^*) = 0$, then the gain from trade of $(\nicefrac{(i^*+1)}K,\nicefrac {i^*}K)$ is at least $-\nicefrac 1K$ (when the valuation of the seller is $\nicefrac{(i^*+1)}K$, and that of the buyer is $\nicefrac {i^*}K$). 
                \item[$(ii)$] If $\gft_t(p^*) > 0$, then posting the pair of prices $(\nicefrac{(i^*+1)}K,\nicefrac {i^*}K)$ makes the trade happen, and guarantees the same gain from trade. 
            \end{itemize}
            Then, by summing up the gain from trade obtained by posting $(\nicefrac{(i^*+1)}K,\nicefrac {i^*}K)$, we immediately obtain the first part of the statement by applying at each $t$ either case $(i)$ or $(ii)$.
            The second part of the statement follows from the observation that the per-round deficit for posting prices $(\nicefrac{(i^*+1)}K,\nicefrac {i^*}K)$ is at most $\nicefrac 1K$. This concludes the proof.
        \end{proof}

        A simple calculation shows that $\gft_t( \nicefrac {(i+1)}K, \nicefrac iK)$ is bounded by the sum of $\gft_t(\nicefrac iK)$ and $\gft_t(\nicefrac {(i+1)}K)$. Therefore, we obtain the following known result as a Corollary to \Cref{prop:additive_discretization}.
        
        \begin{corollary}[Claim 1 of \citet{AzarFF22}]
        \label{prop:2-discretization}   
            For any $K$ and sequence of valuations, we have:
            \[
                \max_{p \in [0,1]} \sum_{t=1}^T \gft_t(p) \le  2 \cdot \max_{p\in G_K} \sum_{t=1}^T \gft_t(p) + \frac TK.
            \] 
        \end{corollary}

    \subsection{Multiplicative Grid for Profit} \label{subsec:multiplicative_discretization}
        \begin{figure} 
        \begin{subfigure}[b]{.45\columnwidth}
            \centering
            \scalebox{0.45}{\input{gridFKp}}
            \end{subfigure}
          \hspace{40pt}
            \begin{subfigure}[b]{.45\columnwidth}
            \centering
            \scalebox{0.45}{\input{gridFKm}}
      \end{subfigure}
      \caption{$F_K^+$ (left) and $F_K^-$ (right) for $K=8$, $T=32$, $i=3$.}
      \label{fig:mult_grid}
  \end{figure}
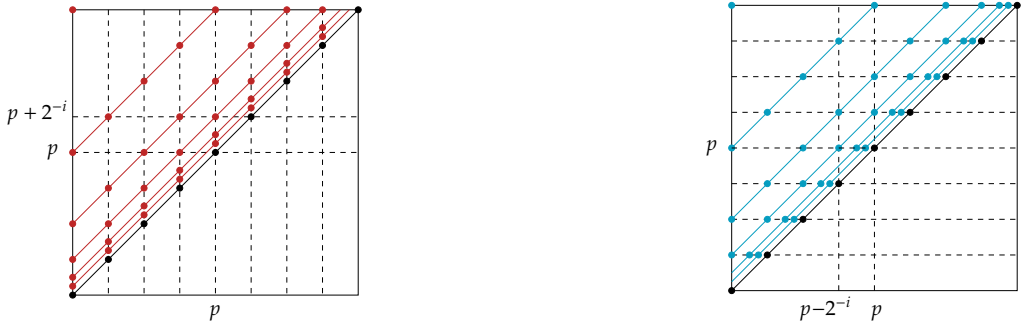
        For any $K$, we construct the two-dimensional grid $F_K$ starting from the points on the one-dimensional grid $G_K$. For each $p\in G_K$, we add to $F_K$ points of the form $(p-2^{-i},p)$ and  $(p,p+2^{-i})$, for $i = 0, 1, \dots, \lceil \log T \rceil$ so that they define intervals of geometrically decreasing length to the left and upward of $(p,p)$. Formally, we define $F_K$ as the union of $F_K^-$ and $F_K^+$ (see also \Cref{fig:mult_grid}):
        \begin{align*}
            F_K^- &\defeq \left\{(p-2^{-i},p): p \in G_K \text{ and } i \in\{0,1, \dots, \lceil \log T \rceil\}\right\} \cap[0,1]^2,  \\
            F_K^+ &\defeq \left\{(p,p + 2^{-i}): p \in G_K \text{ and } i \in \{0,1, \dots, \lceil \log T \rceil\}\right\} \cap[0,1]^2.
        \end{align*}
        The additive-multiplicative nature of $F_K$ endows it with two crucial properties: (i) its cardinality is $O(K\log T)$ an thus only depends linearly in $K$ and (ii) the profit of the best prices in $F_K$ is at least a $O(\log T)$ fraction of the $\gft$ at the best fixed price in $[0,1$], up to an additive factor of $O(\nicefrac TK)$. 
         
        \begin{proposition}
        \label{prop:multiplicative_discretization}
            For any $K$ and sequence of valuations, we have:
            \[
                \max_{p \in [0,1]} \sum_{t=1}^T \gft_t(p) \le 12\log T \cdot \max_{(p,q) \in F_K} \sum_{t=1}^T \prof_t(p,q) + \frac {5T}K.
            \]
        \end{proposition}
        \begin{proof}
            Fix the sequence $\cS$ of valuations and let $p^\ast$ be the price maximizing the gain from trade in $G_K$. We have the following chain of inequalities: 
            \begin{align}
            \nonumber
                \max_{p \in [0,1]} \sum_{t=1}^T \gft_t(p) &\le 2\sum_{t=1}^T (b_t-s_t) \ind{s_t \le p^\ast \le b_t}  + \frac{T}{K} \tag*{(by \Cref{prop:2-discretization})}\\
            \label{eq:s_tandb_t}
                &= 2\sum_{t=1}^T (b_t-p^\ast) \ind{s_t \le p^\ast \le b_t} + 2\sum_{t=1}^T (p^\ast - s_t) \ind{s_t \le p^\ast \le b_t}  + \frac{T}{K}.
            \end{align}
            We bound separately the first and second term of the right-hand side of the inequality. Starting with $(b_t-p^\ast) \ind{s_t \le p^\ast \le b_t}$, we can rewrite the expression through a case analysis depending on the interval of the discretization in which $b_t$ is located. 
            For each time step $t$, we have
            \begin{align}
            \nonumber
                (b_t-p^\ast) \ind{s_t \le p^\ast \le b_t} &\le \sum_{i=0}^{\lceil \log T \rceil} (b_t-p^\ast) \ind{s_t \le p^\ast}\ind{p^\ast + 2^{-i} \le b_t <  p^\ast + 2^{-i+1}} + \frac{1}{T},
            \end{align}
            where we used the fact that $b_t-p^\ast\le \nicefrac{1}{T}$ if $b_t$ belongs to $[p^\ast_t,p^\ast_t+\nicefrac{1}{T}]$. This yields 
            \begin{equation}
            \label{eq:b_tpow}
            (b_t-p^\ast) \ind{s_t \le p^\ast \le b_t} \le 2 \cdot \sum_{i=0}^{\lceil \log T \rceil} 2^{-i} \ind{s_t \le p^\ast}\ind{p^\ast + 2^{-i} \le b_t <  p^\ast + 2^{-i+1}} + \frac{1}{T}.  
            \end{equation}

            Let $n_i$ be the number of time steps satisfying the condition $\{s_t \le p^\ast,  p^\ast + 2^{-i} \le b_t <  p^\ast + 2^{-i+1}\}$. Summing up \Cref{eq:b_tpow} for $t=1, 2, \dots T$ we get
            \begin{align}
            \nonumber
                \sum_{t=1}^T (b_t-p^\ast) \ind{s_t \le p^\ast \le b_t} 
            &\le  2 \cdot \sum_{i=0}^{\lceil \log T \rceil} \frac{n_i}{2^i} + \frac{T}{T} \le 3 \log T\cdot \max_{i\in\mleft\{0,\dots, \log T\mright\}} \frac{n_i}{2^i} + 1 \\
            \label{eq:b_tfinal}
                &\le 3  \log T \cdot \max\limits_{(p, q) \in F_K} \sum_{t=1}^T \prof_t(p, q) + 1.
            \end{align}
            To obtain \Cref{eq:b_tfinal} we use that, for any $i\in\{0,\ldots,\lceil \log T \rceil\}$, if $n_i>0$ then it must be the case that $p^\ast+2^{-i}\in [0,1]$. Therefore, for any $i$, it it possible to obtain a profit of $2^{-i}$ by posting the pair $(p^\ast,p^\ast+2^{-i})$, which is guaranteed to belong to $F_K$ since $p^\ast\in G_K$ by construction, and $p^\ast+2^{-i}\in [0,1]$. A similar argument can be carried over for the other term of \Cref{eq:s_tandb_t}, yielding:
             \begin{equation}
             \label{eq:s_tfinal}
                \sum_{t=1}^T (p^\ast-s_t) \ind{s_t \le p^\ast \le b_t} \le 3 \log T \cdot \max\limits_{(p, q) \in F_K} \sum_{t=1}^T \prof_t(p, q) + 1.
            \end{equation}
            Finally, we plug \Cref{eq:b_tfinal,eq:s_tfinal} into \Cref{eq:s_tandb_t}, and use $K\le T$ to conclude the proof. 
        \end{proof}
        \setcounter{algocf}{0}
        \begin{algorithm}[t]        \caption{\egftmax}
         \DontPrintSemicolon
         \SetKwInput{KwData}{Input}
         \KwData{\bull budget threshold $\beta$\newline\bull integer $K$ and price-grids $F_{K}$ and $H_{K}$\newline \bull regret minimizers $\cAP$ and $\cAG$}
        Run \erevmax($\beta,F_K,\cAP$) \tcc*{\color{commentcolor}Phase I}\label{line:phase1}
        \If{\erevmax terminated at time step $\tau<T$}{
        \label{line:phase2}Initialize $\cAG$ on $H_{K}$ \tcc*{\color{commentcolor}Phase II}
                    \For{ $t = \tau + 1, 2, \dots,T$}{
                         \label{line:ag_rec}Receive from $\cAG$ the prices $(p_t, q_t)$\;
                         \label{line:ag_obs}Post prices $(p_t, q_t)$ and observe feedback $z_t$\;
                         \label{line:ag_update}Feed feedback $z_t$ to $\cAG$\;
                    }
                }
        \Hline{}
        \Fn{\erevmax($\beta,F_K,\cAP$)}{
                 \KwData{\bull budget threshold $\beta$\newline \bull grid $F_K$ of pairs of prices\newline \bull regret minimizer $\cAP$}
                 Initialize $\cAP$ on $|F_K|$ actions, one for each $(\hat p, \hat q) \in F_K$, and set $B_0 \gets 0$\;
                \For{ $t = 1, 2, \dots, T$}{
                     \label{line:ar_rec}Receive from $\cAP$ the prices $(p_t, q_t)$\;
                     \label{line:ar_obs}Post prices $(p_t, q_t)$ and observe feedback $z_t$\;
                     \label{line:ar_update}Feed feedback $z_t$ to $\cAP$\;
                     Update $B_t \gets B_{t-1} + \prof_t(p_t,q_t)$\;
                    \textbf{if} {$B_t \ge \beta$} \textbf{then}
                         Terminate the algorithm\;
                    }
                }
        \end{algorithm}

\subsection{Our Two-Phase Meta-Algorithm: \egftmax}\label{sec:metaalg}

    We describe our two-phase learning approach by presenting the meta-algorithm \egftmax. For details we refer to the pseudocode. The algorithm takes in input a budget threshold $\beta$ and an integer $K$ (which induces the two grids $F_K$ and $H_K$), and employs two regret minimizers---$\cAP$ for the profit and $\cAG$ for the gain from trade---as internal routines. 
    In the first phase (\Cref{line:phase1}), the algorithm uses function \erevmax to maximize profit until the collected budget reaches a given threshold $\beta$. This is achieved by running a regret minimizer $\cAP$ over the set $F_K$ of pairs of prices (see \Cref{subsec:multiplicative_discretization}) using profit as objective. 
    Then, in the second phase (from \Cref{line:phase2} onward), the algorithm exploits a regret minimizer $\cAG$ to maximize the gain from trade over the grid $H_K$, whose prices which are ``almost budget-balanced'' and consume only a small fraction of the previously acquired budget
    (see \Cref{prop:additive_discretization}). 
    In \Cref{sec:full} and \Cref{sec:partial} we provide regret upper bounds for this meta-algorithm in the full and one-bit feedback model, respectively. The budget threshold $\beta$, the regret minimizers, and the grid parameter $K$ are tuned according to the specific case considered.


\section{Full Feedback}\label{sec:full}

    We start by studying the \emph{full feedback} input model where the agents reveal their valuations $(s_t,b_t)$ at the end of each time step $t$. Here, the learner has counterfactual information regarding all the prices they could have posted, {\em independently} of the pair of prices actually posted at time $t$. 
    In \Cref{subsec:full_upper}, we first present a two-phase learning algorithm (\egftmax) which guarantees $\tilde O(\sqrt{T})$ regret with respect to the best fixed price in hindsight. In \Cref{subsec:full_lower} we complement this result by proving that this is tight, up to poly-logarithmic terms.
    
    \subsection{\texorpdfstring{$\tilde O(\sqrt{T})$}{sqrt T} Upper Bound with Full Feedback} \label{subsec:full_upper}

    We start the analysis by looking at the first phase of \egftmax, \erevmax (reported as a function in the pseudocode of \egftmax). We employ the \hedge algorithm (see, \eg Section 5.3 of \citet{slivkins2019introduction}) as the regret minimizer $\cAP$, which is used on the action space of the prices in $F_K$. As a first step, we note that the gain from trade of any fixed price in the first phase (which terminates at the stopping time $\tau$) is not too large.

        \begin{lemma} \label{lm:erevmax}
            Consider  \erevmax with budget threshold $\beta$, grid $F_K$, and learning algorithm \hedge as $\cAP$. Then, with probability at least $1-\nicefrac 1T$, we have
            \[
                \max_{p \in [0,1]}\sum_{t=1}^{\tau} \gft_t(p) \le 8(\beta +1)\log T   + \frac {5T}K+ 32 \log T \sqrt{T \log (T|F_K|)}.
            \]
        \end{lemma}
        \begin{proof}
    We start by observing that, by \Cref{prop:multiplicative_discretization}, there exists a pair of prices $(p^\ast,q^\ast)\in F_K$ such that
    \[
        \max_{p \in [0,1]} \sum_{t=1}^\tau \gft_t(p) \le 8\log T \cdot  \sum_{t=1}^\tau \prof_t(p^*,q^*) + \frac {5T}K.
    \]

    \hedge maintains a distribution  $\gamma_t \in \Delta(F_K)$ at each $t \in [T]$, and such distributions guarantees that the expected regret is $O(\sqrt{T\log(|F_K|)})$ \citep{slivkins2019introduction}. In particular, given $s\in [T]$, we have 
    \[ 
        \sum_{t=1}^s\prof_t(p^*,q^*) \le \sum_{t=1}^s\subE{(p,q)\sim\gamma_t}{\prof_t(p,q)} + 2 \sqrt{T \log(|F_K|)}.  
    \]
    By applying the Azuma-Hoeffding inequality for each round $s \in [T]$, and union bounding over the possible stopping times, we get that with probability at least $1-\nicefrac 1T$, we can write the following also for the stopping time $\tau$:
    \begin{align*}
    \sum_{t=1}^\tau \prof_t(p^*,q^*) &\le \sum_{t=1}^\tau \prof_t(p_t,q_t) + 2 \sqrt{T \log(|F_K|)} + 4 \sqrt{T \log (T)}.
    \end{align*}
    This yields the following chain of inequalities 
    \begin{align*}
    \max_{p \in [0,1]} \sum_{t=1}^\tau \gft_t(p) &\le 8\log T \cdot  \sum_{t=1}^\tau \prof_t(p^*,q^*) + \frac {5T}K \\
    &\le 8\log T \cdot \sum_{t=1}^\tau \prof_t(p_t,q_t) + \frac {5T}K+16\log T \sqrt{T \log(|F_K|)} + 32 \log T \sqrt{T \log (T)}  \\
    &\le   8\log T  \cdot B_\tau + \frac {5T}K+ 32 \log T \sqrt{T \log (T|F_K|)}  \\
    &\le 8\log T  \cdot (\beta +1) + \frac {5T}K++ \frac {5T}K+ 32 \log T \sqrt{T \log (T|F_K|)}
    \end{align*}
    This concludes the proof.
\end{proof}

    \Cref{lm:erevmax} helps us bounding the regret of \egftmax up to the (random) time step $\tau$, when the algorithm switches from profit to gain from trade maximization. Setting $\beta=\sqrt{T}$ and $K=\sqrt{T}$, and using \hedge as regret minimizer also in the second phase, yields the following result. 

        \begin{theorem}\label{thm:full_upper}
            Consider the repeated bilateral trade problem in the full feedback model. There exists a learning algorithm $\cA$ that respects global budget balance and whose regret with respect to the best fixed price in hindsight verifies
            \[
                R_T(\cA) \le 92 \,\log^{\nicefrac 3 2}(T)\,\sqrt{T }.
            \]
        \end{theorem}
        \begin{proof}
            We prove that algorithm \egftmax with the proper choice of budget $\beta$, grids $F_K$ and $H_K$, and algorithms $\cAP$ and $\cAG$ achieves the desired regret bound, while enforcing global budget balance.
            
            First, we show that the algorithm enforces global budget balance for any value of the stopping time $\tau\in [T]$. 
            By construction, the profit at time $\tau$ (\ie right after the end of first phase in which we employ the subroutine \erevmax) is at least $\beta$.
            Moreover, in each round $t\in \{\tau+1,\ldots, T\}$ of the second phase, the profit is at least $-\nicefrac 1 K$. Hence, the cumulative profit at time $T$ is at least
            \[
            \beta-(T-\tau)\frac 1 K \ge \sqrt{T}- T \frac{1}{\sqrt{T}}=0.
            \]

            Then, we prove the upper bound on the cumulative regret.
            We start by considering the regret accumulated in the interval $\{\tau+1,\ldots,T\}$. In particular, for any $\tau\in[T]$, we have
            \begin{align}\nonumber
            \max_{p \in [0,1]}\sum_{t=\tau+1}^T \gft_t(p) \le& \max\limits_{(p,q)\in H_K} \sum_{t=\tau+1}^T \gft_t(p,q) + \frac TK\\
            \le& \subE{(p_t,q_t)\sim \gamma_t}{\sum_{t=\tau+1}^T \gft_t(p_t,q_t)}+\frac TK+4 \sqrt{T \log(|H_K|)},\label{eq:boundGFTsecondTerm}
            \end{align}
            where the first inequality follows from Proposition~\ref{prop:additive_discretization}, and the second inequality follows from the regret bound of \textsc{Hedge} when the range of the rewards is $[-\nicefrac 1 K,1]$ and $\gamma_t$ is the probability distribution over the action set maintained by \hedge (instantiated to maximize gain from trade in the second phase). 
            
            Then, assume that the bound in \Cref{lm:erevmax} holds, which happens with probability at least $1-\nicefrac 1 T$. By employing \Cref{eq:boundGFTsecondTerm} and \Cref{lm:erevmax} we can show that, with probability at least $1-\nicefrac 1 T$,
            \begin{align*}
                \max_{p \in [0,1]} &\sum_{t=1}^T \gft_t(p)  \le \max_{p \in [0,1]} \sum_{t=1}^\tau \gft_t(p) +\max_{p \in [0,1]} \sum_{t=\tau+1}^T \gft_t(p) \\
                &\hspace{-0.5cm}\le 8\log T  (\beta +1) + \frac {5T}K+32 \log T\sqrt{T \log (T|F_K|)} +\max_{p \in [0,1]} \sum_{t=\tau+1}^T \gft_t(p) \\
                &\hspace{-0.5cm} \le  \E{\sum_{t=\tau+1}^T \gft_t(p_t,q_t)}+ 8\log T  (\beta +1) + \frac {6T}K+ 32 \log T\sqrt{T \log (T|F_K|)} +4 \sqrt{T \log(|H_K|)},
            \end{align*}
        where the second inequality follows from Lemma~\ref{lm:erevmax}, and the third one from \Cref{eq:boundGFTsecondTerm}.

        Then, by substituting $\beta = K=\sqrt T$ we can conclude that
        \begin{align*}
            \max_{p \in [0,1]} \sum_{t=1}^T \gft_t(p)& \le  \E{\sum_{t=\tau+1}^T \gft_t(p_t,q_t)}+ 8\log T  (\sqrt{T} +1) + 6\sqrt{T}\\
            &\hspace{.7cm}+32\log T \sqrt{T \log\mleft(2T^{\nicefrac 32} \mleft(\log T+1\mright)\mright)} +4 \sqrt{T \log(\sqrt{T})}  \\
            &\le \E{\sum_{t=\tau+1}^T \gft_t(p_t,q_t)}+ 90 \log^{\nicefrac 3 2}(T)\sqrt{T } ,
        \end{align*}

        By rearranging we have that, with probability at least $1-\nicefrac 1 T$, it holds
        \begin{align*}
        \max_{p \in [0,1]} \sum_{t=1}^T \gft_t(p) -\E{\sum_{t=1}^T \gft_t(p_t,q_t)}&\le \max_{p \in [0,1]} \sum_{t=1}^T \gft_t(p) -\E{\sum_{t=\tau+1}^T \gft_t(p_t,q_t)} \\
        & \le   90 \log^{\nicefrac 32}(T)\sqrt{ T },
        \end{align*}
        where the first inequality follows from the fact that the gain from trade is always non-negative.

        Finally, we can conclude that the expected regret is at most
        \[
        R_T(\egftmax)\le\left(1-\frac 1 T\right) \left(90 \log^{\frac 3 2}(T)\sqrt{T }\right)+ \frac{1}{T} \cdot 2 T \le 92 \log^{\frac 3 2}(T)\sqrt{T }. 
        \]
        This concludes the proof. 
    \end{proof}

\subsection{$\Omega(\sqrt{T})$ Lower Bound with Full Feedback} \label{subsec:full_lower}

    \begin{figure}
    \centering
    \input{instance_LB_Full}
      \caption{Partition of $[0,1]^2$ as in the proof of \Cref{thm:full_lower}.}\label{fig:sqrtLB}
  \end{figure}
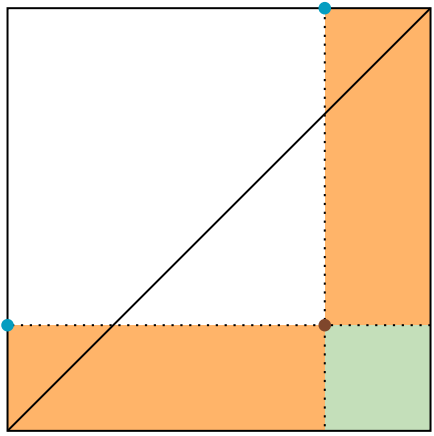
    We present a lower bound that shows how the regret rate in \Cref{thm:full_upper} is optimal up to poly-logarithmic factors. The lower bound is based on the following stochastic sequence: at each time step $t$ the pair $(s_t,b_t)$ is drawn uniformly at random between $3$ pairs of valuations: $(0,\nicefrac14)$, $(\nicefrac34,1)$ and $(\nicefrac34,\nicefrac14)$. These three points naturally partition the $[0,1]^2$ square into four regions (see \Cref{fig:sqrtLB}). Crucially, prices in the $[\nicefrac34,1]\times[0,\nicefrac 13]$ region (\textcolor{mgreen}{\textbf{green}} in \Cref{fig:sqrtLB}) incur in negative expected gain from trade, while prices in the $[0,\nicefrac34)\times(\nicefrac 13,1]$ region (white in \Cref{fig:sqrtLB}) miss all trades. Therefore, the only reasonable option for any learner is to post prices in the two remaining regions (\textcolor{orange2}{\textbf{orange}} in \Cref{fig:sqrtLB}), with an expected gain from trade of $\nicefrac{1}{12}$. This allows for a reduction to an expert problem with $2$ available actions (one for each of the two orange regions). This construction highlights a key difficulty if compared to lower bounds for per-round budget balanced algorithms: we need to disincentivize the learner from choosing non budget balanced prices below the diagonal. We have the following Theorem, which is preceded by a preliminary Lemma. 
    \begin{lemma}
    \label{lem:random}
        Let $S_n$ be a symmetric random walk on the line after $n$ steps, starting from $0$. Then, for $n$ large enough, it holds that $
          \E{|S_n|} \ge \nicefrac 23 \sqrt{n}.$
    \end{lemma}
\begin{proof}
    It is well known that the expected distance of a random walk from the origin grows like $\Theta(\sqrt{n})$. Formally, the following asymptotic result holds (see, \eg \citet{Palacios08}):
        \[
            \lim_{n \to \infty}\frac{\E{|S_n|}}{\sqrt n} = \sqrt{\frac{2}{\pi}}.
        \]
    Observe that $\sqrt{\nicefrac{2}{\pi}} > \nicefrac23$. Then, there exists a finite $n_0$ such that $\E{|S_n|} \ge \nicefrac 23 \sqrt{n}$ for all $n \ge n_0$.
\end{proof}
    
    \begin{theorem}\label{thm:full_lower}
        Consider the repeated bilateral trade problem in the full feedback model. Any learning algorithm that satisfies global budget balance suffers at least $
            \Omega(\sqrt{T})
        $ regret with respect to the best fixed price in hindsight. 
    \end{theorem}

    \begin{proof}
        We prove this result via Yao's principle~\citep{Yao77}.
        We apply the easy direction of the theorem, which reads (using our terminology) as follows: the regret $R_T(\cA)$ of a randomized learner $\cA$ against the worst-case valuations sequence is at least the regret of the optimal deterministic learner $A$ against a stochastic sequence of valuations $\cS$. Formally,
    \[
        R_T(\cA) \ge  \sup_{A} \E{\max_{p \in [0,1]}\sum_{t=1}^T \gft_t(p) - \sum_{t=1}^T \gft_t(p_t,q_t)},
    \]
    where the expectation is with respect to the stochastic valuation sequence $\cS$, while $A$ denotes deterministic learner that posts the $(p_t,q_t)$ prices.
    In particular, we construct a randomized instance $\cS$ such that any deterministic learning algorithm must suffer, in expectation with respect to the randomness of $\cS$, at least $c \sqrt{T}$ regret for some constant $c$.

        The randomized instance is constructed as follows: at each time step $t\in [T]$ the adversary selects uniformly and independently at random one of the following three points $(0, \nicefrac{1}{4})$, $(\nicefrac{3}{4}, 1)$ and $(\nicefrac{3}{4}, \nicefrac{1}{4})$.
        We first compute a lower bound on the expected gain from trade achieved by the best fixed price in hindsight, and then we provide an upper bound on the expected gain from trade which can be attained by any deterministic learning algorithm. Combining these two intermediate results will yield the statement via Yao's principle.

        Let $N_0$ be a random variable denoting the number of times that $(\nicefrac{3}{4},\nicefrac{1}{4})$ is realized. Analogously, let $N_1$ (resp., $N_2$), be the number of times in which $(0,\nicefrac{1}{4})$ (resp., $(\nicefrac{3}{4},1)$) is realized. Clearly, $N_0 + N_1 + N_2 = T$, and $\E{N_i} = \nicefrac T3$ for any $i=0,1,2$. Conditioning on $N_0$, the remaining $T-N_0$ valuations are either $(0,\nicefrac{1}{4})$ or $(\nicefrac{3}{4},1)$, sampled uniformly and independently at random. 
        
        Then, we have that
        \begin{align}
        \nonumber
            \E{ \max_{p \in [0,1]} \sum_{t=1}^T \gft_t(p) \bmid N_0} &\ge \E{\max_{p \in \left\{\nicefrac 14,\nicefrac 34\right\}}\sum_{t=1}^T \gft_t(p)\bmid N_0}\\
            \nonumber
            &=\E{\max_{p \in \left\{\nicefrac 14,\nicefrac 34\right\}} \sum_{t=1}^T \frac{1}{4} \ind{s_t \le p \le b_t}\bmid N_0}\\
            &= \frac 14 \,\E{\max\left\{N_1, T - N_0 - N_1\right\}\mid N_0}\tag*{\text{(Definitions of $N_i$)}}\\
            \nonumber
            &= \frac 18 (T-N_0) + \frac 18\E{\max\left\{2N_1 - T + N_0,\, T - N_0 - 2N_1\right\}\mid N_0}\\
            \nonumber
            &= \frac 18 (T-N_0) + \frac 18\E{\max\left\{N_1 - N_2,\, N_2 - N_1\right\}\mid N_0}\\
            \label{eq:random_walk}
            &= \frac 18 (T-N_0) + \frac 18\E{\lvert S_{T-N_0}\rvert\mid N_0}\\
            &\ge \frac 18 (T-N_0) + \frac{\sqrt{T - N_0}}{12}, \tag*{(\Cref{lem:random})}
        \end{align}
        where \Cref{eq:random_walk} follows by considering a symmetric random walk on a line on $T - N_0$ steps that goes left when $(s_t,b_t) = (0, \nicefrac 14)$, and goes right when $(s_t,b_t) = (\nicefrac 34,1)$. Now, we can take the expectation (with respect to $N_0$) on the first and last term of the previous chain of inequalities to get
        \begin{align}
            \E{\max_{p \in [0,1]}\sum_{t=1}^T \gft_t(p)} &= \E{\E{\max_{p \in [0,1]}\sum_{t=1}^T \gft_t(p) \bmid N_0}} \tag*{(Conditional expectation)}\\
            &\ge  \frac{1}{12} T + \frac {1}{12}\E{\sqrt{T - N_0}}\tag*{\text{($\E{N_0} = \nicefrac{T}{3}$)}}\\
            &\ge \frac 1{12} T + \frac {\sqrt{T}}{24} \mathbb P({N_0 \le \nicefrac{3}{4}T})\tag*{\text{(Conditioning on $N_0 \le \nicefrac{3}{4}T$)}}\\
            \label{eq:best_price_lower}
            &\ge \frac 1{12} T + \frac{5\sqrt{T}}{216},
        \end{align}
        where the last line follows from Markov's inequality.

        Now, we construct an upper bound on the gain from trade achievable by {\em any} deterministic learning algorithm (even without the constraint of enforcing global budget balance). Consider what happens at each fixed time steps $t$: the history of the realized valuations up to that point induce deterministically the pair of prices $(p_t, q_t)$ posted by the learning algorithm. We prove now that no matter $(p_t, q_t)$ chosen, the learner does not achieve more than an expected gain from trade of $\nicefrac{1}{12}$. To see this we study separately four cases: 
        \begin{itemize}
            \item If $(p_t, q_t) \in [0,\nicefrac 34) \times (\nicefrac 14,1]$, then $\gft_t(p_t,q_t) = 0$ with probability $1$ because it misses all the possible trades.
            \item If $(p_t, q_t) \in [0,\nicefrac 34) \times [0,\nicefrac 14]$, then the learner gets $\nicefrac 14$ gain from trade only when $(s_t,b_t) = (0, \nicefrac 14)$ is realized and $0$ otherwise, for an expected gain from trade of $\nicefrac 1{12}$
            \item Similarly, if $(p_t, q_t) \in (\nicefrac 34,1] \times (\nicefrac 14,1]$, then the learner gets $\nicefrac 14$ gain from trade only when $(s_t,b_t) = (\nicefrac 34,1)$ is realized and $0$ otherwise, for an expected gain from trade of $\nicefrac 1{12}$
            \item Finally. if $(p_t, q_t) \in [\nicefrac 34,1] \times [0,\nicefrac 14]$, then the learner always observes a trade, but the expected gain from trade it achieves is $0$ ($\nicefrac 14$ with probability $\nicefrac 23$ and $-\nicefrac{1}{2}$ with the remaining probability).
        \end{itemize}
        Therefore, no matter what the learner does, it gets expected gain from trade at most $\nicefrac{T}{12}$:
        \begin{equation}
        \label{eq:algorithm_upper}
            \E{\sum_{t=1}^T \gft_t(p_t,q_t)} \le \frac{T}{12}.
        \end{equation}

        We can conclude the proof of the Theorem by combining \Cref{eq:best_price_lower} and \Cref{eq:algorithm_upper} to get:
        \[
            \E{\max_{p \in [0,1]}\sum_{t=1}^T \gft_t(p) - \sum_{t=1}^T \gft_t(p_t,q_t)} \ge \left(\frac 1{12} T + \frac{5\sqrt{T}}{216}\right)- \frac{T}{12} = \frac{5\sqrt{T}}{216},
        \]
        where that the randomness is with respect to the sequence generated by the randomized adversary. This concludes the proof.
     \end{proof}

\section{Partial Feedback}\label{sec:partial}

    In this section, we study the more challenging partial feedback models. In \Cref{subsec:partialUpper}, we provide a positive result for the case of one-bit feedback ($z_t = \ind{s_t \le p_t} \cdot\ind{q_t \le b_t}$), where the learner only observes whether the trade happened or not. In particular, we show that  \egftmax, with a suitable initialization, achieves a regret of the order $\tilde{O}(T^{\nicefrac 3 4} )$.  
    Differently from the full-information setting, the design of a no-regret algorithm for the gain from trade (\ie $\cAG$) is particularly challenging as we need to build an estimator for the gain from trade by only playing non-budget balanced prices in $H_K$.
    
    In \Cref{subsec:partialLower} we complement the regret upper bound by proving that every algorithm has regret at least $\Omega(T^{\nicefrac 57})$, even with two-bit feedback ($z_t = (\ind{s_t \le p_t}, \ind{q_t \le b_t})$), \ie where each agent separately reveal their willingness to accept the prices posted.
    One of the main challenges posed by such a lower bound resides in handling non-budget balanced prices, as any algorithm could temporarily sacrifice some profit while collecting large \gft.

    \subsection{\texorpdfstring{$\tilde O(T^{\nicefrac 34})$}{T to the 3/4} Upper Bound with One-Bit Feedback} \label{subsec:partialUpper}

    We show how to employ \egftmax with a suitable choice of parameters $\beta$ and $K$, and regret minimizers $\cAP$ and $\cAG$ to achieve the desired regret bound. \Cref{subsec:reg min rev one bit} presents a regret-minimizing algorithm that can be employed as $\cAP$, while \Cref{subsec:reg min gft one bit} provides a suitable regret minimizer to be employed as $\cAG$. Finally, in \Cref{subsec: one bit final}, we present the final regret upper bound. 

   \subsubsection{Regret Minimizer for Profit under Partial Feedback}\label{subsec:reg min rev one bit}
    As in the full-information setting, we exploit \erevmax to maximize the profit until the accrued budget is at least a given threshold $\beta$.
    In particular, we instantiate the subroutine \erevmax with \textsc{EXP3.P}~\citep{AuerCFS02} as regret minimizer $\cAP$ and grid $F_K$.
    The following lemma shows that the gain from trade of any fixed price $p$ in the first phase is small enough up to the stopping time $\tau$ that terminates the first phase.
        \begin{lemma}\label{lm:erevmax2}
            Consider \erevmax with budget threshold $\beta$, grid $F_K$, and learning algorithm \textsc{EXP3.P} as $\cAP$. Then with probability at least $1-\nicefrac 1T$, we have that $\max_{p \in [0,1]}\!\sum_{t=1}^{\tau} \!\gft_t(p)$ is at most
            \(
                 8(\beta +1)\log T   + \tfrac {5T}K+256\log T  \sqrt{|F_K|T\log(|F_K|T)}.
            \)
        \end{lemma}
        \begin{proof}
    First, note that by \Cref{prop:multiplicative_discretization} there exists a pair of prices $(p^\ast,q^\ast)\in F_K$ such that
    \[
        \max_{p \in [0,1]} \sum_{t=1}^\tau \gft_t(p) \le 8\log T \cdot  \sum_{t=1}^\tau \prof_t(p^*,q^*) + \frac {5T}K.
    \]

    Moreover, \textsc{EXP.P} guarantees that, with probability at least $1-\nicefrac 1T$, it holds
    \begin{align*}
    \sum_{t=1}^\tau \prof_t(p^*,q^*) &\le \sum_{t=1}^\tau \prof_t(p_t,q_t) + 32 \sqrt{|F_K|T\log(|F_K|T)}.
    \end{align*}
    Then,
    \begin{align*}
    \max_{p \in [0,1]} \sum_{t=1}^\tau \gft_t(p) & \le 8\log T \cdot  \sum_{t=1}^\tau \prof_t(p^*,q^*) + \frac {5T}K\\
    &\le 8\log T \cdot \sum_{t=1}^\tau \prof_t(p_t,q_t) + \frac {5T}K + 256\log T  \sqrt{|F_K|T\log(|F_K|T)} \\
    &\le   8 \log T  B_\tau + \frac {5T}K+ 256 \log T  \sqrt{|F_K|T\log(|F_K|T)} \\
    &\le 8\log T  (\beta +1) + \frac {5T}K+ 256 \log T  \sqrt{|F_K|T\log(|F_K|T)}.
    \end{align*}
    This concludes the proof.
\end{proof}

    \subsubsection{Regret Minimizer for Gain from Trade under Partial Feedback}\label{subsec:reg min gft one bit}

    A crucial ingredient we need is an estimation procedure capable of extracting quantitative information from the gain from trade, having only access to one bit of feedback. More precisely, we need an estimation procedure of the gain from trade function $H_K\ni(p,q)\mapsto\gft_t(p,q)$.
    A similar challenge is faced in \citet{AzarFF22}, where the action set consists of a discretization of a single price (\ie their estimation procedure posts $p$ to both seller and buyer). However, in our scenario, such symmetry no longer applies. Here, we must consider the grid $H_K$, which employs distinct prices for the seller and the buyer ($p+\nicefrac 1K$ and $p$, respectively). Thus, our estimation procedure $\gftest$ has an asymmetric structure (see the pseudocode, in particular \Cref{line:est seller,line:est buyer}). 
    
    First, $\gftest$ draws a sample from a Bernoulli distribution with parameter $(pK+1)/(K+1)$ (\Cref{line:bernoulli}). If the result is $1$, it posts price $p$ to the buyer, and the seller receives a price drawn uniformly at random from $[0,p+\nicefrac 1K]$ (\Cref{line:est buyer}). Otherwise, if the result is $0$, $\gftest$ posts price $p$ to the seller, and the buyer's price is drawn uniformly at random from $[p,1]$. We denote the final estimate at $t$ by $\egft_t(p+\nicefrac 1K,p)$ (\Cref{line:est seller}). Overall, our estimator has a small bias, as formalized in the following Lemma.

\begin{algorithm}[t]        \caption{\blockdecomposition}
        \DontPrintSemicolon
        \SetKwInput{KwData}{Input}
        \KwData{\bull Number of rounds $T$ and number of blocks $N$\newline \bull Set of prices $H_K$}
        Initialize \hedge over action space $H_K$ and time horizon $N$\;
        Initialize random mappings $h_j$ for all $j\in\{0,\ldots,N-1\}$\;
        $\cB_j\gets\{j\tfrac TN +1,\ldots, (j+1)\tfrac TN\}$ for all $j\in\{0,\ldots,N-1\}$\;
        
        \For{$j\in \{0,\ldots, N-1\}$}{
            Receive from $\cA$ the distribution over pair of prices $\vx_j$\label{line:block call hedge}\;
            \For{$t\in \cB_j$}{
            \If{$t\notin S_j$}{
                Play $(p,q)\sim \vx_j$ and observe $\ind{s_t\le p\wedge q\le b_t}$\label{line:block play hedge}\;
            }
            \Else{
                Select prices $(p,q)$ such that $h_j(p,q)=t$\label{line:block random}\;
                Compute $\egft_t(p,q)$ through \gftest\;
                \label{line:estimate}
                $\bgft_j(p,q)\gets \egft_t(p,q)$\label{line:block rew}\;
            }
            }
            Update $\cA$ with reward vector $\bgft_j$\label{line:block update hedge}
        }
        \Hline{}
        \Fn{\text{\gftest}}{
                 \SetKwInput{KwData}{Input}
        \KwData{prices $(p+\nicefrac 1K,p)\in H_K$}
        Sample $Z$ from a Bernoulli with parameter $\frac{pK+1}{K+1}$\label{line:bernoulli}\;
        \If{$Z=1$}{
            Post price $(\tilde p,p)$, with $\tilde p\sim U[0,p+\nicefrac 1K]$\label{line:est buyer}\;
            $\widehat\gft_t(p+\nicefrac 1K,p)\gets\ind{s_t\le \tilde p}\ind{p\le b_t}$\;
            }
        \Else{
        Post price $(p,\tilde p)$, with $\tilde p\sim U[p,1]$\label{line:est seller}\;
        $\widehat\gft_t(p+\nicefrac 1K,p)\gets\ind{s_t\le  p}\ind{\tilde p\le b_t}$\;
        }
        \Return $\widehat\gft_t(p+\nicefrac 1K,p)$
        }
        \end{algorithm}


    \begin{lemma}\label{lm:bgftmax}
        For every $(p+\nicefrac 1K,p)\in H_K$, the random variable $\widehat\gft_t(p+\nicefrac 1K,p)$ is an $\nicefrac 1K$-biased estimate of $GFT_t(p+\nicefrac 1K,p)$, \emph{i.e.},
       \[
            \left|\gft_t\left(p+\tfrac 1K,p\right)-\mathbb{E}\left[\widehat\gft_t\left(p+\tfrac 1K,p\right)\right]\right|\le \frac2K.
       \]
    \end{lemma}
    \begin{proof}
        First, we observe that for $\tilde p\sim U[0,p+\nicefrac 1K]$ (\ie drawn independently from the uniform distribution over the $U[0,p+\nicefrac 1K]$ interval) we have that
        \[\E{\ind{s_t\le \tilde p}}=\ind{s_t\le p +\nicefrac 1K}\left(1-\frac{s_t}{p+\nicefrac 1K}\right),\]
        and for $\tilde p\sim U[p,1]$ we have
        \[\E{\ind{\tilde p \le b_t}}=\ind{p\le b_t}\left(\frac{b_t-p}{1-p}\right).\]
        Using these two equations, we can compute the expected value of the random variable returned by \gftest. Indeed, by the law of total expectation, we have

        \begin{align*}
        \mathbb{E}&\left[\widehat\gft_t\left(p+\frac1K,p\right)\right] \\
        &= \frac{pK+1}{K+1}\ind{p\le b_t}\mathop{\mathbb{P}}\limits_{\tilde p\sim U[0,p+\frac1K]}\left[{s_t\le\tilde p}\right]+\frac{1-p}{1+\nicefrac 1K}\ind{s_t\le p+\nicefrac 1K}\mathop{\mathbb{P}}\limits_{\tilde p\sim U[p,1]}[{\tilde p\le b_t)}]\\
            &=\frac{pK+1}{K+1}\ind{p\le b_t}\,\ind{s_t\le p+\nicefrac 1K}\left(1-\frac{s_t}{p+\nicefrac 1K}\right)+\frac{1-p}{1+\nicefrac 1K}\ind{s_t\le p+\nicefrac 1K}\ind{p\le b_t}\frac{b_t-p}{1-p}\\
            &=\frac{K}{K+1}\left(b_t-s_t+\nicefrac 1K\right)\ind{s_t\le p+\nicefrac 1K}\ind{p\le b_t}
            \end{align*}
            
            We can thus conclude the proof by observing that:
            \begin{align*}
                \left|\gft_t\mleft(p+\frac1K,p\mright)-\mathbb{E}\left[\egft_t\mleft(p+\frac1K,p\mright)\right]\right|&
                =\left|\ind{s_t\le p+\frac1K}\ind{p\le b_t}\left(b_t-s_t-\frac{b_t-s_t+\frac1K}{1+\frac1K}\right)\right|\le \frac2K,
            \end{align*}
            where the last inequality holds since $\left|a-\frac{a+\eps}{1+\eps}\right|\le 2\eps$ for all $a\in[-1,1]$ and $\eps<1$.
    \end{proof}

    Given the estimation procedure \gftest, it is possible to turn any no-regret algorithm for the full-feedback setting into a regret minimizer for the partial feedback setting by the standard block decomposition technique (see, e.g., Chapter 4 of \citet{NRTV2007}). The procedure, which we call \blockdecomposition is described in the pseudocode. We assume to employ \hedge as the full-feedback regret minimizer $\cA$.
    
    \blockdecomposition works by subdividing the time horizon $T$ into $N$ blocks $\cB_1,\ldots,\cB_N$ of equal size and contiguous, that is $\cB_j\defeq\{j \nicefrac TN+1,\ldots, (j+1)\nicefrac TN\}$ for any $j\in\{0,1,\ldots,N-1\}$.
    In each block we select uniformly at random $K$ time steps (\ie one for each pair in $H_K$), and we randomly assign each of such time steps to one pair of prices in $H_K$. Formally, for each block $j$, we have a one-to-one map $h_j:\cB_j\to H_K$ which is a uniform random map from prices in $H_K$ to rounds in block $\cB_j$. We call the image of $h_j$ the {\em exploration rounds}, and we denote the set of such rounds by $S_j$. 
    For any block $j$, the algorithm builds a vector $\bgft_j$ such that the entry $\bgft_j(p,q)$ is an estimation of the reward of the pair $(p,q) \in H_K$ in block $\cB_j$. To do that, for any block $j$ and pair of prices $(p,q)\in H_K$, we let $\bgft_j(p,q)=\egft_t(p,q)$, where $t=h_j(p,q)$ and $\egft_t(p,q)$ is computed through the estimation procedure \gftest with prices $(p,q)$ (\Cref{line:block random,line:estimate}). For any block $j$, exploration rounds in $S_j$ are used to build $\bgft_j$. In all the other rounds in $\cB_j\setminus S_j$ the algorithm plays according to the strategy $\vx_j\in\Delta(H_K)$ (\Cref{line:block play hedge}) computed by $\cA$ at the beginning of block $j$ (\Cref{line:block call hedge}). At the end of each block $j$, the full-information subroutine $\cA$ is updated using $\bgft_j$ (\Cref{line:block update hedge}).

    Let $\gft_j(p,q)=\sum_{t\in \cB_j}\gft_t(p,q)/|\cB_j|$ be the average $\gft$ over block $\cB_j$.
    Since we choose exploration rounds uniformly at random throughout block $\cB_j$ we have that, for any $(p,q)\in H_K$,
    \begin{align*}
        \left|\E{\bgft_j(p,q)}-\gft_j(p,q)\right|&=\left|\sum\limits_{t\in \cB_j}\frac{1}{|\cB_j|}\left(\E{\egft_t(p,q)}-\gft_t(p,q)\right)\right|\\
        &\le\sum\limits_{t\in \cB_j}\frac{1}{|\cB_j|}\left|\E{\egft_t(p,q)}-\gft_t(p,q)\right|\\
        &\le\frac{2}{K},
    \end{align*}
    where the last equality follows from \Cref{lm:bgftmax}.
    This yields the following guarantees on the regret of \blockdecomposition.
    Let $\vx_t$ be the distribution over $H_K$ employed to sample $(p,q)$ at time $t$. At time $t$, $t\in \cB_j$, we have
        \[
        \vx_t\defeq\mleft\{\begin{array}{l}
                \displaystyle
                 \vx_j \hspace{4.5cm}\text{\normalfont if } t\notin S_j\,\, \text{\normalfont (\Cref{line:block play hedge})}\\[2mm]
                \text{\normalfont play } (p,q)  \,\,\text{\normalfont s.t. } h_j(p,q)=t\hspace{.5cm}\text{otherwise}\,\, \text{\normalfont (\Cref{line:estimate})}
            \end{array}\mright.
        \]
      where $\vx_j$ is the distribution computed by \hedge for block $j$.
      The following lemma states precisely the guarantees provided by \blockdecomposition.

    \begin{lemma}\label{lm:gftbandit} \blockdecomposition with $K=T^{\nicefrac 14}$ and $N=T^{\nicefrac 12}$ guarantees:
        \[
        \sup_{(p,q)\in H_K} \sum_{t=1}^T\gft_t(p,q) - \sum_{t=1}^T \mathop{\mathbb{E}}_{(p,q)\sim\vx_t}\!\!\!\![\gft_t(p,q)]\le \frac 52\,T^{\nicefrac 34}\sqrt{\log(T)}.
        \]
    \end{lemma}
    \begin{proof}
Let $R_N^{H}$ be the regret accumulated by \hedge over $N$ rounds when it observes utilities in $[0,1]$ and plays over $K$ actions. Each exploration round can cost at most 1 with respect to playing according to $\vx_j$, and there are $NK$ such rounds.  
        Then, we have that
        \begin{align*}
            \sum\limits_{t=1}^T \sum\limits_{(p,q)\in H_K}\gft_t(p,q)\vx_t(p,q)&\ge \sum\limits_{j\in[N]}\sum\limits_{(p,q)\in H_K}\frac{T}{N}\cdot\gft_j(p,q)\,\vx_j(p,q)-NK\\
            &\ge\sum\limits_{j\in[N]} \sum\limits_{(p,q)\in H_K}\frac{T}{N} \cdot\left(\E{\bgft_j(p,q)}-\frac2K\right)\vx_j(p,q)-NK\\
            &=\E{\sum\limits_{j\in[N]} \sum\limits_{(p,q)\in H_K}\frac{T}{N}\,\bgft_j(p,q)\,\vx_j(p,q)}-\frac{2T}{K}-NK\\
            &\ge \E{\sup\limits_{(p,q)\in H_K}\sum\limits_{j\in[N]}\frac{T}{N}\,\bgft_j(p,q)}-\frac{T}{N}R_N^{H}-\frac{2T}{K}-NK\\
            &\ge\sup\limits_{(p,q)\in H_K}\E{\sum\limits_{j\in[N]}\frac{T}{N}\,\bgft_j(p,q)}-\frac{T}{N}R^{H}_N-\frac{2T}{K}-NK\\
            &\ge\sup\limits_{(p,q)\in H_K}\sum\limits_{j\in[N]}\frac{T}{N}\,\gft_j(p,q)-\frac{T}{N}R^{H}_N-\frac{2T}{K}-NK\\
            &=\sup\limits_{(p,q)\in H_K} \sum\limits_{t=1}^T\gft_t(p,q)-\frac{T}{N}R^{H}_N-\frac{2T}{K}-NK.\\
        \end{align*}
        By rearranging we obtain that
        \[
\sup\limits_{(p,q)\in H_K} \sum\limits_{t=1}^T\gft_t(p,q) - \sum\limits_{t=1}^T \sum\limits_{(p,q)\in H_K}\gft_t(p,q)\vx_t(p,q)\le \frac{T}{N}R^{H}_N+\frac{2T}{K}+NK.
        \]
        It is known that $R_N^{H}\le 4\sqrt{N\log K}$~(see, \eg \citet{slivkins2019introduction}). Then, by setting $K=T^{1/4}$ and $N=T^{1/2}$ we obtain
        \begin{align*}
        \sup\limits_{(p,q)\in H_K} \sum\limits_{t=1}^T\gft_t(p,q) - \sum\limits_{t=1}^T \sum\limits_{(p,q)\in H_K}\gft_t(p,q)\vx_t(p,q)&\le T^{3/4}\left(3+4\sqrt{\log(T^{1/4})}\right)\\&\le 16\cdot T^{3/4}\sqrt{\log(T^{\nicefrac 14})},
        \end{align*}
        where the last inequality holds for all $T\ge 2$.
        This concludes the proof.
    \end{proof}

    \subsubsection{Putting Everything Together}\label{subsec: one bit final}

    \egftmax with the two regret minimizers described in \Cref{subsec:reg min rev one bit,subsec:reg min gft one bit} guarantees a $O(T^{\nicefrac 34})$ bound on the regret. 
   \begin{restatable}{theorem}{partialupper}\label{thm:partial_upper}
            Consider the repeated bilateral trade problem in the one-bit feedback model. There exists a learning algorithm $\cA$ that respects global budget balance and whose regret with respect to the best fixed price in hindsight verifies: 
            \[
                R_T(\cA) \le 1282\cdot T^{\nicefrac34}\log^2 T.
            \]
        \end{restatable}
        \begin{proof}
            The proof follows the same structure of \Cref{thm:full_upper}. In this case, we set $\beta=T^{\nicefrac 34}$ and $K=T^{\nicefrac 14}$, and consider \egftmax with \textsc{EXP3.P}~\citep{AuerCFS02} instantiated over set $F_K$ (see \Cref{subsec:reg min rev one bit}) as the regret minimizer $\cAP$, and \blockdecomposition  instantiated over $H_K$ (see \Cref{subsec:reg min gft one bit}) as the regret minimizer $\cAG$.
    

            By construction of \erevmax, for any stopping time $\tau$ the profit is at least $\beta$, and in rounds $\tau+1,\ldots,T$ the budget spent is at most $-\nicefrac 1K$. Therefore, the global budget balance condition is satisfied because the cumulative profit at $T$ is at least 
            \[
            \beta-(T-\tau)\frac 1 K \ge T^{\nicefrac 34}- T \frac{1}{T^{\nicefrac 14}}=0.
            \]

            Now we prove the regret upper bound. For rounds up to $\tau$ we can exploit~\Cref{lm:erevmax2}. On the other hand, for any $\tau$, on rounds in $\tau+1,\ldots, T$ we have 
            \begin{align}\label{eq:boundGFTsecondTerm2}
            \max_{p \in [0,1]} \sum_{t=\tau+1}^T \gft_t(p) &\le \max\limits_{(p,q)\in H_K} \sum_{t=\tau+1}^T \gft_t(p,q) + \frac TK\\\nonumber
            &\le \subE{(p_t,q_t)\sim \vx_t}{\sum_{t=\tau+1}^T \gft_t(p_t,q_t)}+\frac TK+5T^{\nicefrac 34}\sqrt{\log(T^{\nicefrac 14})}
            \end{align}
            where the first inequality follows from~\Cref{prop:additive_discretization}, and the second follows by \Cref{lm:gftbandit} by replacing $T$ with $T-\tau$.
            Then, assume that the regret bound \Cref{lm:erevmax2} holds, which happens with probability at least $1-\nicefrac 1T$. By summing Equation~\eqref{eq:boundGFTsecondTerm2} and Lemma~\ref{lm:erevmax2} we have that, with probability at least $1- \nicefrac 1T$,
            \begin{align*}
                \max_{p \in [0,1]} &\sum_{t=1}^T \gft_t(p)  \le \max_{p \in [0,1]} \sum_{t=1}^\tau \gft_t(p) +\max_{p \in [0,1]} \sum_{t=\tau+1}^T \gft_t(p) \\
                &\hspace{-0.5cm}\le 8\log T  (\beta +1) + \frac {5T}K+ 256 \log T  \sqrt{|F_K|T\log(|F_K|T)} +\max_{p \in [0,1]} \sum_{t=\tau+1}^T \gft_t(p) \\
                &\hspace{-0.5cm} \le  \E{\sum_{t=\tau+1}^T \gft_t(p_t,q_t)}+ 8\log T  (\beta +1) + \frac {6T}K+ 256 \log T  \sqrt{|F_K|T\log(|F_K|T)}+5T^{\nicefrac34}\sqrt{\log(T^{\nicefrac14})},
            \end{align*}
        where the second inequality follows from Lemma~\ref{lm:erevmax2} and the third one from \Cref{lm:gftbandit}.

        Then, by substituting $\beta = T^{\nicefrac34}$ and $K= T^{\nicefrac14}$ we obtain
        \begin{align*}
            \max_{p \in [0,1]} \sum_{t=1}^T \gft_t(p)& \le  \E{\sum_{t=\tau+1}^T \gft_t(p_t,q_t)}+ 8\log T  (T^{\nicefrac34} +1) +  {6T^{\nicefrac34}}\\
            &\hspace{0.cm}+256\cdot \log T  \sqrt{(2T^{\nicefrac14}(\log  T+1))T\log((2T^{\nicefrac14}(\log  T+1))T)}+5T^{\nicefrac34}\sqrt{\log(T^{\nicefrac14})}\\
            &\le \E{\sum_{t=\tau+1}^T \gft_t(p_t,q_t)}+ 1280 \cdot T^{\nicefrac34}\log^2 T.
        \end{align*}
        Then, by rearranging, with probability at least $1-\nicefrac 1T$ it holds
        \begin{align*}
        \max_{p \in [0,1]} \sum_{t=1}^T \gft_t(p) -\E{\sum_{t=1}^T \gft_t(p_t,q_t)}&\le \max_{p \in [0,1]} \sum_{t=1}^T \gft_t(p) -\E{\sum_{t=\tau+1}^T \gft_t(p_t,q_t)} \\
        & \le   1280 \cdot T ^{\nicefrac34}\log^2 T,
        \end{align*}
        where the first inequality follows from the fact that the gain from trade is always non-negative.

        Finally, the expected regret is at most
        \[
        R_T(\egftmax)\le\left(1-\frac 1 T\right) \left(1280 \cdot T ^{\nicefrac34}\log^2 T\right)+ \frac{1}{T} \cdot 2 T \le 1282 \cdot T ^{\nicefrac34}\log^2 T.
        \]This concludes the proof.
        \end{proof}

    \subsection{\texorpdfstring{$\Omega(T^{\nicefrac 57})$}{T to the 5/7} Lower Bound with Two-Bit Feedback}\label{subsec:partialLower}

    In this section, we provide a lower bound for learning the best price against any oblivious adversary, with global budget balance constraints and two-bit feedback. Our construction builds upon the one by \citet{CesaBianchiCCFL23}, but exhibits two key differences. First, we are not constrained to use smooth value distributions. This allows us to simplify the construction, avoiding the reduction to online learning with feedback graphs. 
    Second, we only require algorithms to be globally budget balanced (instead of per-round weakly budget balanced); looser budget balance constraints enhance the capabilities of the learning algorithm. All in all, we derive a lower bound that is slightly looser $T^{\nicefrac57}\approx T^{0.714}$ compared to the $\Omega(T^{\nicefrac34})$. We further elaborate on this comparison at the end of the Section.

    \begin{theorem}\label{thm:lower2bits}
        Consider the problem of repeated bilateral trade in the two-bit feedback model. Any learning algorithm that satisfies global budget balance suffers regret at least $\Omega(T^{\nicefrac57})$.
    \end{theorem}

    The rest of the Section is devoted to the proof of \Cref{thm:lower2bits}; for the missing details, we refer to \Cref{app:LBbandit}. 
    Our lower bound construction is based on $N$ stochastic sequences of valuations. Each one of these sequences is sampled in an i.i.d. way from distributions of valuations with two key properties: (i) they are close with respect to some statistical measure of distance (see \Cref{lem:KL}) and (ii) ensure that any pair of prices that reveals information on the underlying instance is highly suboptimal in terms of \gft~(\emph{i.e.}, gathering information is ``costly'', see \Cref{lem:explorationiscostly}). We proceed in 5 steps.
    
    \xhdr{i) Building a set of hard instances.} We start by introducing a set of $N=N(T)$, to be specified later, hard instances of the bilateral trade problem. Our goal is to show that any learning algorithm has regret at least $\Omega(T^{\nicefrac 57})$ in at least one of the $N$ instances.
    We define a distribution $\mu_k\in\Delta([0,1]^2)$ of valuations $(s,b)$ over $[0,1]^2$ for each $k\in\{0,\ldots, N-1\}$, where we have $N-1$ ``perturbed'' distributions corresponding to indices $k\in\{1,\ldots, N-1\}$, and a ``base'' distribution corresponding to $k=0$.
 
    Let $\ell=\nicefrac 1{12}$ and let $\Delta\defeq \nicefrac \ell{(N-1)}$, and $\delta= \nicefrac \Delta2$.
    Then, for any instance $k\in\{0,\ldots,N-1\}$, the distributions $\mu_k$ are supported on the same set $\cW$ of finitely many valuations. We describe the set $\cW$ by partitioning it into six different sets. An illustration of the valuations set can be found in \Cref{fig:instance_LB}.
    First, we define the two sets $\cW_1$ and $\cW_2$ (respectively red and blue in \Cref{fig:instance_LB}) as follows:
    \[
    \cW_1\defeq\left\{w^i_1=\left(\tfrac{1-\ell}2+i \Delta, 1-\ell\right):\, i=0,\ldots, N-1\right\}
    \]
    and 
    \[
    \cW_2\defeq\left\{w^i_2=\left(\tfrac{1-l}2+i \Delta, 1-\ell-\rho\right):\, i=0,\ldots, N-1\right\},
    \]
    where $\rho=\nicefrac 1{32}$. 
    These valuations are ``balanced out'' by the $N$ valuations in $\cW_3$ (green in \Cref{fig:instance_LB}):
    \[
    \cW_3\defeq\left\{w_3^i=\left(0, \tfrac{1-\ell}{2}-\delta+i \Delta\right):\, \ i=0,\ldots, N-1\right\}.
    \] Moreover, we have a set $\cW_4$ of ``deficit-generating'' valuations (brown in \Cref{fig:instance_LB})
    \[
    \cW_4\defeq\left\{w_4^i=\left(\tfrac{1-\ell}2+i \Delta, \tfrac{1-\ell}{2}-\delta+i \Delta\right):\, i=0,\ldots, N-1\right\},
    \]
    and a single valuation belonging to $\cW_5$ (orange in \Cref{fig:instance_LB})
    \[
    \cW_5\defeq\left\{\left(0, \tfrac{1-\ell}{2}\right)\right\}.
    \] 
    We conclude by defining the set $\cW_6$ (purple in \Cref{fig:instance_LB}) of the four ``extremal'' valuations (in practise, they are needed for \Cref{lem:KL} to hold):
    \[
    \cW_6\defeq\left\{(0, 0), (0,1), (1,1), (1, 0)\right\}.
    \]

    We assign different probabilities to the valuations in each set $\cW_j$ depending on the instance.
    In particular, for any instance $k\in\{1,\ldots,N-1\}$ with distribution $\mu_k$, we have that
    \begin{equation}\label{eq:LBmu1}
    \mu_k(w^i_j)=\frac{1}{64N^2}=\gamma_1,\quad \forall j\in\{1,2\}, i\notin \{k,k+1\},
    \end{equation}
    while we perturb by $\eps$ the probability of the following valuations:
    \begin{align}\label{eq:LBmu2}
    \mu_k(w^k_1)=\gamma_1+\eps,\quad \mu_k(w^{k+1}_1)=\gamma_1-\eps,\quad
    \mu_k(w^k_2)=\gamma_1-\eps,\quad
    \mu_k(w^{k+1}_2)=\gamma_1+\eps.
    \end{align}
    Conversely, for the base instance $\mu_0$, we place equal probability $\mu_0(w)=\gamma_1$ on all the valuations $w\in\cW_1\cup\cW_2$, and hence all these valuations have the same probability.
    For each instance $k \in\{0,\ldots, N-1\}$ with distribution $\mu_k$, the probability of valuations $w_3^i$, with $i\in\{0,\ldots, N-1\}$, is set as
    \[
    \mu_k(w_3^i) = \gamma_1\cdot\frac{1-\ell-\rho-2i \Delta}{\frac{1-\ell}{2}-\delta+i \Delta}\in\left(0,2\gamma_1\right).
    \]
    
    Let $\gamma_3^{\textnormal{tot}}\defeq\sum_{w\in \cW_3}\mu_k(w)<2\gamma_1N$ be the total probability assigned to valuations in $\cW_3$.
    Moreover, for any instance $k\in\{1,\ldots,N-1\}$ with distribution $\mu_k$, we assign to every point in $\cW_4$ probability $\gamma_4\defeq4\gamma_1(13N-14)$, \ie
    \begin{equation}\label{eq:LBmu3}
        \mu_k(w)=4\gamma_1(13N-14), \quad\forall w\in \cW_4.
    \end{equation}
    Then, for any instance $k\in\{1,\ldots,N-1\}$ with distribution $\mu_k$, we assign probability $\gamma_5$ to the single valuation in $\cW_5$, \ie $\mu_k(0, \nicefrac {(1-\ell)}2)=\gamma_5\defeq \nicefrac 1{64}$.
    Finally, all the remaining probability is equally divided into the $4$ extremal points in $\cW_6$, \ie
    \[
    \mu_k(w)=\frac{1-\left(2\gamma_1N+\gamma_3^{\textnormal{tot}} + 4\gamma_1N(13N-14) + \gamma_5\right)}4=\gamma_6,\quad\forall w\in \cW_6.
    \]

    In \Cref{app:LBbandit}, \Cref{lem:welldefinedmuk} shows that this probabilities are positive and,  therefore, $\mu_k$ defines a probability distribution for every $k$.

    Now, we define $\cG_\cW$ as the grid generated by such valuations. Formally:
    \[\cG^s_\cW=\{s\,:\,\exists\, (s,\cdot)\in\cW\},\,\cG^b_\cW=\{b\,:\,\exists\, (\cdot,b)\in\cW\},\quad\text{and}\quad 
    \cG_\cW\defeq\mleft\{(s,b)\,:\,s \in \cG^s_\cW \text{ and } b\in\cG^b_\cW\mright\}.
    \]
    Thus, $\cG^s_\cW$ and $\cG^b_\cW$ represent the projections of $\cG_\cW$ onto its first (seller) and second (buyer) component, respectively.
    \begin{figure} \label{fig:plat}
        \begin{subfigure}[t]{.5\linewidth}
            \centering\scalebox{.6}{\includegraphics{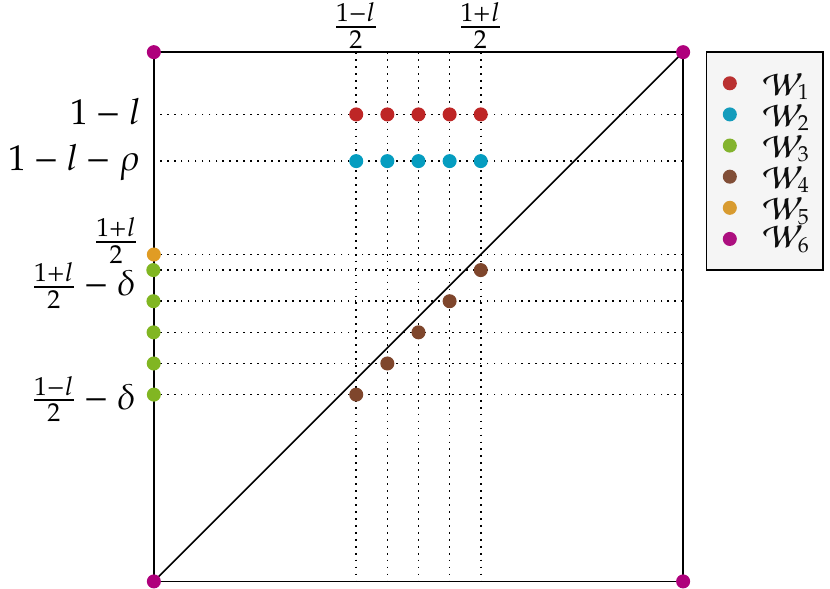}}
            \caption{}
            \label{fig:instance_LB}
            \end{subfigure}
            \begin{subfigure}[t]{.5\linewidth}
            \centering\scalebox{.6}{\includegraphics{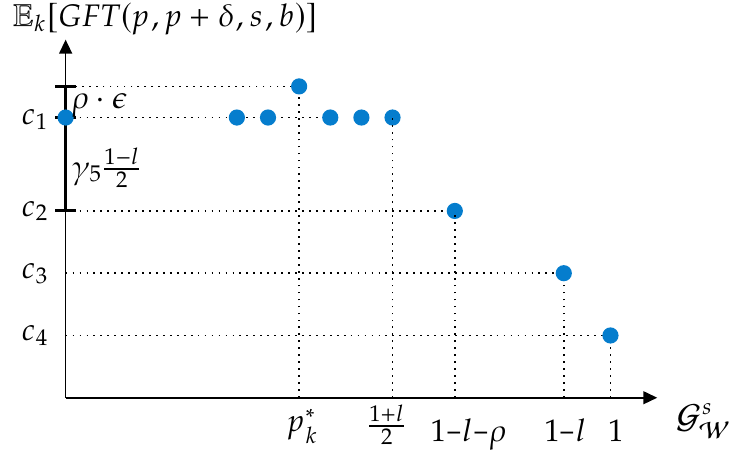}}
            \caption{}
            \label{fig:expectedGFT}
      \end{subfigure}
      \caption{\Cref{fig:instance_LB} represents the valuations support $\cW$ of the instances distributions $\mu_k$, while \Cref{fig:expectedGFT} represents the value of posting the same price to the seller and the buyer in instance $\mu_k$.}
  	\end{figure}

    \xhdr{ii) Analysis of the gain from trade.}
    As a first step, we argue that we can focus on algorithms that play only actions in $\cG_\cW$, without loss of generality. 
    Consider infact any instance $k\in\{1,\ldots, N\}$ and any algorithm $\cA$. Similarly to the proof of~\Cref{pr:benchmark_2} (more specifically \Cref{cl:discrete} therein), one can easily prove that there exists an equivalent algorithm $\cA'$ (in terms of both feedback, \gft, and profit), that only has distribution supported on the grid $\cG_\cW$ generated by the valuations $\cW$. 
    
    Next, for any $p\in\cG^s_\cW$, we characterize the value of posting the pair of prices $(p,p+\delta)$ under distribution $\mu_k$, with $k \in \{0,\ldots,N-1\}$. 
    Note that posting the pair $(p,p+\delta) \in \cG_\cW$ under any instance $\mu_k$, is equivalent to posting a single price $p\in \cG_\cW^s$ to both the seller and the buyer, with the only difference that $(p,p)\notin\cG_\cW$, while $(p,p+\delta)\in\cG_\cW$.
    Then, for any $p\in\cG_\cW^s$, we relate the \gft~obtained by posting a pair $(p,p+\delta)$ under valuations sampled from $\mu_k$, with $k\in\{1,\ldots, N-1\}$, and under the base distribution $\mu_0$.
    For every $k\in\{0,\ldots, N-1\}$, let $\mathbb{E}_k$ and $\mathbb{P}_k$ denote the expectation and the probability measure under instance $\mu_k$, respectively.
    Direct calculations shows that, for all $p\in \cG_\cW^s$ and $k\in\{1,\ldots, N-1\}$, it holds
    \[
    \mathbb{E}_k[\gft(p, p+\delta,s,b)] = \mathbb{E}_0[\gft(p,p+\delta,s,b)] + \rho\eps\ind{p=p_k^*},
    \]
    where $\gft(p, p+\delta,s,b)$ is simply the gain from trade when the prices posted are $(p,p+\delta)$ and valuations $(s,b)$, and $p^*_k=\frac{1-\ell}2+k \Delta$.
    Moreover, for all $p\in\cG_\cW^s$ it holds:
    \[
    \mathbb{E}_0[\gft(p, p+\delta,s,b)] = 
    \begin{cases}
        c_1\defeq\gamma_5\frac{1+\ell}{2}+\mu_0(0,1)+\gamma_1\frac{77}{96}N=&\text{if } p\in[0,\frac{1+\ell}2]\\
        c_2\defeq\mu_0(0,1)+\gamma_1\frac{77}{96}N&\text{if } p\in(\frac{1+\ell}2, 1-\ell-c]\\
        c_3\defeq\mu_0(0,1)+\gamma_1\frac{5}{12}N&\text{if } p\in(1-\ell-c, 1-\ell]\\
        c_4\defeq\mu_0(0,1)&\text{if } p\in(1-\ell, 1]\\
    \end{cases}
    \]

    \Cref{fig:expectedGFT} gives a representation of $\mathbb{E}_k[\gft(p, p+\delta,s,b)]$. 
    From these calculations, we show that in an instance $k\in\{1,\ldots, N-1\}$ the pair that maximizes the expected gain from trade is $(p^*_k,p^*_k+\delta)$.

    \begin{lemma}\label{lem:LBlem1}
        For any instance $k\in\{1,\ldots, N-1\}$, we have that:
        \[
        \max\limits_{(p,q)\in[0,1]^2, \,p\le q}\mathbb{E}_k[\gft(p,q,s,b)] = \mathbb{E}_k[\gft(p^*_k,p^*_k+\delta,s,b)]=c_1+\rho\cdot \eps.
        \]
    \end{lemma}

    The previous lemma characterizes the optimal fixed budget balanced price.
    Then, we show that all the strategies that are \emph{not} budget balanced are dominated.
    Indeed, one of the main challenges of our reduction is that, in general, a globally budget balanced algorithm could get a larger \gft~by temporarily sacrificing some profit and posting prices $(p,q)$ with $q<p$.
    In the following lemma we show that our instances are built in such a way that these strategies are dominated and thus can be discarded. Intuitively, every tuple of prices $p,q$ that tries to gain higher \gft~than the one obtained by playing on the diagonal  must win also trades in $\cW_4$. Then, since trades in $\cW_4$ have negative \gft~and happen with sufficiently high probability $\gamma_4$, we have that posting prices $q<p$ is dominated.

    \begin{restatable}{lemma}{lowertriisdom}\label{lem:lowertriisdom}
        For every pair of posted prices $(p,q)\in\cG_\cW\cap\{(p,q)\in[0,1]^2\,|\,p< q\}$, $(p',q')\in\cG_\cW\cap\{(p,q)\in[0,1]^2\,|\,p\geq q\}$, and instance $k\in\{0,\ldots, N-1\}$, we have that
        \[
        \mathbb{E}_k[\gft(p,q,s,b)]\le \mathbb{E}_k[\gft(p',q',s,b)]\le c_1+\rho\,\eps\,\ind{(p',q')=(p_k^*, p_k^*+\delta)}.
        \]
    \end{restatable}

    We complete this section by showing that also strategies that propose a high price to the buyer are dominated in every instance.
    In particular, we show that when the algorithm places prices $(p,q)$ with $q>\nicefrac{(1+\ell)}2$, it looses a constant \gft~with respect to choosing a smaller $q$. This is because the learner cannot induce the trade $\cW_5$ which guarantees expected \gft~of $\Theta(\gamma_5)$. Formally,

    \begin{lemma}\label{lem:explorationiscostly}
        For any instance $k\in\{0,\ldots, N-1\}$, price $p\in \cW_{\cG}^s\cap \mleft[\frac{1-\ell}{2},\frac{1+\ell}{2}\mright]$, and price $q\in\mleft(\tfrac{1+\ell}{2},1\mright]\cap\cG_\cW^b$ we have that
        \[
        \mathbb{E}_k[\gft(p,p+\delta)]\ge \mathbb{E}_k[\gft(p,q)] + \frac{\gamma_5}{3}.
        \]
    \end{lemma}
    Intuitively, the previous lemma shows that exploring is costly. Indeed, as we show in the following paragraph, the algorithm must post $q\ge \nicefrac{(1+\ell)}2$ to gain information on the instance, \ie on the $k$ that determines the instance.
    
    \xhdr{iii) Analysis of the feedback.}
    In the two-bit feedback model, for a valuation $(s,b)$ we have that posting prices $(p,q)$ generates the feedback $(\ind{s\le p}, \ind{q\le b})$.
    Now, we show that for any instance $\mu_k$ and any posted prices $(p,q)$, the distribution of the feedback is independent on the instance almost everywhere. Specifically, the feedback distribution depends on the instance $k$ only within a ``small'' and instance-dependent region of prices.
    For every instance $k\in\{1,\ldots, N-1\}$, let 
        \[
        \cF_k\defeq\left[\tfrac{1-\ell}{2}+(k-1) \Delta, \tfrac{1-\ell}{2}+k \Delta\right)\times (1-\ell-c, 1-\ell].
        \] 
    It is a simple exercise to see that, for each pair of prices outside the sets $\cF_k$, the feedback received by the learner is independent of the specific instance that is generating the valuations  (see~\citep[Claim~2]{CesaBianchiCCFL23} for a similar result).
    \begin{lemma}\label{lem:same_FB}
        For all $(p,q)\in[0,1]^2\setminus \bigcup_{k'\in\{1,\ldots, N-1\}}\cF_{k'}$ it holds:
        \[
        \mathbb{P}_k\mleft[(\ind{s\le p}, \ind{q\le b})=z\mright]=\mathbb{P}_j\mleft[(\ind{s\le p}, \ind{q\le b})=z\mright],\quad \forall z\in\{0,1\}^2, \forall j,k\in\{0,\ldots, N-1\}.
        \]
    \end{lemma}

    \xhdr{iv) Price regions.}
    The properties uncovered so far naturally partition the square $[0,1]^2$ into the following three regions:
    \begin{itemize}
     \item \textcolor{black}{\emph{Exploration regions}}. We have the $N-1$ regions $\cF_k$. These are the regions in which the probability of observing a certain two-bit feedback depends on the instance $\mu_k$ from which the valuations are sampled.
    
    \item \textcolor{black}{\emph{Exploitation regions}}. We define the regions $\cE_k$ for any $k\in\{1,\ldots,N-1\}$ as follows
    \[
    \cE_k\defeq\left\{(p,q)\in [0,1]^2\,\,\Big\vert\,\, q\ge p,\,\, q\le \tfrac{1+\ell}{2},\,\, p\in\left[\tfrac{1-\ell}{2}+(k-1)\Delta, \tfrac{1-\ell}{2}+k\Delta\right)\right\}.
    \]
    All these regions are such that the \gft~collected by posting $(p,q)\in \cE_k$ is close (and smaller than or equal to) to the optimal \gft, \ie the one obtained by posting $(p_k^*, p_k^*+\delta)$. 
    
    \item \textcolor{black}{\emph{Dominated regions}}. We define $\cD$ as the remaining set of possible valuations, that is
    \[
    \cD\defeq[0,1]^2\setminus\left(\cup_k (\cF_k\cup\cE_k)\right).
    \] 
    It's easy to verify that by posting $(p,q)\in \cD$ one obtains a \gft~that is at most $c_1$.
    \end{itemize}

    \begin{figure}[t!]
        \scalebox{0.6}{\includegraphics{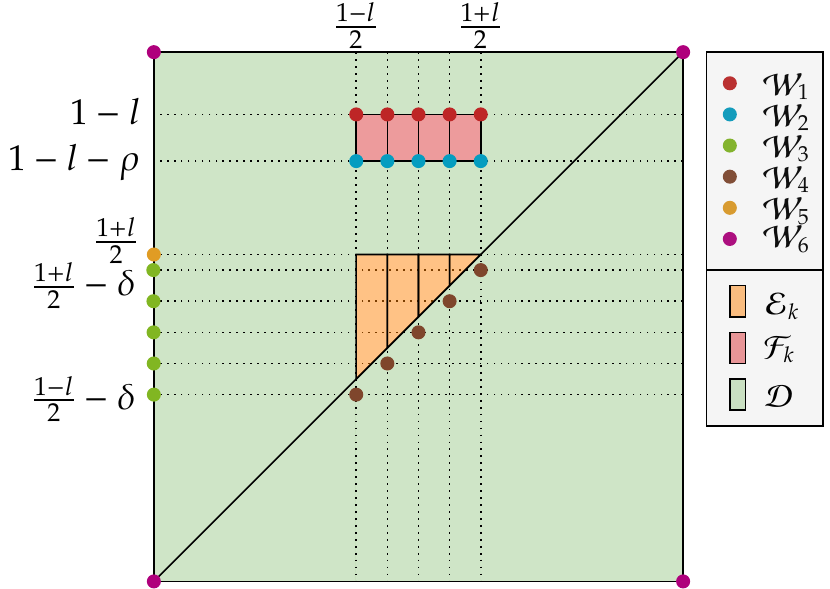}}
        \vspace{0.2cm}
        \centering
        \caption{Partition of $[0,1]^2$ in exploration regions $\cF_k$, exploitation regions $\cE_k$, and dominated regions $\cD$.}
        \label{fig:partitionLB}
    \end{figure}

    \Cref{fig:partitionLB} shows the partition of the square $[0,1]^2$ into exploration, exploitation and dominated figures, which are depicted in \textcolor{Red2}{\textbf{red}}, \textcolor{orange2}{\textbf{orange}} and \textcolor{mgreen}{\textbf{green}}, respectively. Next, we define 
    \begin{align*}
    \textstyle{\cN_k\defeq\sum_{t\in[T]}\ind{(p_t,q_t)\in\cF_k},\quad
    \cM_k\defeq\sum_{t\in[T]}\ind{(p_t,q_t)\in\cE_k},\quad
    \cO\defeq\sum_{t\in[T]}\ind{(p_t,q_t)\in\cD},}
    \end{align*}
    which are the number of times an algorithm plays in the exploration, exploitation and dominated regions, respectively.
    Then, we can upper bound the gain from trade of an algorithm $\cA$ considering only the number of times $\cA$ plays in each region.
    In particular, it holds that in any instance $k$:
    \begin{itemize}
        \item \textbf{Cost of exploration:} the \gft~collected by posting prices in $\cF_j$ is at most $c_2$ for all $j$ (\Cref{lem:explorationiscostly});
        \item \textbf{Exploitation:} the \gft~collected by posting prices in $\mathcal{E}_j$ is at most $c_1+\rho\cdot\eps\ind{j=k}$ (\Cref{lem:lowertriisdom});
        \item \textbf{Cost of domination:} the \gft~collected by posting prices in $\cD$ is at most $c_1$ (\Cref{lem:lowertriisdom}). 
        
    \end{itemize}
    Formally, these observations lead to the following upper bound.

    \begin{lemma}\label{lem:upperboundGFTLB}
        Let $\{(p_t,q_t)\}_{t\in[T]}$ be the sequences of prices posted by any algorithm $\cA$. Then
        \[
        \sum\limits_{t=1}^T\mathbb{E}_k[\gft(p_t, q_t, s, b)]\le \mathbb{E}_k\left[\rho\eps\cdot\cM_k+\sum_{k=1}^{N-1}\left(c_1\cM_j+c_2\cN_j +c_1\cO \right)\right].
        \]
    \end{lemma}

    \xhdr{v) Relating the algorithm behavior on different instances.}
    Now we relate the expected number of exploitation rounds $\cM_k$ in different instances $k$. This difference depends on the probability measures $\mathbb{P}_k$ and $\mathbb{P}_0$ through the Pinsker's inequality on a suitably defined multinomial random variable that encodes the four possible feedback observed when playing in the exploration regions $\cE_k$.

    \begin{restatable}{lemma}{pinsker}\label{lem:KL}
    For all $k\in\{1,\ldots,N-1\}$ we have that
    \[
    \mathbb{E}_k[\cM_k] - \mathbb{E}_0[\cM_k]\le T\eps\sqrt{\frac{2}{\gamma_6}\mathbb{E}_0[\cN_k]}.
    \]
    \end{restatable}
    
\xhdr{vi) Lower bounding the regret.}

We define the expected regret of an algorithm under instance $k\in\{0,\ldots,N-1\}$ as:
\[
R_T^k\defeq\max\limits_{(p,q)\in[0,1]^2, p\ge q}\mathbb{E}_k\left[\sum\limits_{t=1}^T\gft_t(p,q) - \sum\limits_{t=1}^T\gft_t(p_t,q_t)\right].
\]
Then, combining all the previous results leads to the following lemma which gives a lower bound in terms of $\eps,N$, and $T$.

\begin{restatable}{lemma}{lemlowerboundbandit}\label{lem:lowerboundbandit}
There is an instance $k\in\{0,\ldots,N-1\}$ and an absolute constant $c\in(0,1)$ such that:
\[
R_T^k\ge c\cdot\min\left(\tfrac{N}{\eps^2}, {\eps} T\right).
\]
\end{restatable}

    By using \Cref{lem:lowerboundbandit} we can readily conclude the proof of \Cref{thm:lower2bits} as follows.
    Let $\eps=T^{-\alpha}$ and $N=T^\beta$, with $\alpha,\beta>0$.
    Now we simply have to optimize over the choice of parameters $\alpha$ and $\beta$.
    In doing so, we need to take into account the additional constraints necessary to have well-defined instance distributions $\mu_k$. 
    In particular, we have that $\eps\le\gamma_1$ from \Cref{eq:LBmu2}, and $2N\gamma_1<1$ from \Cref{eq:LBmu1}. Moreover, we also need to impose $\gamma_4N<1$ by \Cref{eq:LBmu3}. Since $\gamma_4=4\gamma_1(14N-13)$, this also implies that $\gamma_1<\nicefrac{1}{4N(13N-14)}<\nicefrac{1}{N^2}$ for  $N>2$.
    Therefore, the  constraint  $\eps<\gamma_1$ implies:
    $\eps=T^{-\alpha}\le \nicefrac{1}{T^{2\beta}}=\nicefrac{1}{N^2}$
    which yields that $\alpha\ge2\beta$. Note that this dominates the constraint $\eps<\nicefrac 1N$ (or equivalently written as $\alpha\ge\beta$) that would have been implied by \Cref{eq:LBmu2} alone.
    
    The lower bound of is maximized when $\alpha$ and $\beta$ are solution of the following program:
    \begin{align*}
        &\max\limits_{\alpha\ge0}\ (1-\alpha)\\
        &\textnormal{s.t.}\quad\alpha\ge2\beta\,,\,\, 1-\alpha=\beta+2\alpha
    \end{align*}
    which gives as solution the values $\alpha= \nicefrac 27$ and $\beta=\nicefrac 17$. This implies a $\Omega(T^{\nicefrac 57})$ lower bound.

    \paragraph{Connection with the $\Omega(T^{\nicefrac 34})$ lower bound of~\citet{CesaBianchiCCFL23}.}
    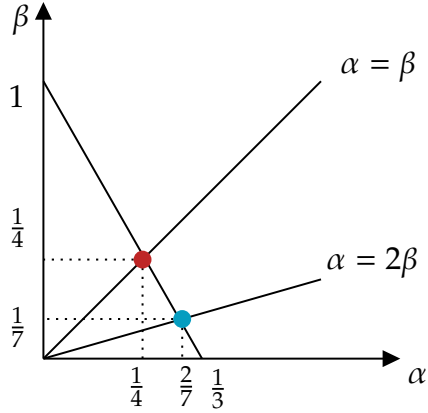
\begin{figure}
        \centering
        \input{LBalpha}
        \caption{Order of $\alpha$ and $\beta$ reachable by~\citet{CesaBianchiCCFL23} (\textcolor{niceRed}{\textbf{red}}) and this work (\textcolor{blueGrotto}{\textbf{blue}}).}
        \label{fig:LBalpha}
    \end{figure}

    While our result and the one of \citet{CesaBianchiCCFL23} build on a similar constructions (at least conceptually), we obtain a weaker lower bound.
   The main reason is that the learner in~\citet{CesaBianchiCCFL23} is weak budget balanced, while in our work the learner has only a global budget balance constraint. To preclude this option to the learner,  we penalize the \gft~of prices in the lower triangle by adding the set of valuations $\cW_4$. If $\gamma_4$ is large enough w.r.t. $\gamma_1$, then posting prices in the lower triangle is dominated. In particular, we must choose $\gamma_4=\Theta(\gamma_1 N)$ as we prove in \Cref{lem:lowertriisdom}. Once we prove that the lower triangle is dominated, we can conceptually reduce our problem to the one of \citet{CesaBianchiCCFL23}. However, the choice of $\gamma_4=\Theta(\gamma_1 N)$ imposes the additional constraint $\alpha\ge 2\beta$, which is not needed in the original construction.
    Hence, they can set $\alpha=\beta=\nicefrac14$, and get a bound of $\Omega(T^{\nicefrac 34})$. This difference is depicted in \Cref{fig:LBalpha}.

\section{Best Feasible Distribution of Prices}\label{sec:best randomized prices}
    
    In this section, we analyse the regret with respect to the best fixed distribution over prices which satisfies global budget balance {\em on average}. First, we present a negative result that clearly separates this new benchmark from the best fixed price in hindsight: in \Cref{thm:alpha_lower}, we prove that it is impossible to achieve sublinear $(1+\varepsilon)$-regret with respect to the best feasible distribution, even in the full feedback setting. On the positive side, we show that the two benchmarks are only a multiplicative factor $2$ apart (\Cref{thm:2_upper}). This implies that any learning algorithm that exhibits sublinear regret with respect to the best fixed price in hindsight automatically achieves sublinear $2$-regret with respect to the best feasible distribution. Finally, we complement this positive result by proving that this multiplicative gap of $2$ is tight (\Cref{thm:2_lower}).

    \subsection{Linear Lower Bound}
        The best feasible distribution has a crucial advantage with respect to any budget balanced learner: it has the possibility to ``run some deficit'' in a preliminary phase of the sequence as it knows it will be possible to extract enough profit to ensure global budget balance in some later stages. 
        For instance, consider a half-sequence where $(s_t,b_t)$ is either $(0,\nicefrac 13)$ or $(\nicefrac 23, 1)$, for $t \le \nicefrac T2$. Any learning algorithm has to enforce budget balance at time $\nicefrac T2$ (to be protected about the possibility that $(s_t,b_t) = 0$ for all future $t$), while the randomized benchmark, which knows the future, may run a deficit and collect more gain from trade by posting the budget unbalanced prices $(\nicefrac 23, \nicefrac 13)$ with some probability. Inspired by this example, we state the following Lemma.

        \begin{lemma}
            \label{lemma:S_1}
                For any algorithm $\cA$ that enforces global budget balance, there exists a deterministic sequence of valuations $\cS_1$ with the following properties: $(i)$ the expected gain from trade of $\cA$ is at most $\nicefrac T9$; $(ii)$ the valuations $(s_t,b_t)$ are either $(0,\nicefrac 13)$ or $(\nicefrac 23,1)$ for all $t \le\nicefrac T2$; $(iii)$ the valuations $(s_t,b_t)$ are equal to $(0,0)$ for all $t > \nicefrac T2$.
        \end{lemma}
               \begin{proof}
                Consider the following randomized instance: $(s_t,b_t)= (0,0)$ for $t > \nicefrac T2$, while for the other time steps the valuations are either $(0,\nicefrac 13)$ or $(\nicefrac 23, 1)$, independently and uniformly at random. Each realized instance of such randomized sequence clearly satisfies requirements $(ii)$ and $(iii)$. Finally, we show that in expectation (with respect to the randomization of the algorithm and of the instance), the total gain from trade of $\cA$ is at most $\nicefrac T9.$ Then the existence of an instance $\cS_1$ with the desired properties follows by an averaging argument.

                Focus on the first $\nicefrac T2$ time steps, and let $N_1$ be the random variable that counts the number of time steps in which $\cA$ posts prices $q_t\le 1/3$ and $p_t\ge 2/3$. Moreover, let $N_2 \defeq \nicefrac T2 - N_1$, and $n_i = \E{N_i}$ for $i\in\{1,2\}$. By assumption, Algorithm $\cA$ is global budget balanced, which this means that 
                \[
                    - \frac {n_1}3 +  \frac {n_2}6 \ge 0.
                \]
                The first term of the inequality follows from the fact that every time the learner posts $q_t\le 1/3$ and $p_t\ge 2/3$, it loses at least $1/3$ revenue. The second term follows from the fact that, by posting other pairs of prices, the learner can extract at most a revenue of $1/6$ (\ie the trade happens with probability $1/2$, and the learner receives $1/3$ revenue). On the other hand, $n_1$ is directly proportional to the final gain from trade of $\cA$, thus the best possible gain from trade is achieved for $n_1 = T/6$ and $n_2 = T/3$, which yields an expected gain from trade of at least
                \[
                    \frac{n_1}3 + \frac{n_2}6 = \frac T{9}.
                    \]
                This concludes the proof of the claim.
\end{proof}

        The lemma is crucial in proving the impossibility result in the following Theorem, which holds even under full feedback. 
        \begin{theorem}\label{thm:alpha_lower}
            Fix any constant $\alpha \in [1,\nicefrac {36}{35})$, and any globally budget balanced learning algorithm $\cA$ with full-feedback. Then there exists a sequence of valuations such that
            \[
                \sum_{t=1}^T \mathop{\mathbb{E}}_{(p,q)\sim\gamma^*}{\gft_t(p,q)} - \alpha \cdot \sum_{t=1}^T \E{\gft_t(p_t,q_t)}  \ge \tfrac 5{18} \left(\tfrac{36}{35} - \alpha  \right)T,  
            \]
            where distribution $\gamma^*$ is the optimal feasible distribution.
        \end{theorem}
        \begin{proof}
            Fix any $\alpha \in [1,\nicefrac {36}{35})$ and any learning algorithm $\cA$. Starting from sequence $\cS_1$ as in \Cref{lemma:S_1}, construct a second sequence of valuations $\cS_2$ which coincides with $\cS_1$ for the first half of the time horizon. In the second half we set $(s_t,b_t) = (\hat s, \hat b)$ for all $t > \nicefrac T2$, where $(\hat s, \hat b)$ is the most frequent value in the first half of $\cS_1$. We compare the total gain from trade collected by $\cA$ on $\cS_2$ with that of the best fixed distribution of prices over $\cS_2$, whose gain from trade we denote $\OPT$. The expected gain from trade of $\cA$ on $\cS_2$ is at most $\nicefrac T9$ in the first half (Claim \ref{lemma:S_1}) and $\nicefrac T6$ in the second half (as it can extract at most $\nicefrac 13$ gain from trade in each one of $\nicefrac T2$ time steps). Therefore, $\cA$ extracts at most a total of $\nicefrac{5T}{18}$ expected gain from trade. On the other hand, the best feasible distribution $\gamma^*$  must perform at least as well as the feasible distribution $\gamma$, under which the prices are $(\hat s, \hat b)$ with probability $\nicefrac 47$, and $(\nicefrac23,\nicefrac 13)$ with the remaining probability. First, we argue that distribution $\gamma$ is indeed budget-feasible:
            \begin{align*}
                \sum_{t=1}^T \subE{(p,q)\sim\gamma}{\prof_t(p,q)} \ge \tfrac 37 \mleft(-\tfrac {T}3\mright) + \tfrac 47 \mleft(\tfrac 13 \cdot \tfrac {3T}4\mright) = 0,
            \end{align*}
            where we used that prices $(\nicefrac 23,\nicefrac 13)$ are posted with probability $\nicefrac 37$ and always induces a negative profit of $\nicefrac 13$ and that, by construction, there are at least $\nicefrac{3T}4$ time steps where $(s_t,b_t) = (\hat s, \hat b)$.
            To conclude the proof, we analyze in a similar way the total gain from trade achieved by $\gamma$: 
            \begin{align*}
                \sum_{t=1}^T\subE{(p,q)\sim\gamma}{\gft_t(p,q)} \ge \tfrac 37 \mleft(\tfrac {T}3\mright) + \tfrac 47 \mleft(\tfrac 13 \cdot \tfrac {3T}4\mright) =  \tfrac 27 T.    
            \end{align*}
            All in all, we have constructed an instance, $\cS_2$ where $\cA$ exhibits an expected gain from trade of at most $\nicefrac{5T}{18}$, while $\OPT$ is at least $\nicefrac{2T}7$. This means that $\cA$ suffers at least the following $\alpha$-regret
            \begin{align*}
                \OPT - \alpha \cdot \sum_{t=1}^T \E{\gft_t(p_t,q_t)} & \ge \sum_{t=1}^T\subE{(p,q)\sim\gamma}{\gft_t(p,q)}    - \alpha \cdot \sum_{t=1}^T \E{\gft_t(p_t,q_t)}\\
                &\ge \left(\tfrac 27 - \alpha \tfrac 5{18} \right) T= \tfrac 5 {18} \left(\tfrac {36}{35} - \alpha \right) T. \qedhere
            \end{align*}            
        \end{proof}

    \subsection{Comparison of the Two Benchmarks}
    Surprisingly, it holds that the performance of the optimal fixed price is to not far from that of optimal global budget balanced distribution.
    \begin{theorem}
        \label{thm:2_upper}
        Denote with $p^*$, resp. $\gamma^\ast$, the best fixed price, resp. the best feasible distribution. Then, for any sequence of valuations: 
        \[
            \sum_{t=1}^T \mathop{\mathbb{E}}_{(p,q)\sim \gamma^\ast}{\gft_t(p,q)} \le 2 \sum_{t=1}^T \gft_t(p^\ast).      
        \]
    \end{theorem}
    \begin{proof}
        Fix any sequence of valuations $\cS$, and let $p^*$ be as in the statement. By standard analytic arguments, it is possible to show that there exists an optimal feasible distribution $\gamma^*$ whose support is either one or two points (we refer to \Cref{pr:benchmark_2} in \Cref{app:model} for a formal proof). We prove the result using this $\gamma^*$ and considering two separate cases, according to the cardinality of the support of $\gamma^*$.If the support of $\gamma^*$ consists of only one point $(p,q)$, and since $\gamma^\ast$ has respect budget feasibility, then it is safe to assume without loss of generality that such point lies above the diagonal, \ie $p\le q$ and that the gain from trade achieved by $\gamma^*$ is exactly the same provided by $p^*$.

        In the second case, the support of $\gamma^*$ consists of two different points $(p_1, q_1)$ and $(p_2,q_2)$. If both prices lie in the upper left diagonal (i.e., $p_1 \le q_1$ and $p_2 \le q_2$), then the total gain from trade is exactly the same as $p^*$, by maximality of $p^*$. If one of the two pair of prices is strongly budget balance, let's say $p_1 = q_1$ and $q_2 < p_1$, then the only possibility (by the budget balance condition) is that these prices never incur in negative profit, so that their gain from trade is once again at most that of $p^*$. All in all, the only meaningful case to study is when $p_1<q_1$ and $p_2>q_2$. 
        
        Consider then this case, i.e., $p_1<q_1$ and $p_2>q_2$, let $\cT_{0}$ be the set of time steps in which the trade is lost by $(p_2, q_2)$, that is $\cT_0=\{t\in[T]\,\vert\, s_t > p_2 \text{ or } b_t<q_2\}.$ For all other $t\in [T]\setminus \cT_0$, every prices $(p_2, q_2)$ make the trade happen. We further partition these time steps as follows:
        \begin{align*}
            \cT_1=\{t:\,(s_t,b_t)\in[0,p_2]\times(p_2,1]\},\,
            \cT_2=\{t:\, (s_t,b_t)\in[0,q_2)\times[q_2,p_2]\},
            \cT_3=\{t:\, (s_t,b_t)\in[q_2,p_2]^2\}.
        \end{align*}
        The sets $\cT_0, \dots, \cT_3$ partition the time horizon.
        Now, for each one of these subset of time steps $\mathcal{T}_i$ it is possible to define two functions over $[0,1]^2$: 
        \begin{align*}\label{eq:f_and_g}
            f_i(p,q) = \sum\limits_{t\in\mathcal{T}_i}\gft_t(p,q),\quad g_i(p,q)=\sum\limits_{t\in\mathcal{T}_i}\prof_t(p,q).
        \end{align*}
        We adopt the usual convention to omit the second argument if it coincides with the first one. Clearly, the sum of the $f_i$ yields the total $\gft$, while that of the total $g_i$ the $\prof$. We relate the value of functions $f_0, f_1, f_2$ in $(p_2,q_2)$ with the total gain from trade it collects. The trades in $ \cT_0$ are lost by $(p_2,q_2)$, so it holds that $f_0(p_2,q_2)=0$ and $g_0(p_2,q_2)=0$. For $\cT_1$ and $\cT_2$ these simple bounds hold:
        \begin{equation}
            \label{eq:T_1 T_2}
            f_1(p_2, q_2) + f_2(p_2, q_2) \le f_1(p_2) + f_2(q_2) \le 2 \sum_{t=1}^T \gft_t(p^*). 
        \end{equation}
        We move our attention to $f_3$, where a more sophisticated argument is needed. As a preliminary step, we prove that the profit extracted by $(p_1,q_1)$ is at most the optimal gain from trade:
        \begin{align}
        \nonumber
            \sum_{t=1}^T \prof_t(p_1, q_1)&=(q_1-p_1)\sum\limits_{t=1}^T\ind{s_t\le p_1}\ind{q_1\le b_t}\\
            &\le\sum\limits_{t=1}^T(b_t-s_t)\ind{s_t\le p_1}\ind{q_1\le b_t}\nonumber\\
            \label{eq:aux_1}&\le\sum\limits_{t=1}^T\gft_t(p_1) \le \sum_{t=1}^T \gft_t(p^*).
        \end{align}
        Let $\pi_1$, respectively $\pi_2$, be the probability with which $\gamma^\ast$ draws $(p_1,q_1)$, respectively $(p_2,q_2)$ we have:
        \begin{align} \label{eq:T_3}
        f_3(p_2,q_2) & \le - g_3(p_2,q_2) 
         \le - \sum_{t=1}^T \prof_t(p_2, q_2)\le \tfrac{\pi_1}{\pi_2} \sum_{t=1}^T \prof_t(p_1, q_1) \le \tfrac{\pi_1}{\pi_2} \sum_{t=1}^T \gft_t(p^*),
        \end{align}
    where the first inequality follows by the definition of $\cT_3$, the second by the fact that the only negative profit by posting $(p_2,q_2)$ comes from $\cT_3$, the third by global budget balance of $\gamma$, and the last one by \Cref{eq:aux_1}.
    We finally have all the ingredients to conclude the proof:
    \begin{align*}
        \sum_{t=1}^T \subE{(p,q)\sim\gamma^*}{\gft_t(p,q)}&=\pi_1\sum_{t=1}^T \gft_t(p_1, q_1)+\pi_2 \sum_{i=0,\dots,3}f_i(p_2, q_2) \tag{linearity of exp.}\\
        &\le \pi_1 \sum_{t=1}^T \gft_t(p_1, q_1)+(\pi_1+2\pi_2) \sum_{t=1}^T \gft_t(p^*) \tag{by Eq. \ref{eq:T_1 T_2} and \ref{eq:T_3}}\\
        &\le 2\sum_{t=1}^T \gft_t(p^*),
    \end{align*}
    where the last inequality follows by optimality of $p^*$ with respect to the budget balanced prices $(p_1,q_1)$ and using that $\pi_1 + \pi_2 =1$.
\end{proof}

As a corollary, we have that any algorithm that achieves sublinear regret with respect to the best fixed price also guarantees sublinear 2-regret with respect to the best feasible prices distribution.

\begin{corollary}  
\label{cor:2-regret}
    Let $\cA$ be a learning algorithm for the repeated bilateral trade problem which guarantees an upper bound of $f(T)$ on the regret with respect to the best fixed price in hindsight. Then, the $2$-regret of $\mathcal{A}$ with respect to the best budget feasible distribution over prices is at most $f(T)$.
\end{corollary}

Surprisingly, the factor $2$ between the two benchmarks is optimal. This implies that the analysis of the performance of the algorithms in \Cref{cor:2-regret} is essentially tight.

    \begin{theorem}
        \label{thm:2_lower}
        For any $\eps >0$, there exists a sequence of valuations such that
        \[
            \sum_{t=1}^T \mathbb{E}_{(p,q)\sim\gamma^*}{\gft_t(p,q)} \ge (2-\eps) \sum_{t=1}^T \gft_t(p^*),      
        \]
        where $p^*$ and $\gamma^*$ are the best fixed price and global budget balanced distribution, respectively. 
    \end{theorem}
    \begin{proof}
        Fix any $\varepsilon>0$, and let $\delta$ be a positive number we set later. Consider the sequence where $(s_t,b_t)=(0, \nicefrac 12-\delta)$ if $t$ is odd, and $(s_t,b_t)=(\nicefrac 12 +\delta, 1)$ otherwise. Any fixed price can make at most half of the trades happen, with a total gain from trade of at most $\nicefrac{T}{2}\left(\nicefrac{1}{2}-\delta \right)$. 
        
        Consider now the distribution over prices $\gamma$ selecting $(p_1,q_1) = (\nicefrac 12 +\delta, \nicefrac 12 - \delta )$ with probability $\alpha=\nicefrac{(1-2\delta)}{(1+6\delta)}$, and $(p_2,q_2) = (0,\nicefrac 12 -\delta)$ otherwise.
        We conclude the proof by arguing that $\gamma$ satisfies the budget balance constraints, and attains total gain from trade that is roughly twice that of $p^*$.
        First, we show that $\gamma$ is global budget balanced. We have 
        \begin{align*}
            \sum\limits_{t\in [T]} \subE{(p,q)\sim\gamma}{\prof_t(p,q)} &= \alpha \, \sum\limits_{t\in[T]}\prof_t(p_1,q_1)+ (1-\alpha) \sum\limits_{t\in [T]}\prof_t(p_2,q_2) \\
            &\ge \alpha \, ( - 2  \delta T) + (1-\alpha) \tfrac T2 (\tfrac 12 - \delta)= 0,
        \end{align*}
        where in the last equality we use the definition of $\alpha$. We move our attention to the gain from trade:
        \begin{align*}
            \sum_{t=1}^T \mathbb{E}_{(p,q)\sim\gamma}[\gft_t(p,q)] &=      \alpha \, \sum_{t=1}^T \gft_t(p_1,q_1) + (1-\alpha) \, \sum_{t=1}^T \gft_t(p_2,q_2) \\ 
            &\ge \alpha  T \left(\tfrac 12-\delta\right) + (1-\alpha) \tfrac{T}{2}\left(\tfrac 12-\delta\right)\\
            &\ge \left(1 + \alpha \right)\sum_{t=1}^T\gft_t(p^*).
        \end{align*}
        Plugging in the last formula the definition of $\alpha$ and setting $\delta = \nicefrac \eps 8$ yields the desired result. 
    \end{proof}

\section{Final Remarks and Open Problems}

In this paper we introduce the notion of global budget balance in the repeated bilateral trade problem. With this notion, we show for the first time that it is possible to achieve sublinear regret with respect to the best fixed price in hindsight, without relying on any additional assumption. In the full feedback model we prove that the minimax regret rate of the learning problem is $\tilde \Theta(\sqrt T)$, while in the partial feedback models, we provide an upper bound on the regret of order $\tilde O(T^{\nicefrac 34})$, which is complemented with a $\Omega(T^{\nicefrac57})$ lower bound. Our regret results proves a clear separation between the two feedback models, but leave an open gap between the $T^{\nicefrac57}$ and $T^
{\nicefrac 34}$ rates in partial feedback. 

Inspired by Bandits with Knapsack, we formulated a new benchmark: the best feasible distribution over prices. Against this harder benchmark we prove that it is possible to achieve sublinear $2$-regret, while no algorithm can achieve sublinear $(1+\eps_0)$-regret. We leave as an open question the characterization of the optimal competitive ratio $\alpha \in [1+\eps_0,2]$ obtainable against this benchmark.

\section*{Acknowledgments}
MB, MC, AC, FF are partially supported by the FAIR (Future Artificial Intelligence Research) project PE0000013, funded by the NextGenerationEU program within the PNRR-PE-AI scheme (M4C2, investment 1.3, line on Artificial Intelligence). FF is also partially supported by ERC Advanced Grant 788893 AMDROMA “Algorithmic and Mechanism Design Research in Online Markets”, and PNRR MUR project IR0000013-SoBigData.it. MC is also partially supported by the EU Horizon project ELIAS (European Lighthouse of AI for Sustainability, No. 101120237). AC is partially supported by MUR - PRIN 2022 project 2022R45NBB funded by the NextGenerationEU program.

\bibliographystyle{plainnat}
\bibliography{references}

\clearpage
\appendix

\section{Appendix}
\subsection{Well Posedness of the Two Benchmarks}
\label{app:model}

    \begin{proposition}\label{prop:Benchmark1}
          For any sequence of valuations $\cS$ there exists a price $p^* \in [0,1]$ such that: 
            \[
                R_T(\cA,\cS) = \sum_{t=1}^T \gft_t(p^*) - \E{\sum_{t=1}^T\gft_t(p_t,q_t)}.
            \]
    \end{proposition}
     \begin{proof}
            %

            Denote the cumulative gain from trade of a pair of price $(p,q)$ as follows:
            \begin{equation}
            \label{eq:usc}
                f(p,q) \defeq \sum_{t=1}^T \gft_t(p,q).
            \end{equation}
            This function is upper semi-continuous on the upper left triangle $\{(p,q) \in [0,1]^2 \mid p \le q\}$ (see \Cref{cl:usc}), thus the $\sup$ in the definition is indeed a max. For the remaining part of the statement, let $(\hat p, \hat q)$ be any pair of prices in the $\argmax$ of \Cref{def:benchmark_1}.
            It is easy to see that any price $p^* \in [\hat p, \hat q]$ achieves the same total gain from trade, while trivially respecting the budget balance constraint. 
        \end{proof}
        \begin{claim}
            \label{cl:usc}
                The function $f$ defined in \cref{eq:usc} is upper semi-continuous on $\{(p,q) \in [0,1]^2 \mid p \le q\}$.
        \end{claim}
        \begin{proof}[Proof of Claim \ref{cl:usc}]
            The function $f$ is the sum of $T$ terms of the following form: 
            \[
            \gft_t(p,q) = \ind{s_t \le p} \ind{q\le b_t} (b_t - s_t).
            \]
            Moreover, for pairs in $\{(p,q) \in [0,1]^2 \mid p \le q\}$  the gain from trade is non-zero only for steps $t\in [T]$ such that $s_t \le b_t$. This implies that $f$ is the sum of at most $T$ step-functions that are upper semi-continuous.
        \end{proof}

        \begin{proposition}
        \label{pr:benchmark_2}
            The definition of the best fixed distribution is well-posed. Moreover, there always exists a feasible distribution $\gamma^*$ with support at most two that attains the $\sup$.
        \end{proposition}
    \begin{proof}
            Fix any sequence of valuations $\{(s_t,b_t)\}_{t=1}^T$, we introduce two auxiliary functions: 
            \[
                f(p,q) \defeq \sum_{t=1}^T \gft_t(p,q) \quad \text{ and } \quad g(p,q) \defeq \sum_{t=1}^T \prof_t(p,q).
            \]
            We can rewrite the program in \Cref{def:benchmark_2} as:
            \begin{subequations}
            \begin{align*}
            \sup\limits_{\gamma\in\Delta([0,1]^2)}	&\quad \subE{(p,q)\sim\gamma}{f(p,q)} \\
    	s.t. &\quad\subE{(p,q)\sim\gamma}{g(p,q)}\ge 0.
            \end{align*}
            \end{subequations}
            As a first step, we show that the support of $\gamma$ can be restricted to a discrete grid $G$.
            To simplify the exposition, we sort the sets of valuations $\{0,1, s_1, \dots, s_T\}$ and $\{0,1, b_1, \dots, b_T\}$ in increasing order.
            Formally, we define the set $\{s^0 = 0, s^1, \dots, s^T, s^{T+1} = 1\}$, where $s^i\le s^{i+1}$ for each $i$, and $\{s^i\}^{T+1}_{i=0} = \{s_t\}_{t=1}^T \cup \{0,1\}$.
            Similarly, we define the set $\{b^0 = 0, b^1, \dots, b^T, b^{T+1} = 1\}$, where $b^i\le b^{i+1}$ for each $i$, and $\{b^i\}^{T+1}_{i=0} = \{b_t\}_{b=1}^T \cup \{0,1\}$.\footnote{For the sake of clarity, we assume that $s_i\neq s_j$ and $s_i,s_j\notin \{0,1\}$, and $b_i\neq b_j$ and $b_i,b_j\notin\{0,1\}$ for each $i,j \in [T]$.
            It is easy to extend our results to the general setting.}
            The grid $G$ contains all the points of the form $(s^i,b^j)$ with $i,j \in \{ 0,1,\dots, T+1\}$.

            Now, we assign any point in $[0,1]^2$ to a point in the grid $G$.
            In particular, we define the map $\pi_G$ that associates each $(p,q) \in [0,1]^2$ to the upper-most and left-most point in its subset of the partition. Formally, $\pi_G:(p,q)\mapsto(s^i,b^j)$, where $i$ is the greatest index such that $ s^{i}\le p$, and $j$ is the smallest index such that $ b^{j}\ge q$.  
            
            From \Cref{cl:discrete} we have that in Program~\eqref{eq:benchmark_2} we can restrict our attention to distributions $\gamma$ that are supported on $G$ and thus are discrete.
            Hence, we can rewrite Program~\eqref{eq:benchmark_2} as the  following linear program: \[\max_{(x_1,x_2)\in\cC} x_1\textnormal{ s.t. }x_2\ge 0,\] 
            where $\mathcal{C}$ is the convex hull of $\cX \defeq \{(f(p,q), g(p,q))\,|\, (p,q)\in G\}$. Consider any optimal solution $(x_1^\star,x_2^\star)$ of such linear program.
            Since $(x_1^\star,x_2^\star)$ belongs to $\cC$, which is a convex hull of a finite set of points, by Caratheodory's Theorem ( see \eg Theorem 5.1 of \citet{Schrijver03}) it can be expressed as a convex combination of $3$ points in $\cX$.

            As a direct implication of first-order optimality conditions (\ie the gradient $(1,0)$ has to belong to the normal cone of $\cC$ at $(x_1^\star,x_2^\star)$) we have that $(x_1^\star,x_2^\star)$ must be on the boundary $\partial \cC$ of $\cC$. 
            This also yields the existence of an hyperplane supporting $\cC$ at $(x_1^\star,x_2^\star)$ \citep[see, \eg Lemma 4.2.1 of ][]{hiriart2004fundamentals}.
            Since $\cC$ is entirely contained in one of the halfspaces defined by the supporting hyperplane, and since $(x_1^\star,x_2^\star)\in\partial\cC$, it must be the case that either $(x_1^\star,x_2^\star)\in \cX$, or we can write the optimal point as a convex combination two points belonging to $\cX$ (\ie the two points defining the face of the polytope containing $(x_1^\star,x_2^\star)$).
            Call $(p_1, q_1)$ and $(p_2, q_2)$ the two points that are in the preimage of the two points generating $(x_1^\star,x_2^\star)$ according to $f$ and $g$, respectively. We showed that there exists an optimal solution whose support consists of the two (possibly coinciding) points $(p_1, q_1)$ and $(p_2, q_2)$. This concludes the proof. 
        \end{proof}

            \begin{claim}
            \label{cl:discrete}
                Let $\gamma \in \Delta([0,1]^2)$. There exists a distribution $\gamma_G\in \Delta([0,1]^2)$ with the following three properties: 
            \begin{OneLiners}
                 \item $\subE{(p,q)\sim\gamma}{f(p,q)} = \subE{(p,q)\sim\gamma_G}{f(p,q)}$;
                 \item $\subE{(p,q)\sim\gamma}{g(p,q)} \ge \subE{(p,q)\sim \gamma_G}{g(p,q)}$;
                 \item $\supp(\gamma_G) \subseteq G$.
            \end{OneLiners}
            \end{claim}
            \begin{proof}
                We define the distribution $\gamma_G$ on $G$ by assigning to each point of the grid $(s^i,b^j)$, with $i\in \{1,\ldots,T\}$ and $j\in\{0,\ldots,T-1\}$, the probability mass which $\gamma$ assigns to points in the cell of the grid $\{(p,q)\in[0,1]^2:s^{i-1}\le i\le s^i, b^j\le q\le b^{j+1}\}$. Formally, the distribution $\gamma_G$ is such that 
                \[
                    \subP{(p,q)\sim \gamma_G}{(p,q) = (s^i,b^j)} = \subP{(p,q)\sim\gamma}{\pi_G(p,q) = (s^i,b^j)}.
                \]
                The new distribution $\gamma_G$ is clearly supported on $G$ and thus verifies the third point of the claim.
                We now prove the remaining two points.
                First, by construction, the expected gain from trade is not affected by the change in distribution.
                Indeed, for each $t$, $p$, and $q$,  it holds 
                \[\gft_t(p,q) = \gft_t(\pi_G(p,q))\] since $\ind{s_t \le p} \ind{q\le b_t}=1$ if and only if $\ind{s_t \le p'} \ind{q'\le b_t}=1$, where   $(p',q')\coloneqq \pi_G(p,q)$. We conclude the proof by showing that the profit does not decrease. It is sufficient to prove that for each $t$, $p$, and $q$, it holds 
                \[\prof_t(p,q) \ge \prof_t(\pi_G(p,q)).\]
                Since $(q-p)\le (q'-p')$, $\pi_G(p,q)$ and $(p,q)$ make the same trades happen. Then,  $\pi_G(p,q)$ extracts at least the same profit of the pair $(p,q)$.
                This concludes the proof.
            \end{proof}

\subsection{Missing Proofs from \Cref{subsec:partialLower}}\label{app:LBbandit}

\begin{lemma}\label{lem:welldefinedmuk}
The distributions $\mu_k$ are well defined for all $k\in\{0,\ldots, N-1\}$.
\end{lemma}
\begin{proof}
    Since for all $w\in\cW_1\cup\cW_2\cup\cW_3\cup\cW_4\cup\cW_5$ the weights $\mu_k(w)$ are positive for all $k\in\{0,\ldots,N-1\}$ we just need to prove that $\mu_k(w)$ is positive for all $w\in \cW_6$.
 
    Then, using the upper bound on $\gamma_3^{\textnormal{tot}}<2\gamma_1N\le 1/32N$ we get that:
    \begin{align*}
        \gamma_6&\ge\frac{1}{4}\left(1-\left(\frac{1}{32N}+\frac{1}{32N}+\frac{13N-14}{16N}+\frac{1}{64}\right)\right)\\
        &= \frac{1}{4}\left(1-\left(\frac{13(N-1)}{16N}+\frac{1}{64}\right)\right)\ge\frac{1}{32}.
    \end{align*}
    This, together with the fact that $\gamma_6\le1/4$, proves that all the probabilities in the instances' distributions $\mu_k$ are well defined.
\end{proof}

\lowertriisdom*
    \begin{proof}
        Consider any point $(p^1,q^1)=\left(\frac{1-\ell}{2}+i\Delta,\frac{1-\ell}{2}+\delta+i\Delta\right)$. For any $i\in\{0,\ldots,N-1\}$ we define $(p^1,q^{1'}_j)=\left(\frac{1-\ell}{2}+i\Delta,\frac{1-\ell}{2}-\delta+(i-j)\Delta\right)$ for each $j\in\{0,\ldots,i\}$.
        Simple calculations show that:
        \begin{align*}
            \mathbb{E}_k[\gft(p^1,q^1,s,b)&-\gft(p,q_j',s,b)] = \sum\limits_{\iota=j}^{i-1}\left[\mu_k(w_3^\iota)\left(\frac{1-\ell}{2}-\delta+\iota\Delta\right)+\gamma_4\delta\right]\\
            &=\sum\limits_{\iota=j}^{i-1}\left[\gamma_1\cdot\frac{1-\ell-\rho-2\iota \Delta}{\frac{1-\ell}{2}-\delta+\iota \Delta}\left(\frac{1-\ell}{2}-\delta+\iota\Delta\right)+4\gamma_1(13N-14)\delta\right]\\
            &=\sum\limits_{\iota=j}^{\iota-1}\left[\gamma_1(1-\ell-\rho-2\iota\Delta+4\delta(13N-14))\right],
        \end{align*}
        which has to hold for all $i\in\{0,\ldots,N-1\}$. The worst case is when $\iota=N-1$. This yields
        \[
        \gamma_1\left(1-\ell-\rho-2\ell+2\frac{\ell}{N-1}(13N-14)\right)>0,\quad \forall N>1
        \]
        and thus for all $j$ we have that:
        \[
        \mathbb{E}_k[\gft(p^1,q^1,s,b)-\gft(p^1,q^{1'}_j,s,b)]>0.
        \]
        The proof is concluded by noting that, for any $(p,q)\in\cG_\cW\cap\{(p,q)\in[0,1]^2\mid p< q\}$ in the lower triangle, the gain from trade is upper bounded by that of some $(p^1,q^{1'}_j)$, that is
        \[
        \mathbb{E}_k[\gft(p,q,s,b)]\le\mathbb{E}_k[\gft(p^1,q^{1'}_j,s,b)].
        \]
        On the other hand, for any $(p',q')\in\cG_\cW\cap\{(p,q)\in[0,1]^2\mid p\geq q\}$ in the upper triangle, the gain from trade is lower bounded by that of a point $(p^1,q^1)$, that is
        \[
        \mathbb{E}_k[\gft(p^1,q^1,s,b)]\le\mathbb{E}_k[\gft(p',q',s,b)].
        \]
        This concludes the proof of the lemma.
    \end{proof}

\pinsker*

    \begin{proof}

    A simple application of the Pinsker's inequality shows that
    \begin{align}\label{eq:KL1}
        \mathbb{E}_k[\cM_k] - \mathbb{E}_0[\cM_k]\le T\sqrt{\frac{1}{2}\textnormal{KL}(\mathbb{P}_0, \mathbb{P}_k)}.
    \end{align}
    
    By \Cref{lem:same_FB} and by the standard $\textnormal{KL}$ decomposition theorem of the $\textnormal{KL}$ divergence~\citep[Chapter~6]{cesa2006prediction}, we have that the KL divergence between $\mathbb{P}_0$ and $\mathbb{P}_k$ only depends on the expected number of times that exploring actions were played, and on the KL divergence of the feedback distributions in such regions:
    \begin{align*}
        \textnormal{KL}(\mathbb{P}_0, \mathbb{P}_k)&=\mathbb{E}_0[\cN_k]\cdot \textnormal{KL}(\cH_0, \cH_k),
    \end{align*}
    where, for each $k$, $\cH_{k}$ is a discrete distribution on the $4$ possible outcomes of the two-bit feedback. Formally,
    \begin{align*}
    \cH_k(z)=
    \begin{cases}(k+1)\gamma_1+\eps \ind{k\neq0} + \gamma_6,&\text{if}\quad z={(1,1)}\\ 
    \gamma_1-\eps\ind{k\neq0}+\gamma_1(N-2-k)+\gamma_6,&\text{if}\quad z={(0,1)}\\
    (k+1)\gamma_1-\eps\ind{k\neq0}+\gamma_6+\gamma_5+\gamma_4(k+1)+\gamma_3^{\textnormal{tot}},&\text{if}\quad z={(1,0)}\\
    \gamma_1+\eps\ind{k\neq0}+\gamma_1(N-2-k)+\gamma_6+\gamma_4(N-k-1)&\text{if}\quad z={(0,0)}
    \end{cases}
    \end{align*}
    By upperbounding the $\textnormal{KL}$ divergence with the $\chi^2$-distance~\citep[Chapter~14]{lattimore2020bandit}, we immediately obtain that 
    \begin{align}
        \textnormal{KL}(\cH_0, \cH_k)\le\chi^2(\cH_0, \cH_k)
        =\sum\limits_{z\in\{0,1\}^2} \frac{(\cH_0(z)-\cH_k(z))^2}{\cH_0(z)}
        \le \eps^2\frac{4}{\gamma_6},\label{eq:KL2}
    \end{align}
    where the last inequality holds since $\cH_0(z)\ge\gamma_6$.
    \end{proof}

\lemlowerboundbandit*
\begin{proof}
    By summing over instances $k\in\{1,\ldots, N-1\}$ the result of \Cref{lem:KL} and using Jensen's inequality, we obtain:
    \begin{align}
        \frac{1}{N-1}\sum\limits_{k=1}^{N-1}(\mathbb{E}_k[\cM_k]-\mathbb{E}_0[\cM_k])
       \le\eps T\sqrt{\frac{2}{\gamma_6}\frac{1}{N-1}\sum\limits_{k=1}^{N-1}\mathbb{E}_0[\cN_k]}.\label{eq:KL3}
    \end{align}
    Then, by rearranging \Cref{lem:upperboundGFTLB} by and substituting $c_2=c_1-\gamma_5\frac{1-\ell}{2}$ we obtain:
    \[
    \sum\limits_{t=1}^T\mathbb{E}_k[\gft(p_t,q_t,s,b)]\le \mathbb{E}_k\left[c_1 T+\rho\cdot \eps \cM_k- \gamma_5\frac{1-\ell}{2}\sum\limits_{j=1}^{N-1}\cN_{j}\right],
    \]
    and by summing over $k\in\{1,\ldots,N-1\}$, dividing by $N-1$, and using \Cref{eq:KL3} we get
    \begin{align*}
        \frac{1}{N-1}\sum\limits_{t=1}^T\sum\limits_{k=1}^{N-1}\mathbb{E}_k[\gft(p_t,q_t,s,b)]
        &\le c_1T+\rho\eps\frac{1}{N-1}\sum\limits_{k=1}^{N-1}\mathbb{E}_k[\cM_k]\\
        &\le c_1T +\rho\eps\frac{1}{N-1}\sum\limits_{k=1}^{N-1}\mathbb{E}_0[\cM_k]+\rho\eps^2 T\sqrt{\frac{2}{\gamma_6}\frac{1}{N-1}\sum\limits_{k=1}^{N-1}\mathbb{E}_0[\cN_k]}.
    \end{align*}
    
    Then, the average of the regret $R_T^k$ over instances $k\in\{1,\ldots, N-1\}$ can be lower bounded by
    \begin{align}
        \frac{1}{N-1}\sum\limits_{k=1}^{N-1}
        R_T^k&\ge T(c_1+\rho\eps)-\frac{1}{N-1}\sum\limits_{k=1}^{N-1}\mathbb{E}_k[\gft(p_t,q_t,s,b)]\nonumber\\
        &\ge T(c_1+\rho\eps)-\left(c_1T +\rho\eps\frac{1}{N-1}\sum\limits_{k=1}^{N-1}\mathbb{E}_0[\cM_k]+\rho\eps^2 T\sqrt{\frac{2}{\gamma_6}\frac{1}{N-1}\sum\limits_{k=1}^{N-1}\mathbb{E}_0[\cN_k]}\right)\nonumber\\
        &\ge\rho\eps T \left(\frac{1}{2}-\eps \sqrt{\frac{2}{\gamma_6}\frac{1}{N-1}\sum\limits_{k=1}^{N-1}\mathbb{E}_0[\cN_k]}\right)\label{eq:LB4},
    \end{align}
    where the first inequality follows by \Cref{lem:LBlem1} while the last inequality holds for any $N\ge 2$.

    Then, we divide the analysis in two cases. Intuitively, 
    the first one correspond to the cases in which the algorithm does not explore enough (\ie $\sum_{k}\mathbb{E}_0[\cN_k]$ is small) and, therefore, it cannot correctly identify the instance. In the second case the algorithm spends a large time exploring  (\ie $\sum_{k}\mathbb{E}_0[\cN_k]$ is large), and thereby accumulates large regret (by \Cref{lem:explorationiscostly}).
    
    Formally,
    if $\eps \sqrt{\frac{2}{\gamma_6}\frac{1}{N-1}\sum_{k=1}^{N-1}\mathbb{E}_0[\cN_k]}\le\frac14$ then \Cref{eq:LB4} implies that the average regret over instances $k\in\{1,\ldots,N-1\}$ is at least
    \[
    \frac{1}{N-1}\sum_{k=1}^{N-1}R_T^k\ge
    \frac14\rho\eps T\ge\frac{1}{10^3}\eps T.
    \]
    Then, there must exist at least an instance $k\in\{1,\ldots,N-1\}$ in which $R_T^k=\Omega(\eps T)$.
    
    Otherwise, if $\eps \sqrt{\frac{2}{\gamma_6}\frac{1}{N-1}\sum_{k=1}^{N-1}\mathbb{E}_0[\cN_k]}\ge \frac14$, then the regret of the base instance can be upper bounded by
    \[
    R_T^0\ge \frac{\gamma_5}{3}\mathbb{E}_0\left[\sum_{k=1}^{N-1} \cN_k\right]\ge \frac{\gamma_5}{3}\frac{\gamma_6}{2}\left(\frac{1}{4\eps}\right)^2(N-1)\ge \frac{\gamma_5}{3}\frac{\gamma_6}{4}\left(\frac{1}{4\eps}\right)^2N\ge\frac{1}{10^{6}}\frac{N}{\eps^2},
    \]
    and hence in the base instance the regret is at least $\Omega\left(\frac{N}{\eps^2}\right)$.
\end{proof}

\end{document}

%% file: gridFKp.tex
\tikzset{every picture/.style={line width=0.75pt}} 

\begin{tikzpicture}[x=0.75pt,y=0.75pt,yscale=-1,xscale=1, scale=1]
\definecolor{niceRed}{RGB}{190,38,38}

\draw   (170,60) -- (490,60) -- (490,380) -- (170,380) -- cycle ;
\draw  [dash pattern={on 4.5pt off 4.5pt}]  (330,60) -- (330,380) ;
\draw  [dash pattern={on 4.5pt off 4.5pt}]  (410,60) -- (410,380) ;
\draw  [dash pattern={on 4.5pt off 4.5pt}]  (250,60) -- (250,380) ;
\draw  [dash pattern={on 4.5pt off 4.5pt}]  (290,60) -- (290,380) ;
\draw  [dash pattern={on 4.5pt off 4.5pt}]  (370,60) -- (370,380) ;
\draw  [dash pattern={on 4.5pt off 4.5pt}]  (210,60) -- (210,380) ;
\draw  [dash pattern={on 4.5pt off 4.5pt}]  (450,60) -- (450,380) ;
\draw    (490,60) -- (170,380) ;
\draw    (170,380) ;
\draw [shift={(170,380)}, rotate = 0] [color={rgb, 255:red, 0; green, 0; blue, 0 }  ][fill={rgb, 255:red, 0; green, 0; blue, 0 }  ][line width=0.75]      (0, 0) circle [x radius= 3.35, y radius= 3.35]   ;
\draw    (490,60) ;
\draw [shift={(490,60)}, rotate = 0] [color={rgb, 255:red, 0; green, 0; blue, 0 }  ][fill={rgb, 255:red, 0; green, 0; blue, 0 }  ][line width=0.75]      (0, 0) circle [x radius= 3.35, y radius= 3.35]   ;
\draw    (450,100) ;
\draw [shift={(450,100)}, rotate = 0] [color={rgb, 255:red, 0; green, 0; blue, 0 }  ][fill={rgb, 255:red, 0; green, 0; blue, 0 }  ][line width=0.75]      (0, 0) circle [x radius= 3.35, y radius= 3.35]   ;
\draw    (410,140) ;
\draw [shift={(410,140)}, rotate = 0] [color={rgb, 255:red, 0; green, 0; blue, 0 }  ][fill={rgb, 255:red, 0; green, 0; blue, 0 }  ][line width=0.75]      (0, 0) circle [x radius= 3.35, y radius= 3.35]   ;
\draw    (370,180) ;
\draw [shift={(370,180)}, rotate = 0] [color={rgb, 255:red, 0; green, 0; blue, 0 }  ][fill={rgb, 255:red, 0; green, 0; blue, 0 }  ][line width=0.75]      (0, 0) circle [x radius= 3.35, y radius= 3.35]   ;
\draw    (330,220) ;
\draw [shift={(330,220)}, rotate = 0] [color={rgb, 255:red, 0; green, 0; blue, 0 }  ][fill={rgb, 255:red, 0; green, 0; blue, 0 }  ][line width=0.75]      (0, 0) circle [x radius= 3.35, y radius= 3.35]   ;
\draw    (290,260) ;
\draw [shift={(290,260)}, rotate = 0] [color={rgb, 255:red, 0; green, 0; blue, 0 }  ][fill={rgb, 255:red, 0; green, 0; blue, 0 }  ][line width=0.75]      (0, 0) circle [x radius= 3.35, y radius= 3.35]   ;
\draw    (250,300) ;
\draw [shift={(250,300)}, rotate = 0] [color={rgb, 255:red, 0; green, 0; blue, 0 }  ][fill={rgb, 255:red, 0; green, 0; blue, 0 }  ][line width=0.75]      (0, 0) circle [x radius= 3.35, y radius= 3.35]   ;
\draw    (210,340) ;
\draw [shift={(210,340)}, rotate = 0] [color={rgb, 255:red, 0; green, 0; blue, 0 }  ][fill={rgb, 255:red, 0; green, 0; blue, 0 }  ][line width=0.75]      (0, 0) circle [x radius= 3.35, y radius= 3.35]   ;

\draw  [dash pattern={on 4.5pt off 4.5pt}]  (170,220) -- (490,220) ;
\draw  [dash pattern={on 4.5pt off 4.5pt}]  (170,180) -- (490,180) ;
\draw [color=niceRed  ,draw opacity=1 ]   (170,60) ;
\draw [shift={(170,60)}, rotate = 0] [color=niceRed  ,draw opacity=1 ][fill=niceRed  ,fill opacity=1 ][line width=0.75]      (0, 0) circle [x radius= 3.35, y radius= 3.35]   ;
\draw [color=niceRed  ,draw opacity=1 ]   (480,60) -- (170,370) ;
\draw [color=niceRed  ,draw opacity=1 ]   (170,370) ;
\draw [shift={(170,370)}, rotate = 0] [color=niceRed  ,draw opacity=1 ][fill=niceRed  ,fill opacity=1 ][line width=0.75]      (0, 0) circle [x radius= 3.35, y radius= 3.35]   ;
\draw [color=niceRed  ,draw opacity=1 ]   (450,90) ;
\draw [shift={(450,90)}, rotate = 0] [color=niceRed  ,draw opacity=1 ][fill=niceRed  ,fill opacity=1 ][line width=0.75]      (0, 0) circle [x radius= 3.35, y radius= 3.35]   ;
\draw [color=niceRed  ,draw opacity=1 ]   (410,130) ;
\draw [shift={(410,130)}, rotate = 0] [color=niceRed  ,draw opacity=1 ][fill=niceRed  ,fill opacity=1 ][line width=0.75]      (0, 0) circle [x radius= 3.35, y radius= 3.35]   ;
\draw [color=niceRed  ,draw opacity=1 ]   (370,170) ;
\draw [shift={(370,170)}, rotate = 0] [color=niceRed  ,draw opacity=1 ][fill=niceRed  ,fill opacity=1 ][line width=0.75]      (0, 0) circle [x radius= 3.35, y radius= 3.35]   ;
\draw [color=niceRed  ,draw opacity=1 ]   (330,210) ;
\draw [shift={(330,210)}, rotate = 0] [color=niceRed  ,draw opacity=1 ][fill=niceRed  ,fill opacity=1 ][line width=0.75]      (0, 0) circle [x radius= 3.35, y radius= 3.35]   ;
\draw [color=niceRed  ,draw opacity=1 ]   (290,250) ;
\draw [shift={(290,250)}, rotate = 0] [color=niceRed  ,draw opacity=1 ][fill=niceRed  ,fill opacity=1 ][line width=0.75]      (0, 0) circle [x radius= 3.35, y radius= 3.35]   ;
\draw [color=niceRed  ,draw opacity=1 ]   (250,290) ;
\draw [shift={(250,290)}, rotate = 0] [color=niceRed  ,draw opacity=1 ][fill=niceRed  ,fill opacity=1 ][line width=0.75]      (0, 0) circle [x radius= 3.35, y radius= 3.35]   ;
\draw [color=niceRed  ,draw opacity=1 ]   (210,330) ;
\draw [shift={(210,330)}, rotate = 0] [color=niceRed  ,draw opacity=1 ][fill=niceRed  ,fill opacity=1 ][line width=0.75]      (0, 0) circle [x radius= 3.35, y radius= 3.35]   ;
\draw [color=niceRed  ,draw opacity=1 ]   (470,60) -- (170,360) ;
\draw [color=niceRed  ,draw opacity=1 ]   (170,360) ;
\draw [shift={(170,360)}, rotate = 0] [color=niceRed  ,draw opacity=1 ][fill=niceRed  ,fill opacity=1 ][line width=0.75]      (0, 0) circle [x radius= 3.35, y radius= 3.35]   ;
\draw [color=niceRed  ,draw opacity=1 ]   (450,80) ;
\draw [shift={(450,80)}, rotate = 0] [color=niceRed  ,draw opacity=1 ][fill=niceRed  ,fill opacity=1 ][line width=0.75]      (0, 0) circle [x radius= 3.35, y radius= 3.35]   ;
\draw [color=niceRed  ,draw opacity=1 ]   (410,120) ;
\draw [shift={(410,120)}, rotate = 0] [color=niceRed  ,draw opacity=1 ][fill=niceRed  ,fill opacity=1 ][line width=0.75]      (0, 0) circle [x radius= 3.35, y radius= 3.35]   ;
\draw [color=niceRed  ,draw opacity=1 ]   (370,160) ;
\draw [shift={(370,160)}, rotate = 0] [color=niceRed  ,draw opacity=1 ][fill=niceRed  ,fill opacity=1 ][line width=0.75]      (0, 0) circle [x radius= 3.35, y radius= 3.35]   ;
\draw [color=niceRed  ,draw opacity=1 ]   (330,200) ;
\draw [shift={(330,200)}, rotate = 0] [color=niceRed  ,draw opacity=1 ][fill=niceRed  ,fill opacity=1 ][line width=0.75]      (0, 0) circle [x radius= 3.35, y radius= 3.35]   ;
\draw [color=niceRed  ,draw opacity=1 ]   (290,240) ;
\draw [shift={(290,240)}, rotate = 0] [color=niceRed  ,draw opacity=1 ][fill=niceRed  ,fill opacity=1 ][line width=0.75]      (0, 0) circle [x radius= 3.35, y radius= 3.35]   ;
\draw [color=niceRed  ,draw opacity=1 ]   (250,280) ;
\draw [shift={(250,280)}, rotate = 0] [color=niceRed  ,draw opacity=1 ][fill=niceRed  ,fill opacity=1 ][line width=0.75]      (0, 0) circle [x radius= 3.35, y radius= 3.35]   ;
\draw [color=niceRed  ,draw opacity=1 ]   (210,320) ;
\draw [shift={(210,320)}, rotate = 0] [color=niceRed  ,draw opacity=1 ][fill=niceRed  ,fill opacity=1 ][line width=0.75]      (0, 0) circle [x radius= 3.35, y radius= 3.35]   ;
\draw [color=niceRed  ,draw opacity=1 ]   (450,60) -- (170,340) ;
\draw [color=niceRed  ,draw opacity=1 ]   (170,340) ;
\draw [shift={(170,340)}, rotate = 0] [color=niceRed  ,draw opacity=1 ][fill=niceRed  ,fill opacity=1 ][line width=0.75]      (0, 0) circle [x radius= 3.35, y radius= 3.35]   ;
\draw [color=niceRed  ,draw opacity=1 ]   (450,60) ;
\draw [shift={(450,60)}, rotate = 0] [color=niceRed  ,draw opacity=1 ][fill=niceRed  ,fill opacity=1 ][line width=0.75]      (0, 0) circle [x radius= 3.35, y radius= 3.35]   ;
\draw [color=niceRed  ,draw opacity=1 ]   (410,100) ;
\draw [shift={(410,100)}, rotate = 0] [color=niceRed  ,draw opacity=1 ][fill=niceRed  ,fill opacity=1 ][line width=0.75]      (0, 0) circle [x radius= 3.35, y radius= 3.35]   ;
\draw [color=niceRed  ,draw opacity=1 ]   (370,140) ;
\draw [shift={(370,140)}, rotate = 0] [color=niceRed  ,draw opacity=1 ][fill=niceRed  ,fill opacity=1 ][line width=0.75]      (0, 0) circle [x radius= 3.35, y radius= 3.35]   ;
\draw [color=niceRed  ,draw opacity=1 ]   (330,180) ;
\draw [shift={(330,180)}, rotate = 0] [color=niceRed  ,draw opacity=1 ][fill=niceRed  ,fill opacity=1 ][line width=0.75]      (0, 0) circle [x radius= 3.35, y radius= 3.35]   ;
\draw [color=niceRed  ,draw opacity=1 ]   (290,220) ;
\draw [shift={(290,220)}, rotate = 0] [color=niceRed  ,draw opacity=1 ][fill=niceRed  ,fill opacity=1 ][line width=0.75]      (0, 0) circle [x radius= 3.35, y radius= 3.35]   ;
\draw [color=niceRed  ,draw opacity=1 ]   (250,260) ;
\draw [shift={(250,260)}, rotate = 0] [color=niceRed  ,draw opacity=1 ][fill=niceRed  ,fill opacity=1 ][line width=0.75]      (0, 0) circle [x radius= 3.35, y radius= 3.35]   ;
\draw [color=niceRed  ,draw opacity=1 ]   (210,300) ;
\draw [shift={(210,300)}, rotate = 0] [color=niceRed  ,draw opacity=1 ][fill=niceRed  ,fill opacity=1 ][line width=0.75]      (0, 0) circle [x radius= 3.35, y radius= 3.35]   ;
\draw [color=niceRed  ,draw opacity=1 ]   (410,60) -- (170,300) ;
\draw [color=niceRed  ,draw opacity=1 ]   (170,300) ;
\draw [shift={(170,300)}, rotate = 0] [color=niceRed  ,draw opacity=1 ][fill=niceRed  ,fill opacity=1 ][line width=0.75]      (0, 0) circle [x radius= 3.35, y radius= 3.35]   ;
\draw [color=niceRed  ,draw opacity=1 ]   (410,60) ;
\draw [shift={(410,60)}, rotate = 0] [color=niceRed  ,draw opacity=1 ][fill=niceRed  ,fill opacity=1 ][line width=0.75]      (0, 0) circle [x radius= 3.35, y radius= 3.35]   ;
\draw [color=niceRed  ,draw opacity=1 ]   (370,100) ;
\draw [shift={(370,100)}, rotate = 0] [color=niceRed  ,draw opacity=1 ][fill=niceRed  ,fill opacity=1 ][line width=0.75]      (0, 0) circle [x radius= 3.35, y radius= 3.35]   ;
\draw [color=niceRed  ,draw opacity=1 ]   (330,140) ;
\draw [shift={(330,140)}, rotate = 0] [color=niceRed  ,draw opacity=1 ][fill=niceRed  ,fill opacity=1 ][line width=0.75]      (0, 0) circle [x radius= 3.35, y radius= 3.35]   ;
\draw [color=niceRed  ,draw opacity=1 ]   (290,180) ;
\draw [shift={(290,180)}, rotate = 0] [color=niceRed  ,draw opacity=1 ][fill=niceRed  ,fill opacity=1 ][line width=0.75]      (0, 0) circle [x radius= 3.35, y radius= 3.35]   ;
\draw [color=niceRed  ,draw opacity=1 ]   (250,220) ;
\draw [shift={(250,220)}, rotate = 0] [color=niceRed  ,draw opacity=1 ][fill=niceRed  ,fill opacity=1 ][line width=0.75]      (0, 0) circle [x radius= 3.35, y radius= 3.35]   ;
\draw [color=niceRed  ,draw opacity=1 ]   (210,260) ;
\draw [shift={(210,260)}, rotate = 0] [color=niceRed  ,draw opacity=1 ][fill=niceRed  ,fill opacity=1 ][line width=0.75]      (0, 0) circle [x radius= 3.35, y radius= 3.35]   ;
\draw [color=niceRed  ,draw opacity=1 ]   (330,60) -- (170,220) ;
\draw [color=niceRed  ,draw opacity=1 ]   (170,220) ;
\draw [shift={(170,220)}, rotate = 0] [color=niceRed  ,draw opacity=1 ][fill=niceRed  ,fill opacity=1 ][line width=0.75]      (0, 0) circle [x radius= 3.35, y radius= 3.35]   ;
\draw [color=niceRed  ,draw opacity=1 ]   (330,60) ;
\draw [shift={(330,60)}, rotate = 0] [color=niceRed  ,draw opacity=1 ][fill=niceRed  ,fill opacity=1 ][line width=0.75]      (0, 0) circle [x radius= 3.35, y radius= 3.35]   ;
\draw [color=niceRed  ,draw opacity=1 ]   (290,100) ;
\draw [shift={(290,100)}, rotate = 0] [color=niceRed  ,draw opacity=1 ][fill=niceRed  ,fill opacity=1 ][line width=0.75]      (0, 0) circle [x radius= 3.35, y radius= 3.35]   ;
\draw [color=niceRed  ,draw opacity=1 ]   (250,140) ;
\draw [shift={(250,140)}, rotate = 0] [color=niceRed  ,draw opacity=1 ][fill=niceRed  ,fill opacity=1 ][line width=0.75]      (0, 0) circle [x radius= 3.35, y radius= 3.35]   ;
\draw [color=niceRed  ,draw opacity=1 ]   (210,180) ;
\draw [shift={(210,180)}, rotate = 0] [color=niceRed  ,draw opacity=1 ][fill=niceRed  ,fill opacity=1 ][line width=0.75]      (0, 0) circle [x radius= 3.35, y radius= 3.35]   ;

\draw (322,410) node [anchor=south west][inner sep=0.75pt]    {\LARGE $p$};
\draw (140,230) node [anchor=south west][inner sep=0.75pt]    {\LARGE  $p$};
\draw (95,190) node [anchor=south west][inner sep=0.75pt]    {\LARGE  $p+2^{-i}$};

\end{tikzpicture}

%% file: gridFKm.tex
\tikzset{every picture/.style={line width=0.75pt}} 

\begin{tikzpicture}[x=0.75pt,y=0.75pt,yscale=-1,xscale=1,scale=1]

\draw   (150,70) -- (470,70) -- (470,390) -- (150,390) -- cycle ;
\draw    (470,70) -- (150,390) ;
\draw    (150,390) ;
\draw [shift={(150,390)}, rotate = 0] [color={rgb, 255:red, 0; green, 0; blue, 0 }  ][fill={rgb, 255:red, 0; green, 0; blue, 0 }  ][line width=0.75]      (0, 0) circle [x radius= 3.35, y radius= 3.35]   ;
\draw    (470,70) ;
\draw [shift={(470,70)}, rotate = 0] [color={rgb, 255:red, 0; green, 0; blue, 0 }  ][fill={rgb, 255:red, 0; green, 0; blue, 0 }  ][line width=0.75]      (0, 0) circle [x radius= 3.35, y radius= 3.35]   ;
\draw    (430,110) ;
\draw [shift={(430,110)}, rotate = 0] [color={rgb, 255:red, 0; green, 0; blue, 0 }  ][fill={rgb, 255:red, 0; green, 0; blue, 0 }  ][line width=0.75]      (0, 0) circle [x radius= 3.35, y radius= 3.35]   ;
\draw    (390,150) ;
\draw [shift={(390,150)}, rotate = 0] [color={rgb, 255:red, 0; green, 0; blue, 0 }  ][fill={rgb, 255:red, 0; green, 0; blue, 0 }  ][line width=0.75]      (0, 0) circle [x radius= 3.35, y radius= 3.35]   ;
\draw    (350,190) ;
\draw [shift={(350,190)}, rotate = 0] [color={rgb, 255:red, 0; green, 0; blue, 0 }  ][fill={rgb, 255:red, 0; green, 0; blue, 0 }  ][line width=0.75]      (0, 0) circle [x radius= 3.35, y radius= 3.35]   ;
\draw    (310,230) ;
\draw [shift={(310,230)}, rotate = 0] [color={rgb, 255:red, 0; green, 0; blue, 0 }  ][fill={rgb, 255:red, 0; green, 0; blue, 0 }  ][line width=0.75]      (0, 0) circle [x radius= 3.35, y radius= 3.35]   ;
\draw    (270,270) ;
\draw [shift={(270,270)}, rotate = 0] [color={rgb, 255:red, 0; green, 0; blue, 0 }  ][fill={rgb, 255:red, 0; green, 0; blue, 0 }  ][line width=0.75]      (0, 0) circle [x radius= 3.35, y radius= 3.35]   ;
\draw    (230,310) ;
\draw [shift={(230,310)}, rotate = 0] [color={rgb, 255:red, 0; green, 0; blue, 0 }  ][fill={rgb, 255:red, 0; green, 0; blue, 0 }  ][line width=0.75]      (0, 0) circle [x radius= 3.35, y radius= 3.35]   ;
\draw    (190,350) ;
\draw [shift={(190,350)}, rotate = 0] [color={rgb, 255:red, 0; green, 0; blue, 0 }  ][fill={rgb, 255:red, 0; green, 0; blue, 0 }  ][line width=0.75]      (0, 0) circle [x radius= 3.35, y radius= 3.35]   ;

\draw  [dash pattern={on 4.5pt off 4.5pt}]  (150,190) -- (470,190) ;
\draw  [dash pattern={on 4.5pt off 4.5pt}]  (150,150) -- (470,150) ;
\draw  [dash pattern={on 4.5pt off 4.5pt}]  (150,270) -- (470,270) ;
\draw  [dash pattern={on 4.5pt off 4.5pt}]  (150,310) -- (470,310) ;
\draw  [dash pattern={on 4.5pt off 4.5pt}]  (150,350) -- (470,350) ;
\draw  [dash pattern={on 4.5pt off 4.5pt}]  (150,110) -- (470,110) ;
\draw  [dash pattern={on 4.5pt off 4.5pt}]  (310,390) -- (310,70) ;
\draw  [dash pattern={on 4.5pt off 4.5pt}]  (270,390) -- (270,70) ;
\draw [color=blueGrotto  ,draw opacity=1 ]   (460,70) -- (150,380) ;
\draw [color=blueGrotto  ,draw opacity=1 ]   (420,110) ;
\draw [shift={(420,110)}, rotate = 0] [color=blueGrotto  ,draw opacity=1 ][fill=blueGrotto  ,fill opacity=1 ][line width=0.75]      (0, 0) circle [x radius= 3.35, y radius= 3.35]   ;
\draw [color=blueGrotto  ,draw opacity=1 ]   (380,150) ;
\draw [shift={(380,150)}, rotate = 0] [color=blueGrotto  ,draw opacity=1 ][fill=blueGrotto  ,fill opacity=1 ][line width=0.75]      (0, 0) circle [x radius= 3.35, y radius= 3.35]   ;
\draw [color=blueGrotto  ,draw opacity=1 ]   (340,190) ;
\draw [shift={(340,190)}, rotate = 0] [color=blueGrotto  ,draw opacity=1 ][fill=blueGrotto  ,fill opacity=1 ][line width=0.75]      (0, 0) circle [x radius= 3.35, y radius= 3.35]   ;
\draw [color=blueGrotto  ,draw opacity=1 ]   (300,230) ;
\draw [shift={(300,230)}, rotate = 0] [color=blueGrotto  ,draw opacity=1 ][fill=blueGrotto  ,fill opacity=1 ][line width=0.75]      (0, 0) circle [x radius= 3.35, y radius= 3.35]   ;
\draw [color=blueGrotto  ,draw opacity=1 ]   (260,270) ;
\draw [shift={(260,270)}, rotate = 0] [color=blueGrotto  ,draw opacity=1 ][fill=blueGrotto  ,fill opacity=1 ][line width=0.75]      (0, 0) circle [x radius= 3.35, y radius= 3.35]   ;
\draw [color=blueGrotto  ,draw opacity=1 ]   (220,310) ;
\draw [shift={(220,310)}, rotate = 0] [color=blueGrotto  ,draw opacity=1 ][fill=blueGrotto  ,fill opacity=1 ][line width=0.75]      (0, 0) circle [x radius= 3.35, y radius= 3.35]   ;
\draw [color=blueGrotto  ,draw opacity=1 ]   (180,350) ;
\draw [shift={(180,350)}, rotate = 0] [color=blueGrotto  ,draw opacity=1 ][fill=blueGrotto  ,fill opacity=1 ][line width=0.75]      (0, 0) circle [x radius= 3.35, y radius= 3.35]   ;
\draw [color=blueGrotto  ,draw opacity=1 ]   (450,70) -- (150,370) ;
\draw [color=blueGrotto  ,draw opacity=1 ]   (410,110) ;
\draw [shift={(410,110)}, rotate = 0] [color=blueGrotto  ,draw opacity=1 ][fill=blueGrotto  ,fill opacity=1 ][line width=0.75]      (0, 0) circle [x radius= 3.35, y radius= 3.35]   ;
\draw [color=blueGrotto  ,draw opacity=1 ]   (370,150) ;
\draw [shift={(370,150)}, rotate = 0] [color=blueGrotto  ,draw opacity=1 ][fill=blueGrotto  ,fill opacity=1 ][line width=0.75]      (0, 0) circle [x radius= 3.35, y radius= 3.35]   ;
\draw [color=blueGrotto  ,draw opacity=1 ]   (330,190) ;
\draw [shift={(330,190)}, rotate = 0] [color=blueGrotto  ,draw opacity=1 ][fill=blueGrotto  ,fill opacity=1 ][line width=0.75]      (0, 0) circle [x radius= 3.35, y radius= 3.35]   ;
\draw [color=blueGrotto  ,draw opacity=1 ]   (290,230) ;
\draw [shift={(290,230)}, rotate = 0] [color=blueGrotto  ,draw opacity=1 ][fill=blueGrotto  ,fill opacity=1 ][line width=0.75]      (0, 0) circle [x radius= 3.35, y radius= 3.35]   ;
\draw [color=blueGrotto  ,draw opacity=1 ]   (250,270) ;
\draw [shift={(250,270)}, rotate = 0] [color=blueGrotto  ,draw opacity=1 ][fill=blueGrotto  ,fill opacity=1 ][line width=0.75]      (0, 0) circle [x radius= 3.35, y radius= 3.35]   ;
\draw [color=blueGrotto  ,draw opacity=1 ]   (210,310) ;
\draw [shift={(210,310)}, rotate = 0] [color=blueGrotto  ,draw opacity=1 ][fill=blueGrotto  ,fill opacity=1 ][line width=0.75]      (0, 0) circle [x radius= 3.35, y radius= 3.35]   ;
\draw [color=blueGrotto  ,draw opacity=1 ]   (170,350) ;
\draw [shift={(170,350)}, rotate = 0] [color=blueGrotto  ,draw opacity=1 ][fill=blueGrotto  ,fill opacity=1 ][line width=0.75]      (0, 0) circle [x radius= 3.35, y radius= 3.35]   ;
\draw [color=blueGrotto  ,draw opacity=1 ]   (430,70) -- (150,350) ;
\draw [color=blueGrotto  ,draw opacity=1 ]   (150,350) ;
\draw [shift={(150,350)}, rotate = 0] [color=blueGrotto  ,draw opacity=1 ][fill=blueGrotto  ,fill opacity=1 ][line width=0.75]      (0, 0) circle [x radius= 3.35, y radius= 3.35]   ;
\draw [color=blueGrotto  ,draw opacity=1 ]   (430,70) ;
\draw [shift={(430,70)}, rotate = 0] [color=blueGrotto  ,draw opacity=1 ][fill=blueGrotto  ,fill opacity=1 ][line width=0.75]      (0, 0) circle [x radius= 3.35, y radius= 3.35]   ;
\draw [color=blueGrotto  ,draw opacity=1 ]   (390,110) ;
\draw [shift={(390,110)}, rotate = 0] [color=blueGrotto  ,draw opacity=1 ][fill=blueGrotto  ,fill opacity=1 ][line width=0.75]      (0, 0) circle [x radius= 3.35, y radius= 3.35]   ;
\draw [color=blueGrotto  ,draw opacity=1 ]   (350,150) ;
\draw [shift={(350,150)}, rotate = 0] [color=blueGrotto  ,draw opacity=1 ][fill=blueGrotto  ,fill opacity=1 ][line width=0.75]      (0, 0) circle [x radius= 3.35, y radius= 3.35]   ;
\draw [color=blueGrotto  ,draw opacity=1 ]   (310,190) ;
\draw [shift={(310,190)}, rotate = 0] [color=blueGrotto  ,draw opacity=1 ][fill=blueGrotto  ,fill opacity=1 ][line width=0.75]      (0, 0) circle [x radius= 3.35, y radius= 3.35]   ;
\draw [color=blueGrotto  ,draw opacity=1 ]   (270,230) ;
\draw [shift={(270,230)}, rotate = 0] [color=blueGrotto  ,draw opacity=1 ][fill=blueGrotto  ,fill opacity=1 ][line width=0.75]      (0, 0) circle [x radius= 3.35, y radius= 3.35]   ;
\draw [color=blueGrotto  ,draw opacity=1 ]   (230,270) ;
\draw [shift={(230,270)}, rotate = 0] [color=blueGrotto  ,draw opacity=1 ][fill=blueGrotto  ,fill opacity=1 ][line width=0.75]      (0, 0) circle [x radius= 3.35, y radius= 3.35]   ;
\draw [color=blueGrotto  ,draw opacity=1 ]   (190,310) ;
\draw [shift={(190,310)}, rotate = 0] [color=blueGrotto  ,draw opacity=1 ][fill=blueGrotto  ,fill opacity=1 ][line width=0.75]      (0, 0) circle [x radius= 3.35, y radius= 3.35]   ;
\draw [color=blueGrotto  ,draw opacity=1 ]   (390,70) -- (150,310) ;
\draw [color=blueGrotto  ,draw opacity=1 ]   (150,310) ;
\draw [shift={(150,310)}, rotate = 0] [color=blueGrotto  ,draw opacity=1 ][fill=blueGrotto  ,fill opacity=1 ][line width=0.75]      (0, 0) circle [x radius= 3.35, y radius= 3.35]   ;
\draw [color=blueGrotto  ,draw opacity=1 ]   (390,70) ;
\draw [shift={(390,70)}, rotate = 0] [color=blueGrotto  ,draw opacity=1 ][fill=blueGrotto  ,fill opacity=1 ][line width=0.75]      (0, 0) circle [x radius= 3.35, y radius= 3.35]   ;
\draw [color=blueGrotto  ,draw opacity=1 ]   (350,110) ;
\draw [shift={(350,110)}, rotate = 0] [color=blueGrotto  ,draw opacity=1 ][fill=blueGrotto  ,fill opacity=1 ][line width=0.75]      (0, 0) circle [x radius= 3.35, y radius= 3.35]   ;
\draw [color=blueGrotto  ,draw opacity=1 ]   (310,150) ;
\draw [shift={(310,150)}, rotate = 0] [color=blueGrotto  ,draw opacity=1 ][fill=blueGrotto  ,fill opacity=1 ][line width=0.75]      (0, 0) circle [x radius= 3.35, y radius= 3.35]   ;
\draw [color=blueGrotto  ,draw opacity=1 ]   (270,190) ;
\draw [shift={(270,190)}, rotate = 0] [color=blueGrotto  ,draw opacity=1 ][fill=blueGrotto  ,fill opacity=1 ][line width=0.75]      (0, 0) circle [x radius= 3.35, y radius= 3.35]   ;
\draw [color=blueGrotto  ,draw opacity=1 ]   (230,230) ;
\draw [shift={(230,230)}, rotate = 0] [color=blueGrotto  ,draw opacity=1 ][fill=blueGrotto  ,fill opacity=1 ][line width=0.75]      (0, 0) circle [x radius= 3.35, y radius= 3.35]   ;
\draw [color=blueGrotto  ,draw opacity=1 ]   (190,270) ;
\draw [shift={(190,270)}, rotate = 0] [color=blueGrotto  ,draw opacity=1 ][fill=blueGrotto  ,fill opacity=1 ][line width=0.75]      (0, 0) circle [x radius= 3.35, y radius= 3.35]   ;
\draw [color=blueGrotto  ,draw opacity=1 ]   (310,70) -- (150,230) ;
\draw [color=blueGrotto  ,draw opacity=1 ]   (150,230) ;
\draw [shift={(150,230)}, rotate = 0] [color=blueGrotto  ,draw opacity=1 ][fill=blueGrotto  ,fill opacity=1 ][line width=0.75]      (0, 0) circle [x radius= 3.35, y radius= 3.35]   ;
\draw [color=blueGrotto  ,draw opacity=1 ]   (310,70) ;
\draw [shift={(310,70)}, rotate = 0] [color=blueGrotto  ,draw opacity=1 ][fill=blueGrotto  ,fill opacity=1 ][line width=0.75]      (0, 0) circle [x radius= 3.35, y radius= 3.35]   ;
\draw [color=blueGrotto  ,draw opacity=1 ]   (270,110) ;
\draw [shift={(270,110)}, rotate = 0] [color=blueGrotto  ,draw opacity=1 ][fill=blueGrotto  ,fill opacity=1 ][line width=0.75]      (0, 0) circle [x radius= 3.35, y radius= 3.35]   ;
\draw [color=blueGrotto  ,draw opacity=1 ]   (230,150) ;
\draw [shift={(230,150)}, rotate = 0] [color=blueGrotto  ,draw opacity=1 ][fill=blueGrotto  ,fill opacity=1 ][line width=0.75]      (0, 0) circle [x radius= 3.35, y radius= 3.35]   ;
\draw [color=blueGrotto  ,draw opacity=1 ]   (190,190) ;
\draw [shift={(190,190)}, rotate = 0] [color=blueGrotto  ,draw opacity=1 ][fill=blueGrotto  ,fill opacity=1 ][line width=0.75]      (0, 0) circle [x radius= 3.35, y radius= 3.35]   ;
\draw  [dash pattern={on 4.5pt off 4.5pt}]  (150,230) -- (470,230) ;
\draw [color=blueGrotto  ,draw opacity=1 ]   (460,70) ;
\draw [shift={(460,70)}, rotate = 0] [color=blueGrotto  ,draw opacity=1 ][fill=blueGrotto  ,fill opacity=1 ][line width=0.75]      (0, 0) circle [x radius= 3.35, y radius= 3.35]   ;
\draw [color=blueGrotto  ,draw opacity=1 ]   (450,70) ;
\draw [shift={(450,70)}, rotate = 0] [color=blueGrotto  ,draw opacity=1 ][fill=blueGrotto  ,fill opacity=1 ][line width=0.75]      (0, 0) circle [x radius= 3.35, y radius= 3.35]   ;
\draw [color=blueGrotto  ,draw opacity=1 ]   (150,70) ;
\draw [shift={(150,70)}, rotate = 0] [color=blueGrotto  ,draw opacity=1 ][fill=blueGrotto  ,fill opacity=1 ][line width=0.75]      (0, 0) circle [x radius= 3.35, y radius= 3.35]   ;

\draw (305,425) node [anchor=south west][inner sep=0.75pt]    {\LARGE $p$};
\draw (120,240) node [anchor=south west][inner sep=0.75pt]    {\LARGE  $p$};
\draw (225,425) node [anchor=south west][inner sep=0.75pt]    {\LARGE $p-$2$^{-i}$};

\end{tikzpicture}

%% file: instance_LB_Full.tex
\definecolor{blueGrotto}{HTML}{059DC0}
\tikzset{every picture/.style={line width=0.75pt}} 

\begin{tikzpicture}[x=0.5pt,y=0.5pt,yscale=-1,xscale=1]

\draw  [draw opacity=0][fill=orange2  ,fill opacity=0.59 ] (130,310) -- (370,310) -- (370,390) -- (130,390) -- cycle ;
\draw  [draw opacity=0][fill=orange2  ,fill opacity=0.59 ] (370,70) -- (450,70) -- (450,310) -- (370,310) -- cycle ;
\draw  [draw opacity=0][fill=mgreen  ,fill opacity=0.6 ] (370,310) -- (450,310) -- (450,390) -- (370,390) -- cycle ;
\draw   (130,70) -- (450,70) -- (450,390) -- (130,390) -- cycle ;
\draw    (450,70) -- (130,390) ;
\draw  [dash pattern={on 0.84pt off 2.51pt}]  (370,70) -- (370,390) ;
\draw  [dash pattern={on 0.84pt off 2.51pt}]  (450,310) -- (130,310) ;
\draw [color=blueGrotto  ,draw opacity=1 ]   (130,310) ;
\draw [shift={(130,310)}, rotate = 0] [color=blueGrotto  ,draw opacity=1 ][fill=blueGrotto  ,fill opacity=1 ][line width=0.75]      (0, 0) circle [x radius= 4., y radius= 4.]   ;
\draw [color=blueGrotto  ,draw opacity=1 ]   (370,70) ;
\draw [shift={(370,70)}, rotate = 0] [color=blueGrotto  ,draw opacity=1 ][fill=blueGrotto  ,fill opacity=1 ][line width=0.75]      (0, 0) circle [x radius= 4, y radius= 4]   ;
\draw [color=Sepia, draw opacity=1 ]   (370,310) ;
\draw [shift={(370,310)}, rotate = 0] [color=Sepia, draw opacity=1 ][fill=Sepia, fill opacity=1 ][line width=0.75]      (0, 0) circle [x radius= 4, y radius= 4]   ;
\end{tikzpicture}

%% file: LBalpha.tex
\tikzset{every picture/.style={line width=0.75pt}} 

\begin{tikzpicture}[x=0.75pt,y=0.75pt,yscale=-1,xscale=1]

\draw    (230,220) -- (230,43) ;
\draw [shift={(230,40)}, rotate = 90] [fill={rgb, 255:red, 0; green, 0; blue, 0 }  ][line width=0.08]  [draw opacity=0] (8.93,-4.29) -- (0,0) -- (8.93,4.29) -- cycle    ;
\draw    (230,220) -- (407,220) ;
\draw [shift={(410,220)}, rotate = 180] [fill={rgb, 255:red, 0; green, 0; blue, 0 }  ][line width=0.08]  [draw opacity=0] (8.93,-4.29) -- (0,0) -- (8.93,4.29) -- cycle    ;
\draw    (310,220) -- (230,80) ;
\draw    (370,80) -- (230,220) ;
\draw    (370,180) -- (230,220) ;
\draw  [dash pattern={on 0.84pt off 2.51pt}]  (280,220) -- (280,170) ;
\draw  [dash pattern={on 0.84pt off 2.51pt}]  (280,170) -- (230,170) ;
\draw  [dash pattern={on 0.84pt off 2.51pt}]  (300,200) -- (230,200) ;
\draw  [dash pattern={on 0.84pt off 2.51pt}]  (300,200) -- (300,220) ;
\draw [shift={(300,200)}, rotate = 0] [color=blueGrotto  ,draw opacity=1 ][fill=blueGrotto  ,fill opacity=1 ][line width=0.75]      (0, 0) circle [x radius= 4, y radius= 4]   ;
\draw [shift={(280,170)}, rotate = 0] [color=niceRed  ,draw opacity=1 ][fill=niceRed  ,fill opacity=1 ][line width=0.75]      (0, 0) circle [x radius= 4, y radius= 4]   ;

\draw (412,223.4) node [anchor=north west][inner sep=0.75pt]    {    \large$\alpha $};
\draw (225,58) node [anchor=south east][inner sep=0.75pt]    {    \large$\beta $};
\draw (211,80.4) node [anchor=north west][inner sep=0.75pt]    {    \large$1$};
\draw (312,223.4) node [anchor=north west][inner sep=0.75pt]    {    \large$\frac{1}{3}$};
\draw (378,63.4) node [anchor=north west][inner sep=0.75pt]    {    \large$\alpha =\beta $};
\draw (371,162.4) node [anchor=north west][inner sep=0.75pt]    {    \large$\alpha =2\beta $};
\draw (296,222.4) node [anchor=north west][inner sep=0.75pt]    {    \large$\frac{2}{7}$};
\draw (271,222.4) node [anchor=north west][inner sep=0.75pt]    {    \large$\frac{1}{4}$};
\draw (211,142.4) node [anchor=north west][inner sep=0.75pt]    {    \large$\frac{1}{4}$};
\draw (211,192.4) node [anchor=north west][inner sep=0.75pt]    {    \large $\frac{1}{7}$};

\end{tikzpicture}